\documentclass[journal,onecolumn,twoside]{IEEEtran}
\normalsize
\ifCLASSINFOpdf
\else
\fi
\usepackage{setspace}
\usepackage{color}
\usepackage{cite}
\usepackage{amsmath}
\usepackage{amsfonts}
\usepackage{amssymb}
\usepackage{amsthm}
\usepackage{tikz}
\usepackage{cases}
\usepackage{bm}
\usepackage{graphicx}
    \graphicspath{{../}}
    \DeclareGraphicsExtensions{.pdf}
\usepackage[caption=true,font=footnotesize]{subfig}
\usepackage{multirow}
\usepackage{makecell}
\usepackage{mathdots}
\usepackage{booktabs}
\usepackage{url}
\usepackage{bm}
\usepackage{xtab}
\usepackage{pifont}
\usepackage{array}
\usepackage{algorithm}
\usepackage{algorithmic}
\usepackage{enumerate}
\usepackage{extarrows}
\usepackage{tikz}
\usepackage{amsfonts}
\usepackage{color}
\usepackage{xcolor}
\usetikzlibrary{arrows}
\usepackage{caption}

\usepackage{lineno} 

\usepackage{mathrsfs}
\usepackage{tikz-cd}
\usepackage{mathtools}
\usepackage{makecell}
\usepackage{booktabs}
\usepackage{array}
\usepackage{enumitem}

\allowdisplaybreaks[4]
\usepackage[colorlinks,
            linkcolor=blue,
            anchorcolor=blue,
            citecolor=blue,
            urlcolor=black]{hyperref}
\theoremstyle{plain}

\newtheorem{theorem}{Theorem}[section]
\newtheorem{lemma}{Lemma}[section]
\newtheorem{proposition}{Proposition} [section]
\newtheorem{definition}{Definition} [section]
\newtheorem{corollary}{Corollary} [section]
\theoremstyle{definition}
\newtheorem{example}{Example}[section]

\newtheorem{remark}{Remark}[section]

\DeclareMathOperator{\Supp}{supp}

\DeclareMathOperator{\Span}{span}

\DeclareMathOperator{\Tr}{Tr}

\newcommand{\Ccal}{\mathcal{C}}

\newcommand{\Lcal}{\mathcal{L}}

\newcommand{\Pcal}{\mathcal{P}}

\newcommand{\Xcal}{\mathcal{X}}

\newcommand{\Cfrak}{\mathfrak{C}}
 
\newcommand{\Efrak}{\mathfrak{E}}

\newcommand{\Gal}{\mathrm{Gal}}

\newcommand{\Pbb}{\mathbb{P}}

\newcommand{\Zbb}{\mathbb{Z}}

\newcommand{\id}{\mathrm{id}}

\newcommand{\Review}{}

\newcommand{\F}{\mathbb{F}}

\newcommand{\Aut}{\mathrm{Aut}}

\newcommand{\Diff}{\mathrm{Diff}} 
\newcommand{\Con}{\mathrm{Con}}
\newcommand{\Div}{\mathrm{Div}}
\newcommand{\Princ}{\mathrm{Princ}}
\newcommand{\Cl}{\mathrm{Cl}}
\newcommand{\Char}{\mathrm{char}}
\newcommand{\GL}{\mathrm{GL}}
\newcommand{\oq}{{\overline{q}}}

\newcommand{\urep}{{I}}
\newcommand{\irep}{i}
\newcommand{\Frep}{{f}}
\newcommand{\Emph}{}

\DeclarePairedDelimiter\efloor{\lfloor}{\rfloor}
\DeclarePairedDelimiter\eceil{\lceil}{\rceil}
\DeclarePairedDelimiter\eparentheses{(}{)}
\DeclarePairedDelimiter\ebrackets{[}{]}
\DeclarePairedDelimiter\ebraces{\{}{\}}

\newcommand{\floor}{\efloor*}
\newcommand{\ceil}{\eceil*}
\newcommand{\parentheses}{\eparentheses*}
\newcommand{\brackets}{\ebrackets*}
\newcommand{\braces}{\ebraces*}

\makeatletter

\newcommand{\Rmnum}[1]{\expandafter\@slowromancap\romannumeral #1@}
\makeatother
 
\newcommand{\Rnum}[1]{\uppercase\expandafter{\romannumeral #1\relax}}
\newcommand{\rnum}[1]{\expandafter{\romannumeral #1\relax}}

\begin{document}

\title{New Constructions of Optimal $(r,\delta)$-LRCs via Algebraic Function Fields} 
\author{Yuan~Gao,~\IEEEmembership{}
        Haoming Shi,~\IEEEmembership{}
        Weijun~Fang~\IEEEmembership{}
\IEEEcompsocitemizethanks{\IEEEcompsocthanksitem Yuan Gao, Haoming Shi, and Weijun Fang are with State Key Laboratory of Cryptography and Digital Economy Security, Shandong University, Qingdao, 266237, China, Key Laboratory of Cryptologic Technology and Information Security, Ministry of Education, Shandong University, Qingdao, 266237, China and School of Cyber Science and Technology, Shandong University, Qingdao, 266237, China (emails: gaoyuan862023@163.com, 202421328@mail.sdu.edu.cn, fwj@sdu.edu.cn).
}
\thanks{The work is supported in part
by the National Key Research and Development Program of China under Grant Nos. 2022YFA1004900 and 2021YFA1001000, the National Natural Science Foundation of China under Grant No. 62571301, and the Shandong Provincial Natural Science Foundation
under Grant No. ZR2025QA05.  {\it (Corresponding Author: Weijun Fang)}.}
\thanks{Manuscript submitted }}

\maketitle
\begin{abstract} 
Constructing optimal $(r,\delta)$-LRCs that attain the Singleton-type bound is an active and important research direction, particularly due to their practical applications in distributed storage systems. In this paper, we focus on the construction of optimal $(r,\delta)$-LRCs with flexible minimum distances, especially for the case $\delta \geq 3$.
We first extend a general framework---originally proposed by Li \textit{et al.} (IEEE Trans. Inf. Theory, vol. 65, no. 1, 2019) and Ma and Xing (J. Comb. Theory Ser. A., vol. 193, 2023)---for constructing optimal $r$-LRCs via automorphism groups of elliptic function fields to the case of $(r,\delta)$-LRCs. This newly extended general framework relies on certain conditions concerning the group law of elliptic curves. By carefully selecting elliptic function fields suitable for this framework, we arrive at several families of explicit $q$-ary optimal $(r,3)$-LRCs and $(2,\delta)$-LRCs with lengths slightly less than $q + 2\sqrt{q}$.
Next, by employing automorphism groups of hyperelliptic function fields of genus $2$, we develop a framework for constructing optimal $(r,3)$-LRCs and obtain a family of explicit $q$-ary optimal $(4,3)$-LRCs with code lengths slightly below $q+4\sqrt{q}$. We then consider the construction of optimal $(r,\delta)$-LRCs via hyperelliptic function fields of arbitrary genus $g \geq 2$, yielding a class of explicit $q$-ary optimal $(g+1-g',g+1+g')$-LRCs for $0 \leq g' \leq g-1$ with lengths up to $q + 2g\sqrt{q}$. 
Finally, applying certain superelliptic curves derived from modified Norm-Trace curves, we construct two families of explicit optimal $(r,\delta)$-LRCs with even longer code lengths and more flexible parameters.
Notably, many of the newly constructed optimal $(r,\delta)$-LRCs attain the largest known lengths among existing constructions with flexible minimum distances.
\end{abstract}



\IEEEpeerreviewmaketitle

\section{Introduction}
\label{sec:intro}

To reduce the repair overhead of failed nodes in large-scale distributed storage systems, the concept of \Emph{locally repairable codes} (LRCs), also known as locally recoverable codes, was formally introduced in \cite{gopalan2012on} by Gopalan \textit{et al.}
Let $[n]:= \{1,2,\dots,n\}$. For a linear code $\Ccal$ of length $n$ over the finite field $\F_q$, a code symbol $c_i$ of $\Ccal$ has \Emph{locality $r$} 
if there exists a subset $R_i\subseteq [n]$ such that $i\in R_i,|R_i|\leq r+1$
and $c_i$ is a linear combination of $\{c_j\}_{j\in R_i\backslash\{i\}}$ over $\F_q$. Here, $R_i$ is called a \Emph{local repair group} of the $i$-th symbol $c_i$.
A linear code $\Ccal$ is called an $r$-\Emph{locally repairable code} ($r$-LRC) if each code symbol of $\Ccal$ has locality $r$. However, when multiple node failures occur in a distributed storage
system, the $r$-LRCs can not recover failed nodes efficiently. To address this problem, Prakash \textit{et al.} \cite{prakash2012optimal} generalized the concept of $r$-LRCs to
$(r,\delta)$-LRCs which can tolerate any $\delta-1$ erasures ($\delta\geq 2$). A code symbol $c_i$ of $\Ccal$ has \Emph{$(r,\delta)$-locality} if there exists a subset $R_i\subseteq [n]$ containing $i$ such that $|R_i|\leq r+\delta-1$
and $d(\Ccal|_{R_i})\geq\delta$, where $\Ccal|_{R_i}$ denotes the punctured code on the set $[n]\backslash R_i$. Similarly, $R_i$ is called a local repair group of the $i$-th symbol $c_i$.
A linear code $\Ccal$ is called an \Emph{$(r,\delta)$-locally repairable code} ($(r,\delta)$-LRC) if each code symbol of $\Ccal$ has $(r,\delta)$-locality. When $\delta=2$, $(r,\delta)$-LRCs reduce to $r$-LRCs. Due to their interesting
algebraic structures and practical applications in distributed storage systems, $(r,\delta)$-LRCs have drawn significant interest in recent years. In the following, we review some known results on $(r,\delta)$-LRCs. 

\subsection{Some Known Results of $(r,\delta)$-LRCs}\label{sec:1.1}
In \cite{gopalan2012on}, the \Emph{Singleton-type bound} for $r$-LRCs with parameters $[n,k,d]_q$, analogous to the classical Singleton bound for general codes, was proposed. This bound was later generalized in \cite{prakash2012optimal} to the case of \((r,\delta)\)-LRCs with parameters $[n,k,d]_q$, yielding 
\begin{align}\label{eq:SLboundfor_r_delta}
	d\leq n-k+1-(\left\lceil k/r\right\rceil -1)(\delta -1), 
\end{align}
which is also called the Singleton-type bound. When $\delta=2$, it reduces to the Singleton-type bound for $r$-LRCs.
If an $(r,\delta)$-LRC achieves the Singleton-type bound \eqref{eq:SLboundfor_r_delta} with equality, then it is called a \Emph{Singleton-optimal} $(r,\delta)$-LRC, which we refer to simply as an \Emph{optimal} $(r,\delta)$-LRC in this paper. In particular, an optimal $r$-LRC attaining the Singleton-type bound \eqref{eq:SLboundfor_r_delta} (fixing $\delta=2$) is called an \Emph{optimal} $r$-LRC. 
Constructing optimal $(r,\delta)$-LRCs with large code lengths over fixed finite fields is of practical importance, as it enables a reduction in the required field size and thereby lowers the overall computational complexity.
In what follows, we review some existing constructions of optimal $(r,\delta)$-LRCs, along with relevant results on upper bounds for their code lengths. These constructions can be broadly classified into two categories based on their parameter characteristics, particularly their minimum distances.
\begin{itemize} 
    \item  
   The first category of constructions of optimal $(r,\delta)$-LRCs \Review{consists of those with flexible minimum distances}
   \footnote{\Review{We say that a family of optimal $(r,\delta)$-LRCs has flexible minimum distances if, for every $\tau\in [0,1]$, there exists a sequence of codes from the family with lengths tending to infinity such that $d/n\to \tau$.
}}\Review{, where the minimum distances $d$} can be either small or proportional to the code length $n$, and can be adjusted as needed. These constructions are typically obtained via evaluation-based methods using tools such as polynomials \cite{tamo2014family}, \cite{gao2024new}, algebraic curves and surfaces \cite{barg2017locally_curve, barg2017locally_curveandsurface, jin2020construction, li2019optimal,  ma2023group, salgado2021locally}. Additionally, some constructions in this category are derived from cyclic and constacyclic codes, such as \cite{tamo2015cyclic,tamo2016cyclic,chen2018constructions, chen2019constructions, qiu2021new}. 
   As for the upper bound on the length of this category of constructions, Guruswami \textit{et al.} \cite[Theorem 13 and Corollary 14]{Guruswami2019how} established an upper bound on the code length $n$ of $q$-ary optimal $r$-LRCs with minimum distance $d=\Theta(n)$ and constant $r$, yielding $n\leq O(q)$. This upper bound can be generalized to the case of $(r,\delta)$-LRCs, which similarly yields $n\leq O(q)$ for optimal $(r,\delta)$-LRCs with minimum distance $d=\Theta(n)$, $\frac{d}{n}\leq \frac{2}{3}$, and constant $r,\delta$ (see Theorem \ref{thm:bound_for_length_of_optimal_(r,delta)} in Appendix \ref{appsec:1}). 
\Review{When $r$ and $\delta$ are constants,} this upper bound naturally applies to optimal \((r,\delta)\)-LRCs with flexible minimum distances.

   Below, we briefly review the above-referenced works. In 2014, Tamo and Barg \cite{tamo2014family} made a breakthrough in constructing optimal $(r,\delta)$-LRCs by using carefully chosen polynomials. They proposed the well-known RS-like optimal $(r,\delta)$-LRCs, whose lengths are at most $q$. This famous family of codes is referred to as the \Emph{Tamo–Barg codes}. As a side note, Gao and Yang \cite{gao2024new} later improved upon this result in 2024 by applying an extension technique, obtaining optimal $(r,\delta)$-LRCs with lengths up to $q+\delta$. 
Returning to the main line of development, in 2017, Barg \textit{et al.} \cite{barg2017locally_curve,barg2017locally_curveandsurface} extended the Tamo-Barg codes \cite{tamo2014family} by utilizing covering maps and quotient maps of algebraic curves, and selecting appropriate evaluation function spaces via the Riemann–Roch theorem. They presented several asymptotically good $r$-LRCs using high-genus curves and towers of function fields. Additionally, some optimal $r$-LRCs based on algebraic surfaces were also proposed. Although these optimal $r$-LRCs have small minimum distances $d\leq 3$, they offer valuable theoretical insights.  
In \cite{jin2020construction}, by employing the automorphism groups of the rational function
fields, Jin \textit{et al.} generalized the Tamo-Barg codes \cite{tamo2014family}, constructing optimal $r$-LRCs with code length up to $q+1$ and more flexible locality $r$. 
In \cite{li2019optimal}, by leveraging the rich algebraic structures of elliptic curves, Li \textit{et al.} constructed optimal $r$-LRCs with code length reaching or slightly below $q+2\sqrt{q}$ for $r=2,3,5,7,11,23$. 
They also provided many maximal elliptic curves with rich automorphism groups. 
Subsequently, Ma and Xing \cite{ma2023group} extended the results in \cite{li2019optimal} by incorporating the translation automorphism group of elliptic function fields into the construction, thereby obtaining optimal $r$-LRCs with a wider range of locality $r$. They also provided a clear exposition of the group structure of the automorphism groups of elliptic function fields. In 2021, Salgado \textit{et al.} \cite{salgado2021locally} proposed a family of optimal $3$-LRCs with length $4q$ based on algebraic surfaces, which has the largest known code length among optimal $r$-LRCs with flexible minimum distances. Very recently, using automorphism groups of hyperelliptic curves of genus $2$, Huang and Zhao \cite{huang2025optimal} proposed several $r$-LRCs with lengths approaching $q+4\sqrt{q}$ and $(r+1)\mid 240$, which are either optimal or almost optimal\footnote{An $(r,\delta)$-LRC with parameters $[n,k,d]_q$ satisfying $d=n-k-(\left\lceil {k}/{r}\right\rceil-1)(\delta-1)$ is referred to as an \Emph{almost optimal} $(r,\delta)$-LRC.}.

There are also several notable results based on cyclic codes and constacyclic codes. 
Motivated by the construction of Tamo-Barg codes, Tamo \textit{et al.} \cite{tamo2015cyclic,tamo2016cyclic} constructed a family of optimal cyclic
$r$-LRCs with length $q-1$.
Later, Chen \textit{et al.} \cite{chen2018constructions, chen2019constructions} generalized these results by employing both cyclic and constacyclic codes, and constructed several families of optimal $(r,\delta)$-LRCs with $(r+\delta-1)\mid n$ and $n$ dividing either $q-1$ or $q+1$. These constructions were further extended by Qiu \textit{et al.} \cite{qiu2021new}, who unified and generalized the constructions in \cite{chen2018constructions, chen2019constructions}, and efficiently produced optimal cyclic $(r,\delta)$-LRCs with flexible parameters, including cases where $(r + \delta - 1) \nmid n$.
    \item 
    The second category of constructions of optimal $(r,\delta)$-LRCs consists of \Review{those} with \Review{prescribed} minimum distance $d$. \Review{One first prescribes the desired minimum distance $d$ and locality parameters $r,\delta$, and then obtains a family of optimal $(r,\delta)$-LRCs with minimum distance $d$ and large code lengths.} In this case, the minimum distance $d$ is less flexible, and the relative minimum distance $d/n$ vanishes as the code length $n$ tends to infinity.
 Notably, some constructions in this category achieve super-linear code lengths with respect to the field size $q$. Moreover, certain constructions even attain unbounded code length when $d\leq 2\delta$. 
 As shown in \cite[Theorem 10]{Guruswami2019how} and \cite[Theorem 2]{cai2020linear}, the length of a $q$-ary optimal $(r,\delta)$-LRC with minimum distance $d> 2\delta$
is upper bounded by $O(q^{C_{d,r,\delta}})$, where $C_{d,r,\delta}$ is a constant depending only on $d$, $r$ and $\delta$.
Constructions in this category are usually derived using parity-check matrix approaches combined with combinatorial tools, as in \cite{Guruswami2019how,xing2022construction,jin2019explicit, cai2020linear, chen2021improved, kong2021new, chen2022some, fang2024bound, tao2025some}. 
In addition, some constructions are obtained using cyclic codes and constacyclic codes, such as those in \cite{luo2019optimal, fang2023singleton, qiu2024newlower}.
\end{itemize} 
Apart from these two main categories, there also exist several interesting constructions that do not fall neatly into either one, e.g., \cite{hao2020bound, luo2023three, galindo2024optimal}. We do not go into details.  
 
To facilitate comparison, we summarize the aforementioned first category of constructions of
optimal $(r,\delta)$-LRCs, namely, those with flexible minimum distances, in Table~\ref{tab:1}. Row 15 is marked with a superscript $^{*}$ to indicate that the corresponding construction is either optimal or almost optimal. We remark that the constructions in \cite{jin2020construction} by Jin \textit{et al.} were originally presented as $r$-LRCs, but they can be naturally generalized to the case of $(r,\delta)$-LRCs, as indicated in Rows 1 and 2 of Table~\ref{tab:1}. Some of our new constructions are also listed in Table~\ref{tab:1} and will be discussed in detail in the next subsection. 
\renewcommand{\arraystretch}{1.5}
\setlength\tabcolsep{3pt}
\begin{table}[tbhp]  
\small
\centering
\caption{Known Constructions of $q$-ary Optimal $(r,\delta)$-LRCs with Flexible Minimum Distances and Lengths $\geq q$}\label{tab:1} 
\begin{tabular}{|c|c|c|c|}\hline
 NO. &\textbf{Length $n$}&\textbf{Locality $(r,\delta)$ and Conditions}&\textbf{References} \\ \hline
1    &$q$     & \makecell{$(r,\delta)$, with $(r+\delta-1)=p^l$, \\\text{where} $1\leq l\leq \log_p(q)$, $p=\Char(\F_q)$}&   \cite{tamo2014family}, \cite{jin2020construction}     \\ \hline 
     
2     &$q+1$     &   $(r,\delta)$, with $(r+\delta-1)\mid (q+1)$    &       \makecell{\cite{jin2020construction, chen2018constructions}, \\ \cite{chen2019constructions, qiu2021new}}     \\ \hline

3 & \makecell{$q+\delta$ \\ \Review{( $q-1+\delta$, respectively})}  &  \makecell{$(r,\delta)$, with $p=\Char(\F_q)\mid\delta,(r+\delta-1)=p^l$, $1<l\leq \log_p(q)$\\ \Review{(\text{with } $(r+\delta-1)\mid (q-1)$, \text{respectively}})} & \cite{gao2024new} \\ \hline 

4&  $3\big\lfloor\frac{q+2\sqrt{q}}{3}\big\rfloor$  & \makecell{$(r=2,\delta=2)$, with $q=p^{2s}$ \\ for $p=3\text{ or }p\equiv 2\pmod{3}$} &  \cite[Theorem 1]{li2019optimal}  \\ \hline

5&  $(r+1)\big\lfloor\frac{q+2\sqrt{q}-r-2}{r+1}\big\rfloor$  & $(r,\delta=2)$, with $r=3,5,7,11,23$ &  \cite[Theorem 2]{li2019optimal}  \\ \hline
 
6& $4q$ &$(r=3,\delta=2)$, with $4\mid (q-1)$ &\cite{salgado2021locally}  \\ \hline
 
7& $2h(\big\lceil\frac{N(E)}{2h}\big\rceil-2)$ & $(r=2h-1,\delta=2)$, with $h\mid N(E)=|\Pbb_E^1|\leq q+2\sqrt{q}+1$ & \cite[Proposition 4.6]{ma2023group} \\ \hline

8& $q+2\sqrt{q}+1-3h$ & \makecell{$(r=2h-2,\delta=3)$ and $(r=2,\delta=2h-1)$, \\with $q=2^{2s}$, $h\mid (q+2\sqrt{q}+1)$}& Corollary~\ref{cor:Cons_via_EFF_1.1} \\ \hline

9& $q+2\sqrt{q}-3h$ &\makecell{$(r=2h-2,\delta=3)$ and $(r=2,\delta=2h-1)$, with\\ $q=p^{2s}$, $p\geq 3$, $h\mid (q+2\sqrt{q})$}& Corollary~\ref{cor:Cons_via_EFF_1.2} \\ \hline

10&  $ah\big\lceil\frac{q+2\sqrt{q}+1-2h-ah}{ah}\big\rceil$ & \makecell{$(r=ah-1,\delta=2)$, with $q=p^{2s}$,  \\ $a\mid 24, h=h_0^2, h_0\mid (\sqrt{q}+1)$}&\cite[Theorem 4.8]{ma2023group}  \\ \hline

11& $3h\big\lceil\frac{q+2\sqrt{q}+1-3h}{3h}\big\rceil$ &\makecell{$(r=3h-2,\delta=3)$ and $(r=2,\delta=3h-1)$, with \\$q=2^{2s},h=h_0^2,h_0\mid (\sqrt{q}+1)$}& Corollary~\ref{cor:Cons_via_EFF_2.1} {\rm(i)} \\ \hline 

12& $6h\big\lceil\frac{q+2\sqrt{q}+1-8h}{6h}\big\rceil$ &\makecell{$(r=3h-2,\delta=3)$ and $(r=2,\delta=3h-1)$, with $q=p^{2s}$ for\\ an odd prime $p\equiv 2 \pmod{3}$ or $p=3$, $h=h_0^2,h_0\mid (\sqrt{q}+1)$} & \makecell{Corollary~\ref{cor:Cons_via_EFF_2.1} {\rm(ii)}\\\text{and }{\rm (iii)}} \\ \hline  
 
13& $q+2\sqrt{q}-8$&$(r=8,\delta=2)$, with $q=4^{2s+1}$& \cite[Theorem 4.9]{ma2023group} \\ \hline

14& $q+2\sqrt{q}-8$&\makecell{$(r=7,\delta=3)$ and $(r=2,\delta=8)$, with $q=4^{2s+1}$}& Corollary~\ref{cor:Cons_via_EFF_2.2_(7,3)} \\ \hline

15$^*$& $\leq q+4\sqrt{q}$ &$(r,\delta=2)$, with $(r+1)\mid 240$   &\cite{huang2025optimal}  \\ \hline  
 
16& $12\big\lceil\frac{q+4\sqrt{q}-1}{12}\big\rceil-30$ &$(r=4,\delta=3)$, with $q=5^{2s}$, $2\nmid s$ &  Corollary~\ref{cor:Cons_via_HEFF_(4,3)_p=5} {\rm (i)}\\ \hline
 
17& $12\big\lceil\frac{q+4\sqrt{q}-1}{12}\big\rceil-30$ &\makecell{$(r=4,\delta=3)$, with $q=\oq^{2s}$, $\oq\neq 5$,\\ $\oq\equiv 5,15,21, \text{ or } 23\pmod{24},2\nmid s$}&  Corollary~\ref{cor:Cons_via_HEFF_(4,3)_p=5} {\rm (ii)}\\ \hline
 
18& $q+2g\sqrt{q}$ & \makecell{$(r=g+1-g',\delta=g+1+g')$, with $2g+1\geq 5$ being\\  a prime power, $0\leq g'\leq g-1, q=(2g+1)^{2s}$} &  Theorem~\ref{thm:Cons_via_HEFF_(g+1-g',g+1+g')_p=2g+1_and_p_neq_2g+1} {\rm(i)}\\ \hline 

19& $(2g+1)\big\lfloor \frac{q+2g\sqrt{q}}{2g+1}\big\rfloor$ & \makecell{$(r=g+1-g',\delta=g+1+g')$, with $g\geq 2$,\\ $0\leq g'\leq g-1$, $q=\oq^{2s},\oq\equiv-1\pmod{2g+1},2\nmid \oq$} &  Theorem~\ref{thm:Cons_via_HEFF_(g+1-g',g+1+g')_p=2g+1_and_p_neq_2g+1} {\rm(ii)}\\ \hline 

20& $\gcd(b,\frac{\oq^c-1}{\oq-1}){\cdot} \frac{(q-1)q}{b\oq^c}{+}\frac{q}{\oq^c}$ & \makecell{$(r=\big\lfloor\frac{N-1}{M}\big\rfloor+1-b',\delta=N+1-r)$, with $q=\oq^{s},b\mid \frac{\oq^s-1}{\oq-1}$, \\  $c\mid s$,  $1<M=\frac{\oq^s-1}{b(\oq-1)}<N=\oq^{s-c},0\leq b'\leq \big\lfloor\frac{N-1}{M}\big\rfloor-1$} & Theorem \ref{thm:cons_via_SEFF_1} {\rm (i)}\\ \hline 

21& $\gcd(b,\frac{\oq^c-1}{\oq-1}){\cdot} \frac{(q-1)q}{b\oq^c}$ & \makecell{$(r=\big\lfloor\frac{M-1}{N}\big\rfloor+1-b',\delta=M+1-r)$, with $q=\oq^{s},b\mid \frac{\oq^s-1}{\oq-1}$, \\  $c\mid s,M=\frac{\oq^s-1}{b(\oq-1)}>N=\oq^{s-c}>1,0\leq b'\leq \big\lfloor\frac{M-1}{N}\big\rfloor-1$} & Theorem \ref{thm:cons_via_SEFF_1} {\rm (ii)} \\ \hline 

22& $\frac{1}{b\oq}q^2+\frac{b-1}{b\oq}q$ & \makecell{$(r=\big\lfloor{b\frac{(\oq-1)(\oq^{s-1}-1)}{\oq^s-1}}\big\rfloor+1-b',\delta=\oq^{s-1}+1-r)$, with $q=\oq^{s}$,\\ 
$b\mid \frac{\oq^s-1}{\oq-1},1< b< \frac{\oq^s-1}{\oq-1},0\leq b'\leq \big\lfloor b\frac{(\oq-1)(\oq^{s-1}-1)}{\oq^s-1}\big \rfloor-1$} & Example \ref{exm:cons_via_SEFF_1} {\rm (i)} \\ \hline 

23& $\frac{q(q-1)}{\oq^c}$ & \makecell{$(r=\big\lfloor\frac{(\oq^s-\oq)}{\oq^{s-c}(\oq-1)}\big\rfloor+1-b',\delta=\frac{\oq^s-1}{\oq-1}+1-r)$, with $q=\oq^{s}$,\\ $c\mid s, c<s, 0\leq b'\leq \big\lfloor\frac{(\oq^s-\oq)}{\oq^{s-c}(\oq-1)}\big\rfloor-1$} & Example \ref{exm:cons_via_SEFF_1} {\rm (ii)} \\ \hline 

24&  $q+(\frac{\oq+1}{b}-1)(\oq-1)\sqrt{q}$ &\makecell{$(r=b-b',\delta=\oq+1-r)$, with $q=\oq^{2s}$, \\ $2\nmid s, 1<b<\oq+1, b\mid (\oq+1)$, $0\leq b'\leq b-2$}
&Theorem~\ref{thm:cons_via_SEFF_2} {\rm (i)}\\ \hline

25& $q+\oq(\oq-1)\sqrt{q}-\oq$ &$(r=2,\delta=\oq)$, with $q=\oq^{2s}$, $2\nmid s$
&Theorem~\ref{thm:cons_via_SEFF_2} {\rm (ii)}\\ \hline  
   
\end{tabular} 
\end{table}	
\renewcommand{\arraystretch}{1.0}
\subsection{Our Motivations and Contributions}
As shown in Table~\ref{tab:1}, there have been several constructions of optimal $r$-LRCs (i.e., $(r,\delta=2)$-LRCs) with flexible minimum distances and lengths exceeding $q+1$ (see Rows 4-7, 10, 13). 
\Review{It is worth noting that, when $\delta\geq 3$, there is only one previously known work \cite{gao2024new}, to the best of our knowledge, that provides optimal $(r,\delta)$-LRCs with flexible minimum distances and lengths exceeding $q+1$; see Row~3 of Table~\ref{tab:1}.}
However, this construction has a drawback: its code length depends on $\delta$. When $\delta$ is fixed, the resulting code length exceeds $q+1$ by only a constant. These facts motivate us to consider the construction of long optimal $(r,\delta)$-LRCs with flexible minimum distances, especially for $\delta\geq 3$. Naturally, we begin to consider whether the elliptic function field-based constructions of optimal $r$-LRCs proposed by Li \textit{et al.} \cite{li2019optimal} and by Ma and Xing \cite{ma2023group} can be extended to the general optimal $(r,\delta)$-LRCs. We find that, unlike the constructions based on polynomials and rational function fields in \cite{tamo2014family} and \cite{jin2020construction}, the generalization of elliptic function field-based constructions in \cite{li2019optimal, ma2023group} from $r$-LRCs to $(r,\delta)$-LRCs is not straightforward. For further details, see Section~\ref{sec:3.1}, especially Remark~\ref{rem:equivalentconditionforEFF_1}~{\rm(i)}. Therefore, we turn our attention to the problem of constructing long optimal $(r,\delta)$-LRCs with $\delta \geq 3$ and flexible minimum distances, using elliptic function fields or more generally, algebraic function fields of higher genus. Our main contributions are organized into the following three parts.
\begin{itemize}
\item  
By utilizing the abelian group structure of elliptic curves, we generalize the framework for constructing optimal $r$-LRCs via automorphism groups of elliptic function fields in \cite{li2019optimal,ma2023group} to the case of $(r,\delta)$-LRCs in Propositions~\ref{prop:equivalentconditionforEFF_1} and~\ref{prop:construction_of_(r,delta)from_EFF_1}. 
Later in Theorems~\ref{thm:estimation_of_length_EFF_and_cons_1} and~\ref{thm:estimation_of_length_EFF_and_cons_2}, we propose two distinct sufficient conditions for elliptic function fields and subgroups of their automorphism groups, under which we can obtain constructions of optimal $(r,3)$-LRCs and $(2,\delta)$-LRCs by the generalized framework described above. By selecting suitable explicit elliptic function fields and their automorphism subgroups, we obtain several classes of explicit optimal $(r,3)$-LRCs and $(2,\delta)$-LRCs with lengths slightly less than  $q+2\sqrt{q}$. Their parameters are outlined in Rows 8, 9, 11, 12, 14 of Table~\ref{tab:1}.  

\item
Inspired by the constructions of either optimal or almost optimal $r$-LRCs via automorphism groups of hyperelliptic function fields of genus $2$ proposed by Huang and Zhao \cite{huang2025optimal}, we develop a general framework for constructing optimal $(r,3)$-LRCs via automorphism groups of such hyperelliptic function fields, as presented in Propositions~\ref{prop:sufficientconditionfor_(r,3)_viaHEFF_1} and~\ref{prop:construction_of_(r,3)from_HEFF_1}. By applying this framework to specific hyperelliptic function fields, we arrive at a family of explicit optimal $(4,3)$-LRCs with length slightly below $q+4\sqrt{q}$. We then further consider the construction of optimal $(r,\delta)$-LRCs via hyperelliptic function fields of genus $g\geq 2$, obtaining optimal $(g+1-g',g+1+g')$-LRCs ($0\leq g'\leq g-1)$ with length $q+2g\sqrt{q}$ in Theorem~\ref{thm:Cons_via_HEFF_(g+1-g',g+1+g')_p=2g+1_and_p_neq_2g+1}. Their parameters are outlined in Rows 16--19 of Table~\ref{tab:1}. 

\item We propose a framework for constructing optimal $(r,\delta)$-LRCs via superelliptic curves in Proposition~\ref{prop:construction_of_(r,delta)from_SEFF}, and obtain several classes of explicit constructions based on it. Their parameters are partially listed in Rows 20--25 of Table~\ref{tab:1}. Specifically, as shown in Rows 22 and 23 of Table~\ref{tab:1}, over $\F_q=\F_{\oq^s}$ with $s\geq 2$, Theorem~\ref{thm:cons_via_SEFF_1} produces $q$-ary optimal $(r,\delta)$-LRCs with lengths up to $\frac{1}{b\oq}q^2+\frac{b-1}{b\oq}q$ and $\frac{q(q-1)}{\oq^c}$, respectively. 
In Theorem~\ref{thm:cons_via_SEFF_2} {\rm (i)}, we generalize Theorem~\ref{thm:Cons_via_HEFF_(g+1-g',g+1+g')_p=2g+1_and_p_neq_2g+1} {\rm (i)}, yielding optimal $(r,\delta)$-LRCs with longer code lengths for smaller value of $r$ (with $(r+\delta-1)$ fixed); moreover, it can be carried out over fields of even characteristic (see Remark~\ref{rem:cons_via_SEFF_2_2}).
\end{itemize}
It is evident from Table~\ref{tab:1} that many of our new optimal $(r,\delta)$-LRCs have the longest lengths among existing constructions of optimal $(r,\delta)$-LRCs with flexible minimum distances.
 
\subsection{Organization of This Paper}
The rest of the paper is organized as follows. 
In Section~\ref{sec:2}, we review some preliminaries for
this paper, including algebraic function fields, algebraic geometry codes, extension theory of algebraic function fields, elliptic function fields and hyperelliptic function fields, along with their automorphism groups.
In Section~\ref{sec:3}, we present a general framework for constructing optimal $(r,\delta)$-LRCs via automorphism groups of elliptic function fields. 
Based on this framework, we construct two distinct classes of optimal $(r,3)$-LRCs and $(2,\delta)$-LRCs with lengths slightly below $q+2\sqrt{q}$.
In Section~\ref{sec:4}, using the automorphism subgroups of hyperelliptic function fields of genus $2$, we obtain optimal $(4,3)$-LRCs with lengths slightly below $q+4\sqrt{q}$. 
Optimal $(g+1-g',g+1+g')$-LRCs ($0\leq g'\leq g-1$) via hyperelliptic curves of genus $g\geq 2$ are also presented.
In Section~\ref{sec:5}, we introduce a general framework for constructing optimal $(r,\delta)$-LRCs via superelliptic curves. Based on it, we present two classes of explicit constructions with large code lengths.
Section~\ref{sec:6} concludes the paper.  
\section{Preliminaries}
\label{sec:2}
In this section, we present some preliminaries on algebraic function fields, algebraic geometry codes, extension theory of function fields, elliptic function fields and hyperelliptic function fields, as well as their automorphism groups.
For omitted details, the reader is referred to \cite{niederreiter2001rational, stichtenoth2009algebraic, silverman2009arithmetic}, \cite{li2019optimal}, \cite{ma2023group}, \cite{oskar1887onbinary,cardona1999oncurves,cardona2003onthenumber}, \cite{huang2025optimal}. 
\subsection{Algebraic Function Fields and Algebraic Geometry Codes}\label{sec:2.1}
Let $E/\F_q$ be a \Emph{function field} of genus $g(E)$ with the full constant field $\F_q$.
Let $\Pbb_E$ denote the set of all \Emph{places} of $E$, and let $\Pbb_{E}^1$ denote the set of all \Emph{rational places} of $E$.
The free abelian group generated by $\Pbb_E$ is called the \Emph{divisor group} of $E/\F_q$ and is denoted by $\Div(E)$. 
For $w\in E^*=E\backslash\{0\}$, its \Emph{principal divisor} is defined by $$(w):=\sum_{P\in \Pbb_E} v_P(w)P \in \Div(E),$$ where $v_P$ is the \Emph{normalized discrete valuation} with respect to the place $P$. 
For $D\in \Div(E)$, the \Emph{Riemann-Roch space} $$\Lcal(D):=\{w\in E^*:\; (w)\ge -D\}\cup \{0\}$$ is a finite-dimensional vector space over $\F_q$. We denote its dimension by $\ell(D):=\dim_{\F_q}\Lcal(D)$, which is at least $\deg(D)+1-g(E)$ by Riemann's theorem (see \cite[Theorem 1.4.17]{stichtenoth2009algebraic}). If $\deg(D)\geq 2g(E)-1$, then it holds 
\begin{align}\label{eq:ell(D)_of_deg(D)_geq_2g-1}
    \ell(D)=\deg(D)+1-g(E)
\end{align} 
by the Riemann-Roch theorem (see \cite[Theorem 1.5.15]{stichtenoth2009algebraic}).
When dealing with multiple function fields, we use superscripts and subscripts to indicate the underlying function field of the principal divisors and the Riemann-Roch spaces, respectively. For example, we write $(w)^E$ instead of $(w)$, and write $\Lcal_{E}(D)$ instead of $\Lcal(D)$. 

Let $\Pcal=\{P_1,\dots,P_n\}$ be a set of $n$ distinct rational places of $E$, which will be used for evaluation. For a divisor $D$ of $E$ with $0\leq \deg(D)<n$ and $\Supp(D)\cap\Pcal=\varnothing$, the algebraic geometry code associated with $\Pcal$ and $D$ is defined to be
\begin{align}\label{eq:10}
\Ccal(\Pcal,D):=\{(\phi(P_1),\dots,\phi(P_n)): \; \phi\in\Lcal_E(D)\}. 
\end{align}
Then $\Ccal(\Pcal,D)$ is a linear code with dimension $\ell(D)$ and minimum distance at least $n-\deg(D)$.
For any subspace $V$ of $\Lcal_E(D)$, we define a (linear) subcode of $\Ccal(\Pcal,D)$ by  
\begin{align}\label{eq:11}
\Ccal(\Pcal,V):=\{(\phi(P_1),\dots,\phi(P_n)):\; \phi\in  V\}.
\end{align}
Consequently, $\Ccal(\Pcal,V)$ is an $[n,k,d]_q$-linear code with dimension $k=\dim_{\F_q}(V)$, and its minimum distance $d$ remains at least $n-\deg(D)$.
  
\subsection{Extension Theory of Function Fields}

Let $E/\F_q$ be a function field with the full constant field $\F_q$ and let $F$ be a subfield of $E$ with the same full constant field $\F_q$ such that $E/F$ is a finite separable extension. For any place $P$ of $E$ and place $Q$ of $F$ such that $P$ lies over $Q$, we use $e(P|Q),f(P|Q)$, and $d(P|Q)$ to denote the \Emph{ramification index}, \Emph{relative degree}, and \Emph{different exponent} of $P$ over $Q$, respectively. By Dedekind's different theorem (see \cite[Theorem 3.5.1]{stichtenoth2009algebraic}), it holds
\begin{align}\label{eq:Dedekinddifferentthm}
    d(P|Q)\geq e(P|Q)-1.
\end{align}  
For a place $Q$ of $F$, its \Emph{conorm} (with respect to $E/F$) is defined to be $$\Con_{E/F}(Q):=\sum_{P|Q}e(P|Q)P\in \Div(E),$$
where the sum runs over all places $P\in \Pbb_E$ lying over $Q$.
The \Emph{different divisor} of $E/F$ is defined to be 
$
\Diff(E/F):=\sum_{Q\in\Pbb_F}\sum_{P|Q}d(P|Q)P\in \Div(E).
$ 
Let $g(E)$ and $g(F)$ denote the genus of $E$ and $F$, respectively.  
Then the Hurwitz genus formula (see \cite[Theorem 3.4.13]{stichtenoth2009algebraic}) yields
\begin{align}\label{eq:kurwitzformulacoro_1}
2g(E)-2=(2g(F)-2)[E:F]+\deg \Diff(E/F).
\end{align}  
So far, we have assumed that $E/F$ is a finite separable extension and recalled some known results. We now consider a more specific setting of $F$ to facilitate the later constructions ($F=E^G$, see below). Let $G$ be a finite subgroup of $\Aut(E/\F_q):=\{\sigma: \;  \sigma \text{ is an } \F_q\text{-automorphism of } E\}$. The subfield of elements of $E$ fixed by $G$ is defined by
$$
E^{G}:=\{u\in E: \sigma(u)=u \text{ for all } \sigma\in G\}.
$$ 
From the Galois theory, $E/E^{G}$ is a Galois extension with $\Gal(E/E^{G})=G.$ 
By \cite[Lemma 3.5.2]{stichtenoth2009algebraic}, for any automorphism $\sigma\in{\Gal}(E/E^G)=G$ and any place $P\in\Pbb_E$, $\sigma(P):=\{\sigma(u):u\in P\}$ is still a place of $E$. 
Moreover, if $P$ lies over $Q\in \Pbb_{E^G}$, then $\sigma(P)$ also lies over $Q$. By \cite[Theorem 3.7.1 and Corollary 3.7.2]{stichtenoth2009algebraic}, which characterize Galois extensions of function fields, the following result holds. 
\begin{lemma}[{\cite[Theorem 3.7.1 and Corollary 3.7.2]{stichtenoth2009algebraic}}]\label{lem:thepropertyofGalExtension_1}
   Maintaining the above setting. Let $Q$ be a place of $E^G$, and let $P_1,P_2,\dots,P_n$ be all the distinct places of $E$ lying over $Q$. Then the following statements hold.
   \begin{itemize}
  \item[{\rm (i)}] The Galois group $\Gal(E/E^G)=G$ acts transitively on the set $\{P_1,\dots,P_{n}\}$.
   \item[{\rm (ii)}] $e(P_1|Q)=\dots=e(P_n|Q)$, $f(P_1|Q)=\dots=f(P_n|Q)$, and $d(P_1|Q)=\dots=d(P_n|Q)$.
   \item[{\rm (iii)}] $n\cdot e(P_i|Q)f(P_i|Q)=[E:E^G]=|G|$ for any $1\leq i\leq n$.
   \end{itemize}
\end{lemma}    
\begin{remark}\label{rem:thepropertyofGalExtension_2} Lemma~\ref{lem:thepropertyofGalExtension_1} implies the following two useful facts. 
\begin{itemize}
\item[{\rm (i)}] For any $P\in \Pbb_{E}$, $P\cap E^G$ splits completely in $E/E^G$ if and only if the places $\sigma(P)$, for all $\sigma\in G$, are pairwise distinct.
\item[{\rm (ii)}] For any rational place $P$ of $E$, the rational place $P\cap E^G$ splits completely in $E/E^G$ if and only if $e(P|P\cap E^G)=1$.
 \end{itemize}
\end{remark} 
 In this paper, we always hope that the subfield $E^G$ of $E$ can be determined to be a rational function field over $\F_q$. Thus, the following necessary and sufficient condition will be useful. 
\begin{lemma} \label{lem:is_rational_FF}
Maintaining the above setting. $E^G$ is a rational function field if and only if $\deg\Diff(E/E^G)>2g(E)-2$.
\end{lemma}
\begin{proof}
We have $2g(E)-2=(2g(E^G)-2)[E:E^G]+\deg\Diff(E/E^G)$ by the Hurwitz genus formula (see \eqref{eq:kurwitzformulacoro_1}). Hence, $\deg\Diff(E/E^G)>2g(E)-2$ if and only if $g(E^G)=0$. This is equivalent to $E^G$ being a rational function field by \cite[Proposition 1.6.3, Eq. (5.3) and  Corollary 5.1.11]{stichtenoth2009algebraic}.  
\end{proof}   
\subsection{Elliptic Curves and Elliptic Function Fields}
\label{sec:2.3}
Throughout this paper, a curve is by default referred to as a projective, smooth, and absolutely irreducible algebraic curve.
In particular, an \Emph{elliptic curve} $\Efrak$ over $\F_q$ is defined by a nonsingular Weierstrass equation
\begin{align}\label{eq:Weierstrass_equation}
y^2+a_1xy+a_3y=x^3+a_2x^2+a_4x+a_6,
\end{align}
where $a_i$ are elements of $\F_q$. The genus of $\Efrak$ is $1$.
 An elliptic curve over $\F_q$ is also denoted by a pair $(\Efrak,O)$, where $\Efrak$ is the curve defined by the above Weierstrass equation \eqref{eq:Weierstrass_equation}, and $O$ is the \Emph{point at infinity} of $\Efrak$. Denote by $E/\F_q$ and $\Efrak(\F_q)$ the function field of $\Efrak/\F_q$ and the set of all rational points on $\Efrak/\F_q$, respectively. The function field $E$ is given by $E=\F_q(x,y)$, where transcendental elements $x$ and $y$ satisfy the above Weierstrass equation \eqref{eq:Weierstrass_equation}. Recall that $\Pbb^1_E$ denotes the set of rational places of $E$. There is a natural bijection between $\Efrak(\F_q)$ and $\Pbb_{E}^1$. Specifically, a rational point $(\alpha,\beta)$ corresponds to the unique common zero of $x-\alpha$ and $y-\beta$; and the point at infinity $O$ corresponds to the unique common pole of $x$ and $y$, which will still be denoted by $O$.
 
The set of all rational points $\Efrak(\F_q)$ has a natural
  structure of abelian group $(\Efrak(\F_q),\oplus)$ with zero element $O$ given by the
  chord-tangent group law (see \cite[Chapter III.2]{silverman2009arithmetic}).  
We now identify $\Pbb_E^1$ with the abelian group $\Efrak(\F_q)$ via the bijection described above. This means that $\Pbb^1_E$ is also an abelian group with zero element $O$, and we continue to use the symbol $\oplus$ for its addition. Moreover, for $P\in \Pbb_E^1$, we use $\ominus P$ to denote its inverse, and for $P,Q\in \Pbb_E^1$, we use $P\ominus Q$ to represent $P\oplus (\ominus Q)$, i.e, the subtraction. We also write $[m]P$ to stand for 
$$
[m]P:=\begin{cases}
    \underbrace{P\oplus\dots \oplus P}_{m \text{ times }},   &\text{ if } m \text{ is a positive integer;}\\
   O,    &\text{ if } m=0;\\
    \underbrace{\ominus P\dots \ominus P}_{-m \text{ times }},   &\text{ if } m \text{ is a negative integer.}\\
\end{cases}
$$ 

In what follows, we review the correspondence between the geometric group law of $\Pbb_E^1$ and the algebraic group law of $\Cl^0(E)$. 
The set of divisors of degree zero forms a subgroup of $\Div(E)$, denoted by $\Div^0(E)$.
Two divisors $A,B\in \Div(E)$ are called equivalent if there exists $w\in E^*$ such that $A=B+(w)$, and we denote this by $A\sim B$.
The set of divisors 
$\Princ(E):=\braces{(w)^E=\sum_{P\in\Pbb_E}v_P(w)P:\; w\in E^*}
$ is called the \Emph{group of principal divisors} of $E/\F_q$, which is a subgroup of the abelian group $\Div^0(E)$.
The \Emph{group of divisor classes of degree zero} of $E/\F_q$ is defined as the following quotient group  $$\Cl^0(E):=\Div^0(E)/\Princ(E).$$
By \cite[Chapter {\rm III}, Proposition 3.4 (e)]{silverman2009arithmetic}, there is a group isomorphism between $(\Pbb^1_E,\oplus)\cong (\Efrak(\F_q),\oplus)$ and $\Cl^0(E)$ given by
\begin{equation}\label{eq:202506032304}
\varphi:\begin{cases}\Pbb^1_E\xrightarrow{\sim} \Cl^0(E),\\ P\mapsto [P-O],
\end{cases}
\end{equation}
where $[P-O]$ denotes $P-O+\Princ(E)\in \Cl^0(E)=\Div^0(E)/\Princ(E)$.
 This implies the following lemma. 
\begin{lemma}\label{lem:deg0divisor_property}
Let $E/\F_q$ be an elliptic function field and let $P_{1},\dots,P_{n},P_{1}',\dots,P_{n}'$ be $2n$ (not necessarily distinct) rational places of $E$. Then
$$P_1+\dots+P_n\sim P_1'+\dots+P_n' \text{ if and only if } P_1\oplus \dots\oplus P_n=P_1'\oplus\dots\oplus P_n'.$$
\end{lemma}
\begin{proof}
  We have 
    $P_1+\dots+P_n\sim P_1'+\dots+P_n'$ if and only if
$P_1-O+\dots+P_n-O\sim P_1'-O+\dots+P_n'-O$, which is equivalent to
    $[P_1-O]+\dots+[P_n-O]=[P_1'-O]+\dots+[P_n'-O]$.
This is equivalent to $P_1\oplus \dots\oplus P_n=P_1'\oplus\dots\oplus P_n'$ by the group isomorphism $\varphi$ in \eqref{eq:202506032304}. 
\end{proof}
There is an upper bound on the number $N(E)=|\Pbb_E^1|=|\Efrak(\F_q)|$, which is the special case of the well-known Hasse–Weil bound for curves of genus $1$ (see \cite[Chapter {\rm V.1}, Theorem 1.1]{silverman2009arithmetic}). It states that 
$| N(E)-q-1| \leq 2\sqrt{q}.
$
An elliptic function field \( E/\F_q \) is called maximal if $N(E)$ attains the Hasse-Weil upper bound, i.e.,
$
N(E) = q + 2\sqrt{q} + 1.
$ 

We now recall some results on $N(E)$ and the group structure of $(\Pbb_E^1,\oplus)\cong (\Efrak(\F_q),\oplus)$.
We say that two elliptic curves $\Efrak_1$ and $\Efrak_2$ over $\F_q$ are \Emph{isogenous} if there is a non-constant smooth $\F_q$-morphism from $\Efrak_1$ to $\Efrak_2$ that sends the zero of $\Efrak_1$ to the zero of $\Efrak_2$ (see \cite{silverman2009arithmetic}). It is well known that two elliptic curves $\Efrak_1$ and $\Efrak_2$ over $\F_q$ are isogenous if and only if they have the same number of rational points. The following precise result is due to \cite{waterhouse1969abelian}. 

\begin{lemma}[{\cite[Theorem 4.1]{waterhouse1969abelian}}] \label{lem:classification_of_isogeny_class_of_EFF}
The isogeny classes of elliptic curves over $\F_q$ for $q=p^s$ are in one-to-one correspondence with the rational integers $t$ having $|t|\leq 2\sqrt{q}$
and satisfying some one of the following conditions:
\begin{itemize}
\item[{\rm (i)}] $(t,p)=1$;
\item[{\rm (ii)}] If $s$ is even: $t=\pm 2\sqrt{q}$;
\item[{\rm (iii)}] If $s$ is even and $p\not \equiv 1\pmod{3}$: $t=\pm \sqrt{q}$;
\item[{\rm (iv)}] If $s$ is odd and $p=2$ or $3$: $t=\pm p^{\frac{s+1}{2}}$;
\item[{\rm (v)}] If either {\rm (1)} $s$ is odd or {\rm (2)} $s$ is even and $p\not \equiv 1\pmod{4}: t=0.$
\end{itemize}
Furthermore, an elliptic curve in the isogeny class corresponding to $t$ has $q+1+t$ rational points.
\end{lemma}

As for the group structure of $(\Pbb_{E}^1,\oplus)\cong (\Efrak(\F_q),\oplus)$, the following result can be found in \cite[Theorem 3]{ruck1987note} and \cite[Theorem 9.97]{hirschfeld2008algebraic}, which is  summarized by Ma and Xing in \cite[Proposition 2.4]{ma2023group}.
\begin{lemma}\label{lem:group_structure_of_EFF} 
Let $\F_q$ be the finite field with $q=p^s$ elements. Let $h=\prod_\ell \ell^{h_\ell}$ be a possible number of rational places of an elliptic function field $E$ over $\F_q$. Then all the possible groups $\Pbb_E^1$ are
$\mathbb{Z}/p^{h_p}\Zbb\times \prod_{\ell \neq p}\left( \mathbb{Z}/\ell^{a_\ell}\Zbb\times \mathbb{Z}/\ell^{h_\ell-a_\ell}\Zbb\right)$
with \begin{itemize}
 \item[{\rm (a)}] In case {\rm (ii)} of Lemma~\ref{lem:classification_of_isogeny_class_of_EFF}: Each $a_\ell$ is equal to $h_\ell/2$, i.e, $\Pbb_E^1\cong \Zbb/(\sqrt{q}\pm1)\Zbb \times  \Zbb/(\sqrt{q}\pm1)\Zbb.$
 
\item[{\rm (b)}] In other cases of Lemma~\ref{lem:classification_of_isogeny_class_of_EFF}: $a_\ell$ is an arbitrary integer satisfying $0\le a_\ell \le \min\{\nu_\ell(q-1),[h_\ell/2]\}$.
In cases {\rm (iii)} and {\rm (iv)} of Lemma~\ref{lem:classification_of_isogeny_class_of_EFF}: $\Pbb_E^1\cong \Zbb/h\Zbb.$
 In case {\rm (v)} of Lemma~\ref{lem:classification_of_isogeny_class_of_EFF}: if $q\not\equiv -1(\text{mod } 4)$, then $\Pbb_E^1\cong \Zbb/(q+1)\Zbb$; otherwise, $\Pbb_E^1\cong \Zbb/(q+1)\Zbb$ or $\Pbb_E^1\cong \Zbb/2\Zbb \times \Zbb/\frac{q+1}{2}\Zbb$.
\end{itemize}
\end{lemma}
\subsection{Automorphism Groups of Elliptic Curves and Elliptic Function Fields}\label{subsec:2.4}

First, we review the automorphism groups of elliptic curves.  
Let $\Efrak/\F_q$ be an elliptic curve defined by the Weierstrass equation \eqref{eq:Weierstrass_equation}. We denote by $\Aut(\Efrak)$ the set of automorphisms of the elliptic curve $\Efrak$ over the algebraic closure $\overline{\F_q}$. We emphasize that every automorphism $\sigma\in \Aut(\Efrak)$ is required to fix the point at infinity $O$, that is, it must be an isogeny. For a characterization of $\Aut(\Efrak)$, see \cite[Chapter {\rm III}, Theorem 10.1]{silverman2009arithmetic} and its proof.
 
Next, we review the automorphism groups of elliptic function fields.
Let $E/\F_q$ be the function field of $\Efrak/\F_q$.
Define
$
\Aut(E/\F_q):=\{\sigma:\: \sigma \text{ is an } \F_q\text{-automorphism of } E \}.
$
It is a subgroup of the automorphism group $\Aut(E\overline{\F}_q/\overline{\F}_q)$.

For $\sigma\in \Aut(E/\F_q)$ and $P\in \Pbb_E$, it follows from the proof of \cite[Lemma 3.5.2 (a)]{stichtenoth2009algebraic} that $\sigma(P)$ is also a place of $E$.
Let $\Aut(E,O):=\{\sigma\in \Aut(E/\F_q):\;\sigma(O)=O\}$ be the set of $\F_q$-automorphisms of $E$ fixing $O$.
Then $\Aut(E,O)$ is a subgroup of $\Aut(\Efrak)$ in which every automorphism is defined over $\F_q$, i.e., it holds that $\Aut(E,O)=\Aut(\Efrak) \cap \Aut(E/\F_q).$
For each $Q\in \Pbb_E^1$, the translation-by-$Q$ map
$\tau_Q$ defined by $\tau_Q(P)=P\oplus Q$ induces an $\F_q$-automorphism of $E$. Let $T_E$ be the translation group $\{\tau_Q:\; Q\in \Pbb_E^1\}$ of the elliptic function field $E$, which is naturally isomorphic to the abelian group $\Pbb_E^1$. The following two results characterize the automorphism group of an elliptic function field and its subgroups. 
\begin{lemma}[{\cite[Theorem 3.1]{ma2023group}}]\label{lem:groupstructure_of_Aut_of_EFF}
Let $E/\F_q$ be an elliptic function field. The automorphism group of $E$ over $\F_q$ is the semidirect product of the translation group $T_E$ and the stabilizer $\Aut(E,O)$ of the infinite place $O$, i.e., \[\Aut(E/\F_q)=T_E \rtimes \Aut(E,O). \] The group law of $\Aut(E/\F_q)$ is given by $(\tau_P\alpha)\cdot (\tau_Q\beta)=\tau_{P\oplus \alpha(Q)}\cdot \alpha\beta$ for any $\tau_P,\tau_Q\in T_E$ and $\alpha,\beta\in \Aut(E,O)$.
\end{lemma} 
\begin{lemma}[{\cite[Proposition 3.2]{ma2023group}}]\label{lem:subgroup_of_auto_ECC}
Let $E/\F_q$ be an elliptic function field and let $G$ be a subgroup of $\Aut(E/\F_q)$. Then, we have $G \cong (T_E\cap G) \rtimes \pi(G)$, i.e., every subgroup of $\Aut(E/\F_q)$ is isomorphic to a semiproduct of a subgroup of $T_E$ and a subgroup of $\Aut(E, O)$.
\end{lemma}   
Conversely, given
a subgroup $T$ of $T_E$ and a subgroup $A$ of $\Aut(E, O)$, one may wonder under what conditions the product $TA$ is a subgroup of $\Aut(E/\F_q)$. The following lemma provides a useful necessary and sufficient condition.
 \begin{lemma}[{\cite[Proposition 3.3]{ma2023group}}]
 \label{lem:202506212038}
  Let $T$ be a subgroup of the translation group $T_E$ and let $A$ be a subgroup of $\Aut(E,O)$.
Then $TA$ is a subgroup of $\Aut(E/\F_q)$ if and only if $\tau_{\sigma^{-1}(Q)}\in T$ for all $\sigma\in A$ and $\tau_Q\in T$.
 \end{lemma}  
By the isomorphism between $T_E$ and $\Pbb_E^1$, any subgroup $T$ of $T_E$ can be written as $T=T_H:=\{\tau_Q:Q\in H\}$ for some subgroup $H$ of $\Pbb_E^1$. In the rest of the paper, we always adopt the symbol $T_H$, as several arguments will require explicit computations with $H$. With this notation, the above lemma is restated as follows.
\begin{lemma}[Restatement of {\cite[Proposition 3.3]{ma2023group}}]
\label{lem:isTAsubgroupofAutE}
  Let $H$ be a subgroup of $\Pbb_E^1$ and let $A$ be a subgroup of $\Aut(E,O)$.
Then $T_HA$ is a subgroup of $\Aut(E/\F_q)$ if and only if $\sigma(Q)\in H$ for all $\sigma\in A$ and $Q\in H$. 
\end{lemma}
We end this subsection with two facts, which will be applied in Section~\ref{sec:3}. They can be found in \cite[Chapter {\rm III}, Theorem~3.6 and its proof, Theorem~4.8]{silverman2009arithmetic} and \cite[Proposition 4.5]{ma2023group}.
\begin{remark}\label{rem:invo_of_EFF_and_isogeny_is_endo}
Let $E/\F_q$ be an elliptic function field defined by the Weierstrass equation  \eqref{eq:Weierstrass_equation}. Then the following hold.
\begin{itemize} 
\item[{\rm (i)}] There always exists an 
element of order $2$ in $\Aut(E,O)$ induced by the inversion operation in the group law of $\Pbb_E^1$, which is known as the \Emph{elliptic involution}. It can be explicitly defined by its action on $x$ and $y$ as $(x \mapsto x,\;y\mapsto -y-a_1x-a_3)$. 
With slight abuse of notation, we denote this automorphism by $[-1]$. For any subgroup $H\leq \Pbb_{E}^1$, it holds $\sigma(Q)\in H$ for all $\sigma \in A:= \langle[-1]\rangle =\{[-1],\mathrm{id}\}$ and $Q\in H$, which implies that $T_H\langle[-1]\rangle$ is a subgroup of $\Aut(E/\F_q)$ by Lemma~\ref{lem:isTAsubgroupofAutE}. 

\item[{\rm (ii)}] Any $\sigma \in \Aut(E,O)$ is an endomorphism (actually, an isomorphism) of the group $(\Pbb_E^1, \oplus)$, and, in particular, commutes with the above-mentioned elliptic involution $[-1]\in \Aut(E,O)$.
\end{itemize}
\end{remark}
\subsection{Hyperelliptic Curves and Hyperelliptic Function Fields}\label{sec:2.5}

Let $\F_q$ be a finite field of odd characteristic. A \Emph{hyperelliptic curve} $\Cfrak/\F_q$ of genus $g\geq 2$ over $\F_q$ is a projective, smooth, absolutely irreducible curve defined by the following equation
\begin{align}\label{eq:hyperellipticcurveequation}
 y^2 = f(x),
\end{align}
where \( f(x)\in \F_q[x]\) is a square-free polynomial of degree $2g+1$ or $2g+2$. It has one or two rational points at infinity, depending on whether the degree of the polynomial \( f(x) \) is odd or even. In this paper, for hyperelliptic curves, we will only consider the case $\deg(f(x))=2g+1$, in which case there is exactly one rational point at infinity, denoted by $P_{\infty}$.

Let \(E/\F_q\) denote the function field of \( \Cfrak/\F_q \), and let \( \Cfrak(\F_q) \) denote the set of rational points on $\Cfrak/\F_q$. The function field $E$ is given by $E=\F_q(x,y)$, where the transcendental elements $x$ and $y$ satisfy the equation \eqref{eq:hyperellipticcurveequation}. There is a one-to-one correspondence between \( \Cfrak(\F_q) \) and \( \Pbb^1_E \). Specifically, the rational point \( (\alpha, \beta) \) on \( \Cfrak/\F_q \) corresponds to the unique common zero of \( x - \alpha \) and \( y - \beta \), which we denote by \( P_{(\alpha,\beta)} \). The point at infinity $P_{\infty}$ corresponds to the unique common pole of $x$ and $y$, which we also denote by $P_{\infty}$. Throughout this paper, rational points and rational places not at infinity are called \Emph{affine rational points} and \Emph{affine rational places}, respectively.

Given an affine rational point $(\alpha,\beta)\in \Cfrak(\F_q)$, its hyperelliptic conjugate $(\alpha, -\beta)$ also lies on $\Cfrak/\F_q$ and corresponds to the unique common zero of $x - \alpha$ and $y + \beta$, which we denote by $\overline{{P}_{(\alpha,\beta)}}:=P_{(\alpha,-\beta)}\in \Pbb_E^1$.
 For a divisor $\sum_{i=1}^{g}P_i \in \Div(E)$ with $P_{1},\dots,P_{g}\in \Pbb_{E}^1$, we say that its affine part is \Emph{reduced} if $P_{i}\neq \overline{P_j}$ for any $1\leq i\neq j\leq g$ such that $P_i$ and $P_j$ are affine rational places. 
 Based on this definition, we have the following corollary derived from \cite[Proposition 1]{galbraith2008efficient}.  
It will be applied in Section~\ref{sec:4.2} in the special case $g=2$. 
\begin{lemma}\label{lem:Cl^0ofHyperellipticFF}
Let $q$ be an odd prime power, and let $E/\F_q$ be a hyperelliptic function field of genus $g\geq 2$ defined by the equation~\eqref{eq:hyperellipticcurveequation} with $\deg(f(x))=2g+1$.
Let $D_0 = \sum_{i=1}^g P_i$ and $D_0'=\sum_{i=1}^g P_i'$ be two effective divisors whose affine parts are reduced, where $P_1,\dots,P_g,P_1',\dots,P_g'\in \Pbb_E^1$. If $D_0\sim D_0'$, then we have $D_0=D_0'$.
\end{lemma}
\begin{proof}
Let $D_\infty:=gP_{\infty}\in \Div(E)$. Since $D_0\sim D_0'$, we have $[D_0-D_{\infty}]=[D_0'-D_{\infty}]\in \Cl^0(E)$. By the uniqueness of the representative described in \cite[Proposition 1]{galbraith2008efficient}, we have 
$D_0=D_0'$.
\end{proof}

Let $\Cfrak$ be an arbitrary curve of genus $g\geq 0$ defined over an arbitrary finite field $\F_q$, and let $E/\F_q$ be its function field. The Hasse–Weil bound \cite[Theorem~5.2.3]{stichtenoth2009algebraic} provides a bound on 
$N(E):=|\Pbb_E^1|=|\Cfrak(\F_q)|$.
\begin{lemma}
    Let $N(E)$ be defined as above. Then we have the following Hasse–Weil bound
    \begin{align}\label{Hasse-Weil_Bound}
        |N(E)-(q+1)|\leq 2g\sqrt{q}.
    \end{align}
\end{lemma}

A curve $\Cfrak/\F_q$ and its function field $E/\F_q$ are called maximal (minimal, respectively) if $N(E)=q+1+2g\sqrt{q}$ (if $N(E)=q+1-2g\sqrt{q}$, respectively).
 The following result is well known. 
\begin{lemma}\label{lem:maximal_curve_lift}
Assume that $q$ is a square and $\Cfrak$ is a curve of genus $g\geq 1$ over $\F_q$.
 If  $\Cfrak/\F_q$ is maximal, then $\Cfrak$ is maximal over $\F_{q^s}$ if and only if $s$ is odd. Furthermore, $\Cfrak$ is minimal over $\F_{q^s}$ if and only if $s$ is even.
\end{lemma}
\begin{proof} Since $\Cfrak/\F_q$ is maximal, its $L$-polynomial is $L(\Cfrak/\F_q,T)=1+2g\sqrt{q} T+(\text{higher order terms of } T)=\prod_{i=1}^{2g}(1-\alpha_{i}\sqrt{q}T)$ with $|\alpha_i|=1$ by \cite[Theorem 5.1.15]{stichtenoth2009algebraic}. Thus, $\alpha_1=\dots=\alpha_{2g}=-1$ and the $L$-polynomial over $\F_{q^s}$ is $L(\Cfrak/\F_{q^s},T)=(1-(-\sqrt{q})^sT)^{2g}$ by \cite[Theorem 5.1.15]{stichtenoth2009algebraic}. Since $\Cfrak$ is maximal (minimal, respectively) over $\F_{q^s}$ if and only if its $L$-polynomial is $(1+q^{s/2}T)^{2g}$ ($(1-q^{s/2}T)^{2g}$, respectively), this lemma is proved.  
\end{proof} 
The following two classes of hyperelliptic curves will be useful in our later constructions.
\begin{lemma}[{\cite[Theorem 1]{tafazolian2012note}}] 
\label{lem:ismaximalhypercurve_1}
 Let $q$ be an odd prime power. The smooth complete hyperelliptic curve $\Cfrak$ corresponding to
$
y^2 = x^{2g+1} + x
$ 
is maximal over \( \mathbb{F}_{q^2} \) if and only if \( q \equiv -1 \) or \( 2g + 1 \pmod{4g} \).   
\end{lemma}

\begin{lemma}[{\cite[Theorem 6]{tafazolian2012note}}] 
\label{lem:ismaximalhypercurve_2}
Let $q$ be an odd prime power. The smooth complete hyperelliptic curve $\Cfrak$ corresponding to
$
y^2 = x^{2g+1} + 1
$
is maximal over \( \mathbb{F}_{q^2} \) if and only if \( 2g + 1 \) divides \( q + 1 \).
\end{lemma} 
\subsection{Automorphism Groups of Hyperelliptic Curves and Hyperelliptic Function Fields of Genus $2$}
Every automorphism of a hyperelliptic curve $\Cfrak$ of genus $2$ is given by
\begin{align}\label{eq:matrix_resentation_of_HEFF_auto}
    \sigma:\left(x\mapsto\frac{ax+b}{cx+d},\; y\mapsto\frac{(ad-bc)y}{(cx+d)^3}\right),
\end{align}
associated with a uniquely determined matrix
$
     M_{\sigma}=\begin{pmatrix}
a & b \\
c & d
\end{pmatrix}\in \GL_2(\overline{\mathbb{F}}_q).
$
We use $\Aut(\Cfrak)$ to denote the automorphism group of $\Cfrak$ over $\overline{\F_q}$. It is isomorphic to a finite subgroup of $\GL_2(\overline{\mathbb{F}_q})$. For a hyperelliptic curve $\Cfrak$ of genus $2$ defined over $\F_q$, we say that an automorphism $\sigma \in \Aut(\Cfrak)$ is defined over $\F_q$ if
its associated matrix $M_{\sigma}$ is in $\GL_2(\F_q)$.
We denote ${\Aut}(\Cfrak/\F_q):=\{\sigma\in \Aut(\Cfrak): \sigma \text{ is defined over }\F_q\}$, which is a subgroup of ${\Aut}(\Cfrak)$. For each $\sigma\in \Aut(\Cfrak/\F_q)$, it naturally induces an $\F_q$-automorphism of the function field $E/\F_q=\F_q(\Cfrak)$, which will also be denoted by $\sigma$ for a little abuse of notations. Actually, this gives a one-to-one correspondence between $\Aut(\Cfrak/\F_q)$ and $\Aut(E/\F_q)$. Sometimes, we do not distinguish between them.    

Every hyperelliptic curve admits a special automorphism $\iota$ of order $2$, which is known as the hyperelliptic involution.
\begin{remark}\label{rem:hyperelliptic_involution}
Let $\Cfrak$ be a hyperelliptic curve of genus $2$ defined over $\F_q$ ($\Char(\F_q)\neq 2$) by the equation $y^2=f(x)$, and let $E/\F_q$ be its function field. There exists a special automorphism, the \Emph{hyperelliptic involution} $\iota\in \Aut(\Cfrak/\F_q)=\Aut(E/\F_q)$, defined by
$
    \iota: (x\mapsto x,\;y\mapsto -y),
$
with associated matrix $M_{\iota}=\parentheses{\begin{array}{cc} 
    -1 & 0 \\
    0 & -1
\end{array}}.$ It commutes with all elements in $\Aut(\Cfrak/\F_q)=\Aut(E/\F_q)$ since $M_{\iota}$ is a scalar matrix. Moreover, it holds $\iota(P_{(\alpha,\beta)})=\overline{P_{(\alpha,\beta)}}$ for any affine rational place $P_{(\alpha,\beta)}\in \Pbb_E^1$.
\end{remark}
There have been many studies on the automorphism group of hyperelliptic curves of genus $2$. For example, \cite{oskar1887onbinary, cardona1999oncurves, cardona2003onthenumber}.
Based on the results in \cite{cardona1999oncurves, cardona2003onthenumber}, Huang and Zhao \cite{huang2025optimal} determined the automorphism group of the hyperelliptic curve defined by $y^2=x^5+x$ over particular finite fields, which is helpful for our calculations later in Section~\ref{sec:4.2}. 
\begin{lemma}[{\cite[Lemma 8 and its proof]{huang2025optimal}}] \label{lem:autogroupbuhuangandzhao}
    Let $q$ be a power of an odd prime such that $8\mid (q-1)$ and $2^{1/2}\in\mathbb{F}_q$. Let $\Cfrak$ be a hyperelliptic curve defined over $\mathbb{F}_q$ by the equation $y^2=x^5+x$. Let $\tilde{S}_4$ and $\tilde{S}_5$ denote certain 2-coverings of the permutation groups $S_4$ and $S_5$, respectively. Then the following statements hold.
    \begin{itemize}
        \item[{\rm (i)}]  If $\Char(\F_q)\neq 3, 5$, then ${\Aut}(\Cfrak/\mathbb{F}_q)\simeq \tilde{S}_4$. 
   Specifically, $\Aut(\Cfrak/\F_q)\cong <U',V'>\leq \GL_2(\F_q)$ with \[
        U'=2^{-1/2}\begin{pmatrix}
        1 & -(-1)^{1/4} \\
        (-1)^{3/4} & -1
        \end{pmatrix}\quad \text{ and  } \quad V'=2^{-1/2}\begin{pmatrix}
        (-1)^{1/2}-1 & 0 \\
        0 & (-1)^{1/2}+1
        \end{pmatrix}.
    \] 
        \item[{\rm (ii)}] If $\Char(\mathbb{F}_q)=5$, then ${\Aut}(\Cfrak/\mathbb{F}_q)\simeq \tilde{S}_5$. 
        Specifically, $\Aut(\Cfrak/\F_q)\cong <U',V',W'>\leq \GL_2(\F_q)$ with
         \[
        U'=\begin{pmatrix}
        0 & -(-1)^{-1/4}\cdot2 \\
        -(-1)^{1/4}\cdot2 & 0
        \end{pmatrix}, \quad
        V'=\begin{pmatrix}
        0 & -(-1)^{-1/4}\cdot2 \\
        -(-1)^{1/4}\cdot2 & 1
        \end{pmatrix} \quad 
        \text{ and  } \quad W'=2^{1/2}\begin{pmatrix}
        1 & 0 \\
        0 & 2
        \end{pmatrix}.
        \] 
    \end{itemize}
\end{lemma} 
\section{Constructions of Optimal $(r,\delta)$-LRCs via Automorphism Groups of Elliptic Function Fields}\label{sec:3} 
\subsection{A General Framework for Constructing Optimal $(r,\delta)$-LRCs via Automorphism Groups of Elliptic Function Fields}
\label{sec:3.1}
In this subsection, we present a general framework for constructing optimal $(r,\delta)$-LRCs via automorphism groups of elliptic function fields, which is a generalization of that proposed in the works \cite{li2019optimal} and \cite{ma2023group}. 
We transform the construction of optimal $(r,\delta)$-LRCs based on automorphism groups of elliptic function fields into some conditions concerning the group law of $\Pbb_{E}^1$ (see Section~\ref{sec:2.3}). 
Before proceeding, recall that for any $H\leq \Pbb_E^1$, we denote by $T_H$ the subgroup $\{\tau_Q:\; Q\in H\}$ of the translation group $T_E$.
\begin{proposition}\label{prop:equivalentconditionforEFF_1}
	Let $E/\F_q$ be an elliptic function field. 
    Let $H$ be a subgroup of $\Pbb_E^1$ and let $A$ be a nontrivial subgroup of $\Aut(E, O)$ such that $G = T_HA$ is a
subgroup of $\Aut(E/\F_q)$ (see Lemma~\ref{lem:isTAsubgroupofAutE}). Let $|G|=r+\delta-1$ with $r\geq 1,\delta\geq 2$.
	Let $F:=E^{G}$, and let $[P_1,P_2,\dots,P_{r+\delta-1}]$ be a list\footnote{Here, a \Emph{list} refers to an ordered multiset of rational places of $E$, or equivalently, a finite sequence of (not necessarily distinct) rational places of $E$.} of rational places of $E$ such that $\sum_{j=1}^{r+\delta-1}P_j=\Con_{E/F}(Q_\infty)$ for a rational place $Q_{\infty}\in \Pbb_F^1$. Then the following statements hold.
\begin{itemize}
    \item 	
	[{\rm (i)}] There exists a function $z\in F$ such that $F=\F_q(z)$ and $(z)^E_{\infty}=\Con_{E/F}(Q_{\infty})=\sum_{j=1}^{r+\delta-1}P_j$.
	\item
	[{\rm (ii)}] 
    The set $\Lcal_{E}(\sum_{j=1}^{i+1}P_{j})\backslash \Lcal_{E}(\sum_{j=1}^{i}P_{j})$ is non-empty for each $1\leq i\leq r-1$.     
	  Moreover, let $w_i$ be an arbitrary element of $ \Lcal_{E}(\sum_{j=1}^{i+1}P_{j})\backslash \Lcal_{E}(\sum_{j=1}^{i}P_{j})$ for each $1\leq i\leq r-1$. Then $w_0:=1,w_1,\dots,w_{r-1}$ are linearly independent over $F$.
	\item
	[{\rm (iii)}] Let $\{P_{i,1},P_{i,2},\dots,P_{i,r+\delta-1}\}$ be pairwise distinct rational places of $E$ lying over a rational place $Q_i$ of $F$ for each $1\leq i\leq \ell$, such that $Q_{\infty},Q_1,\dots,Q_{\ell}$ are pairwise distinct.  
	Then all $r\times r$ submatrices of the following matrix
	\begin{align}\label{eq:Generator_Matrix_of_Local_Code}
		M_i:=\left(\begin{array}{cccc}
			w_0(P_{i,1})&w_0(P_{i,2})&\cdots&w_0(P_{i,r+\delta-1})\\
			w_1(P_{i,1})&w_1(P_{i,2})&\cdots&w_1(P_{i,r+\delta-1})\\
			\vdots &\vdots &\ddots &\vdots\\
			w_{r-1}(P_{i,1})&w_{r-1}(P_{i,2})&\dots&w_{r-1}(P_{i,r+\delta-1})\\
		\end{array}\right)  
	\end{align} 
	are invertible for each $1\leq i \leq \ell$ if and only if 
	\begin{align}\label{eq:equivconditionofEFFcons_3}
		\oplus_{j=1}^{r}P_{j}\notin \bigcup_{i=1}^{\ell}
		\braces{\oplus_{j=1}^{r}P_{i,u_j}:\; 1\leq u_1< \dots < u_{r}\leq r+\delta-1},
	\end{align}	
	or equivalently, 
	\begin{align}\label{eq:equivconditionofEFFcons_4}
		\oplus_{j=r+1}^{r+\delta-1}P_{j}\notin \bigcup_{i=1}^{\ell}
		\braces{\oplus_{j=1}^{\delta-1}P_{i,u_j}:\; 1\leq u_1< \dots < u_{\delta-1}\leq r+\delta-1}.
	\end{align}
    \end{itemize}
\end{proposition} 
 \begin{proof}
{\rm (i)} By Lemma~\ref{lem:thepropertyofGalExtension_1} {\rm (i)}, the set of all distinct places of $E$ lying over $O\cap F$ is $\{\sigma(O):\; \sigma\in T_HA\}=H$. Since $O$ is a rational place, we have $f(O|O\cap F)=1$. Applying Lemma~\ref{lem:thepropertyofGalExtension_1} {\rm (iii)} yields $|H|\cdot e(O|O\cap F)=|G|=|A||H|$, and hence, $e(O|O\cap F)=|A|\geq 2$. 
By Dedekind's different theorem (see \eqref{eq:Dedekinddifferentthm}), we have $\deg\Diff(E/F)\geq \deg((e(O|O\cap F)-1)O)\geq 1$. By Lemma~\ref{lem:is_rational_FF}, $F$ is a rational function
 	field. Therefore, we have $\dim_{\F_q}(\Lcal_{F}(Q_{\infty}))=2$. Let $z\in \Lcal_F(Q_{\infty})\backslash \F_q$. We have $[F:\F_q(z)]=\deg((z)_{\infty}^F)=\deg(Q_{\infty})=1$ and $(z)^E_{\infty}=\Con_{E/F}((z)^F_{\infty})=\Con_{E/F}(Q_{\infty})=\sum_{j=1}^{r+\delta-1}P_j$.   

  {\rm (ii)} Since the case $r=1$ is trivial, we henceforth assume $r\geq 2$.
   By the Riemann-Roch theorem (see \eqref{eq:ell(D)_of_deg(D)_geq_2g-1}), for any $1\leq i\leq r-1$, we have $\Lcal_{E}(\sum_{j=1}^{i+1}P_{j})\backslash \Lcal_{E}(\sum_{j=1}^{i}P_j)\neq \varnothing$. 
   Let $w_0:=1$ and $w_i\in \Lcal_{E}(\sum_{j=1}^{i+1}P_{j})\backslash \Lcal_{E}(\sum_{j=1}^{i}P_{j})$ for each $1\leq i\leq r-1$. 
   We need to show that $w_0,w_1,\dots,w_{r-1}$ are $F$-linearly independent. Assume towards a contradiction that there exist rational functions $\Frep_{0}(z),\dots,\Frep_{r-1}(z)\in F=\F_q(z)$, not all zero, such that $\sum_{i=0}^{r-1}\Frep_i(z)w_i=0$. By clearing denominators, we may assume that $\Frep_{0}(z),\dots,\Frep_{r-1}(z)$ are polynomials of $z$. 
   Let $t:=\max\{0\leq i\leq r-1:\;\deg(\Frep_i(z))=\max\{\deg(\Frep_j(z)):\;0\leq j\leq r-1\}\}$, where we adopt the convention $\deg(0):=-\infty$. It holds that $\Frep_t(z)\neq 0$. Then we consider the following two possible cases of $t$ towards deriving a contradiction. Before proceeding, we denote $e:=e(P_1|Q_{\infty})=\dots=e(P_{r+\delta-1}|Q_{\infty})$, and observe that $v_{P_1}(z)=\dots=v_{P_{r+\delta-1}}(z)=-e$.
\begin{itemize}
    \item  
  If $t=0$, then we have $\deg(\Frep_{0}(z))>\deg(\Frep_i(z))$ for any $1\leq i\leq r-1$. Hence, for any $1\leq i\leq r-1$ such that $\Frep_i(z)\neq 0$, we have
  $$v_{P_{r+\delta-1}}(\Frep_{0}(z)w_0)=-e\deg(\Frep_{0}(z))<-e\deg(\Frep_i(z))-(e-1)\leq -e\deg(\Frep_i(z))+v_{P_{r+\delta-1}}(w_i)= v_{P_{r+\delta-1}}(\Frep_i(z)w_i),$$
 where the symbol ``$\leq$'' is due to $w_i\in \Lcal_{E}(P_1+\dots+P_{r+\delta-1}-P_{r+\delta-1})=\Lcal_{E}(\sum_{P|Q_{\infty}}eP-P_{r+\delta-1})$. This implies $v_{P_{r+\delta-1}}(\sum_{i=0}^{r-1}\Frep_i(z)w_i)=v_{P_{r+\delta-1}}(\Frep_{0}(z)w_0)\neq \infty$ by the strict triangle inequality (see \cite[Lemma 1.1.11]{stichtenoth2009algebraic}).
  \item 
  If $1\leq t\leq r-1$, then we have $\deg(\Frep_t(z))\geq \deg(\Frep_i(z))$ for any $0\leq i\leq t-1$; and $\deg(\Frep_t(z))>\deg(\Frep_i(z))$ for any $t+1\leq i\leq r-1$. 
  Hence, for any $0\leq i\leq t-1$ such that $\Frep_i(z)\neq 0$, we have 
  \begin{align}\label{eq:202507131727}
  v_{P_{t+1}}(\Frep_t(z)w_t)=-e\deg(\Frep_t(z))+v_{P_{t+1}}(w_t)<-e\deg(\Frep_i(z))+v_{P_{t+1}}(w_i)=v_{P_{t+1}}(\Frep_i(z)w_i),
  \end{align}  
  where ``$<$'' is due to $\deg(\Frep_t(z))\geq \deg(\Frep_i(z))$ and $v_{P_{t+1}}(w_t)<v_{P_{t+1}}(w_i)$ since $v_{P_{t+1}}(w_t)\leq -v_{P_{t+1}}(\sum_{j=1}^{t+1}P_j)$,\footnote{This inequality holds because otherwise we would have $w_{t}\in \Lcal_E(\sum_{j=1}^{t}P_j)$, contradicting the fact that $w_t\in \Lcal_E(\sum_{j=1}^{t+1}P_j)\backslash \Lcal_E(\sum_{j=1}^{t}P_j)$.} $w_i\in \Lcal_{E}(\sum_{j=1}^{t+1}P_j-P_{t+1})$; and for any $t+1\leq i\leq r-1$ such that $\Frep_i(z)\neq 0$, we have
  \begin{align}\label{eq:202507131728}
  v_{P_{t+1}}(\Frep_t(z)w_t)=-e\deg(\Frep_t(z))+v_{P_{t+1}}(w_t)<-e\deg(\Frep_i(z))+v_{P_{t+1}}(w_i)=v_{P_{t+1}}(\Frep_i(z)w_i),
  \end{align}
  where the symbol ``$<$'' is due to $\deg(\Frep_t(z))> \deg(\Frep_i(z))$, $v_{P_{t+1}}(w_t)\leq -1$ and $v_{P_{t+1}}(w_i)\geq -e$ since $v_{P_{t+1}}(w_t)\leq -v_{P_{t+1}}(\sum_{j=1}^{t+1}P_j)$ and $w_i\in \Lcal_{E}(\sum_{j=1}^{r+\delta-1}P_j)=\Lcal_{E}(\sum_{P|Q_{\infty}}eP)$. 
  Inequalities~\eqref{eq:202507131727} and~\eqref{eq:202507131728} imply $v_{P_{t+1}}(\sum_{i=0}^{r-1}\Frep_i(z)w_i) \allowbreak=v_{P_{t+1}}(\Frep_t(z)w_t)\neq \infty$.
\end{itemize}   
  The above two cases both lead to a contradiction with $\sum_{i=0}^{r-1}\Frep_i(z)w_i=0$. 
  Therefore, $w_0,w_1,\dots,w_{r-1}$ are $F$-linearly independent.
  
   {\rm (iii)} For any $1\leq i\leq \ell$ and $1\leq u_1<\cdots<u_{r} \leq r+\delta-1$, we will show that the submatrix consisting of the $u_1,\dots,u_{r}$-th columns of $M_i$ in \eqref{eq:Generator_Matrix_of_Local_Code} 
    is singular if and only if $\oplus_{j=1}^{r}P_j= \oplus_{j=1}^{r}P_{i,u_j}$. 
	
	If the submatrix consisting of the $u_1,\dots,u_{r}$-th columns of $M_i$ is singular, then there exist $c_0,c_1,\dots,c_{r-1}\in \F_q$, not all zero, such that the vector $(c_0,c_1,\dots,c_{r-1})M_i$ vanishes at the positions $u_1,\dots,u_{r}$. Thus, the function $w:=\sum_{j=0}^{r-1}c_jw_j$ has $r$ zeros $P_{i,u_1},\dots,P_{i,u_r}$. Since $w\in \Lcal_{E}(P_1+\dots+P_{r})$, we have $w \in \Lcal_{E}(P_1+\dots+P_{r}-P_{i,u_1}-\dots-P_{i,u_r})$. This implies $(w)^E=-P_1-\dots-P_{r}+P_{i,u_1}+\dots+P_{i,u_r}$. And then it holds $\oplus_{j=1}^{r}P_j= \oplus_{j=1}^{r}P_{i,u_j}$ by Lemma~\ref{lem:deg0divisor_property}. 
	
	Conversely, $\oplus_{j=1}^{r}P_j = \oplus_{j=1}^{r}P_{i,u_j}$ implies $P_1+\dots+P_{r}\sim P_{i,u_1}+\dots+P_{i,u_r}$ by Lemma~\ref{lem:deg0divisor_property}, i.e., there exists a nonzero function $w\in E$ such that $(w)^E=-P_1-\dots-P_{r}+P_{i,u_1}+\dots+P_{i,u_r}$. Then we have $w\in \Lcal_{E}(\sum_{j=1}^{r}P_j)=\Span_{\F_q}\{w_0,w_1,\dots,w_{r-1}\}$, and thus $w=\sum_{j=0}^{r-1}c_jw_j$ for some $c_0,c_1,\dots,c_{r-1}\in \F_q$ that are not all zero, which means that $(c_0,c_1,\dots,c_{r-1})M_i$ vanishes at the positions $u_1,\dots,u_{r}$. And then the submatrix consisting of the $u_1,\dots,u_{r}$-th columns of $M_i$ is singular. 
	
	The above statements imply the equivalences stated in Proposition~\ref{prop:equivalentconditionforEFF_1} {\rm (iii)}, except for \eqref{eq:equivconditionofEFFcons_4}. It remains to establish the equivalence between \eqref{eq:equivconditionofEFFcons_4} and \eqref{eq:equivconditionofEFFcons_3}.
	
	Note that for any $1\leq i\leq \ell$, $\Lcal_{F}(Q_{\infty}-Q_i)\neq \{0\}$ since $F$ is a rational function field. Thus, there exists a nonzero function $z'\in \Lcal_{F}(Q_{\infty}-Q_i)\subseteq F$ such that $(z')^F=-Q_{\infty}+Q_i$, and then $(z')^E=\Con_{E/F}((z')^F)=\Con_{E/F}(-Q_{\infty}+Q_i)=(-P_1-\dots-P_{r+\delta-1})+(P_{i,1}+\dots+P_{i,r+\delta-1})$. This implies that $\oplus_{j=1}^{r+\delta-1}P_j=\oplus_{j=1}^{r+\delta-1}P_{i,j}$ by Lemma~\ref{lem:deg0divisor_property}. 
    Therefore,
    $$
    \oplus_{j=1}^{r}P_{j}\notin\{\oplus_{j=1}^{r}P_{i,u_j}:\; 1\leq u_1< \dots < u_{r}\leq r+\delta-1\}
    $$ 
    is equivalent to 
    $$
    \oplus_{j=r+1}^{r+\delta-1}P_{j}\notin\{\oplus_{j=1}^{\delta-1}P_{i,u_j}:\; 1\leq u_{1}< \dots < u_{\delta-1}\leq r+\delta-1\}.$$
    Consequently, \eqref{eq:equivconditionofEFFcons_4} and \eqref{eq:equivconditionofEFFcons_3} are equivalent. The proof is completed.
 \end{proof}   
\begin{remark}
\phantomsection
\label{rem:equivalentconditionforEFF_1}
\begin{itemize}
    \item[\rm (i)]  
In Proposition~\ref{prop:equivalentconditionforEFF_1}, when $\delta=2$, the equivalent condition \eqref{eq:equivconditionofEFFcons_4} becomes 
\begin{align*}
	P_{r+1}\notin \bigcup_{i=1}^{\ell}\{P_{i,u_1}: 1\leq u_1\leq r+1\}.
\end{align*}
This condition holds naturally due to the fact that $P_{r+1}\neq P_{i,j}$ for any $1\leq i\leq \ell$ and $1\leq j\leq r+1$. And thus all $r\times r$ submatrices of $M_i$ in \eqref{eq:Generator_Matrix_of_Local_Code} are invertible for each $1\leq i\leq \ell$.
This is exactly what is demonstrated in \cite[Proposition 4.2 {\rm (iii)}]{ma2023group}, where all $r\times r$ submatrices of the following matrix 
\begin{align*}
	\left(\begin{array}{cccc}
	 w_0(P_{i,1})&w_0(P_{i,2})&\cdots&w_0(P_{i,r+1})\\
	 w_1(P_{i,1})&w_1(P_{i,2})&\cdots&w_1(P_{i,r+1})\\
	 \vdots  &\vdots&\ddots &\vdots\\
	 w_{r-1}(P_{i,1})&w_{r-1}(P_{i,2})&\dots&w_{r-1}(P_{i,r+1})\\
	\end{array}\right) 
\end{align*} 
are proved to be invertible for each $1\leq i\leq \ell$. However, when $r>1$ and $\delta>2 $, the conditions \eqref{eq:equivconditionofEFFcons_3} and \eqref{eq:equivconditionofEFFcons_4} are not guaranteed to hold without additional assumptions. Later, in Sections \ref{sec:3.2} and \ref{sec:3.3}, for $r=2$ or $\delta=3$, we will provide two distinct sufficient conditions for an elliptic function field and the subgroups of its automorphism group, under which we can select rational places $[P_1,\dots,P_{r+\delta-1}],\{P_{1,1},\dots,P_{1,r+\delta-1}\}, \dots,\{P_{\ell,1},\dots,P_{\ell,r+\delta-1}\}$ satisfying conditions \eqref{eq:equivconditionofEFFcons_3} and \eqref{eq:equivconditionofEFFcons_4}. This, together with Proposition~\ref{prop:construction_of_(r,delta)from_EFF_1} below, leads to several classes of optimal $(r,3)$-LRCs and optimal $(2,\delta)$-LRCs.

\item[{\rm (ii)}]
In addition to generalizing the framework from the case of $r$-LRCs in \cite[Proposition 4.2]{ma2023group} to the case of $(r,\delta)$-LRCs, Proposition~\ref{prop:equivalentconditionforEFF_1} also provides the following improvement. In \cite[Proposition 4.2]{ma2023group}, $P_{1},\dots,P_{r+1}$ are required to be all distinct rational places lying over a rational place $Q_{\infty}\in \Pbb_{F}^1$ that splits completely in $E/F$, and then the functions $w_0,w_1,\dots,w_{r-1}$ are defined by the set $\{P_{1},\dots,P_{r+1}\}$. In Proposition~\ref{prop:equivalentconditionforEFF_1}, we allow the functions $w_0,w_1,\dots,w_{r-1}$ to be defined using a list $[P_1,\dots,P_{r+\delta-1}]$ of rational places of $E$ satisfying $\sum_{j=1}^{r+\delta-1}P_j=\Con_{E/F}(Q_{\infty})$ for a rational place $Q_{\infty}$ of $F$, regardless of whether $Q_{\infty}$ splits completely in $E/F$. This improvement is crucial for our later construction of optimal $(r,3)$-LRCs (and $(2,\delta)$-LRCs) in Theorem~\ref{thm:estimation_of_length_EFF_and_cons_2}. There, the elliptic involution $[-1]\notin G\leq \Aut(E/\F_q)$, and we need to select $[P_1,P_2,\dots,P_{r+2}]$ such that 
   \begin{align}\label{eq:3331}
       P_{r+1}\oplus P_{r+2}=O,
   \end{align}
  which is fulfilled by setting $Q_{\infty}=O\cap F$ and $P_{r+1}=P_{r+2}=O$. However, it is hard to fulfill \eqref{eq:3331} if we require that $Q_{\infty}$ splits completely in $E/F$. 
\item[{\rm (iii)}]
Assume that for $r=a$, $\delta=b$ (with $r+\delta-1=a+b-1=|G|$), there exist rational places $[P_1,P_2,\dots,P_{r+\delta-1}], \{P_{1,1}, \allowbreak\dots,P_{1,r+\delta-1}\},\dots,\{P_{\ell,1}, \dots,P_{\ell,r+\delta-1}\}$ satisfying the conditions in Proposition~\ref{prop:equivalentconditionforEFF_1}, including conditions \eqref{eq:equivconditionofEFFcons_3} and \eqref{eq:equivconditionofEFFcons_4}. Then for $r=b-1,\delta=a+1$ (note that $r+\delta-1=a+b-1=|G|$ still holds), there also exist rational places $[P'_1,P'_2,\dots,P'_{r+\delta-1}], \{P'_{1,1}, \allowbreak\dots,P'_{1,r+\delta-1}\},\dots,\{P'_{\ell,1}, \dots,P'_{\ell,r+\delta-1}\}$ satisfying the conditions in Proposition~\ref{prop:equivalentconditionforEFF_1}, including conditions \eqref{eq:equivconditionofEFFcons_3} and \eqref{eq:equivconditionofEFFcons_4}. Indeed, it suffices to set $P'_1=P_{r+\delta-1},\dots,P'_{r+\delta-1}=P_1$, and $P'_{i,j}=P_{i,j}$ for $1\leq i\leq \ell$ and $1\leq j\leq r+\delta-1$.
\end{itemize}
\end{remark}
Building directly on Proposition \ref{prop:equivalentconditionforEFF_1}, we have the following construction of optimal $(r,\delta)$-LRCs.
\begin{proposition}\label{prop:construction_of_(r,delta)from_EFF_1}
	We adopt the settings of Proposition~\ref{prop:equivalentconditionforEFF_1} and assume that at least one of conditions~\eqref{eq:equivconditionofEFFcons_3} and~\eqref{eq:equivconditionofEFFcons_4} holds.
	 Suppose two integers $t$ and $m$ satisfy $1\leq t< m\leq \ell$. Let $V:=\{a_{0,t}w_0z^t+\sum_{i=0}^{r-1}\sum_{j=0}^{t-1}a_{i,j}w_iz^j:a_{0,t}\in \F_q \text{ and }a_{i,j}\in \F_q \text{ for } 0\leq i\leq r-1, 0\leq j\leq t-1\}$, and let $\Pcal:=\{P_{1,1},\dots,P_{1,r+\delta-1},\dots,P_{m,1},\dots,P_{m,r+\delta-1}\}$. Define $\Ccal(\Pcal,V)$ by  
	 \begin{align}\label{eq:202506060032}
     \Ccal(\Pcal,V):=\left\{(\phi(P_{1,1}),\dots,\phi(P_{1,r+\delta-1}),\dots,\phi(P_{m,1}),\dots,\phi(P_{m,r+\delta-1})):\; \phi\in V\right\}.
	 \end{align}
Then the linear code $\Ccal(\Pcal,V)$ is an optimal $(r,\delta)$-LRC with parameters $[m(r+\delta-1),tr+1,(m-t)(r+\delta-1)]_q.$  
\end{proposition} 
\begin{proof}
By Proposition~\ref{prop:equivalentconditionforEFF_1} {\rm (ii)}, $w_0,w_1,\dots,w_{r-1}$ are $F$-linearly independent, which, along with the fact that $1,z,\dots,z^t\in F$ are $\F_q$-linearly independent, implies that $\dim_{\F_q}(V)=tr+1$. Note that $V\subseteq \Lcal_{E}(t(P_1+P_2+\dots+P_{r+\delta-1}))$. By Section~\ref{sec:2.1}, $\Ccal(\Pcal,V)$ is a linear code with parameters 
$[n=m(r+\delta-1),k=tr+1,d\geq (m-t)(r+\delta-1)]_q.$
In the following, we prove that $\Ccal(\Pcal,V)$ is an optimal $(r,\delta)$-LRC. 

Since $z\in F$ and $P_{i,1},\dots,P_{i,r+\delta-1}$ all lie over $Q_{i}\in \Pbb_F^1$ for each $1\leq i\leq m$, we have $z(P_{i,1})=\dots=z(P_{i,r+\delta-1})=z(Q_i)$, and thus $\Ccal(\Pcal,V)|_{\{(i-1)(r+\delta-1)+1,\dots,i(r+\delta-1)\}}=\{(\phi(P_{i,1}),\dots,\phi(P_{i,r+\delta-1})):\; \phi\in \Span_{\F_q}\{w_0,w_1,\dots,w_{r-1}\}\}$.  
The minimum distance of $\Ccal(\Pcal,V)|_{\{(i-1)(r+\delta-1)+1,\dots,i(r+\delta-1)\}}$ is equal to $\delta$ since the matrix $M_i$ in \eqref{eq:Generator_Matrix_of_Local_Code} is a generator matrix of $\Ccal(\Pcal,V)|_{\{(i-1)(r+\delta-1)+1,\dots,i(r+\delta-1)\}}$ and all $r\times r$ submatrices of $M_i$ are invertible by Proposition~\ref{prop:equivalentconditionforEFF_1} {\rm (iii)}.   
    Thus, $\Ccal(\Pcal,V)$ is an $(r,\delta)$-LRC. By the Singleton-type bound~\eqref{eq:SLboundfor_r_delta}, we have 
    $d\leq n-k+1-(\lceil k/r\rceil-1)(\delta-1)=(m-t)(r+\delta-1).$
Therefore, the minimum distance $d$ of $\Ccal(\Pcal,V)$ is determined to be $(m-t)(r+\delta-1)$ and $\Ccal(\Pcal,V)$ is an optimal $(r,\delta)$-LRC.  
\end{proof} 
\begin{remark}
\phantomsection
\label{rem:construction_of_(r,delta)from_EFF_1} 
\begin{itemize}
 \item[{\rm (i)}] Ma and Xing (see the proof of \cite[Proposition 4.4]{ma2023group}) employed the modified algebraic geometry codes to lengthen the optimal $r$-LRC by $(r+1)$, by allowing $P_1,\dots,P_{r+1}$ to be evaluation points (where $P_1,\ldots,P_{r+1}$ are required to be pairwise distinct). Our framework (Proposition~\ref{prop:equivalentconditionforEFF_1} and Proposition~\ref{prop:construction_of_(r,delta)from_EFF_1}) can also achieve this (when $\delta=2$, and $P_1,\dots,P_{r+1}$ are pairwise distinct, or equivalently, $Q_{\infty}$ splits completely in $E/F$).
     However, in general, or more precisely, in the case where $\delta\geq 3$ and $P_1,\dots,P_{r+\delta-1}$ are pairwise distinct, the technique of modified algebraic geometry code may not be feasible, since the $(r,\delta)$-locality can no longer be guaranteed on the extended evaluation points $P_1,\dots,P_{r+\delta-1}$.

 \item[{\rm (ii)}] 
          When $\delta=2$, by \cite[Proposition 4.1]{ma2023group}, Proposition~\ref{prop:equivalentconditionforEFF_1}, Remark~\ref{rem:equivalentconditionforEFF_1} {\rm (i)}, Proposition~\ref{prop:construction_of_(r,delta)from_EFF_1}, and Remark~\ref{rem:construction_of_(r,delta)from_EFF_1} {\rm (i)}, we can recover \cite[Proposition 4.4]{ma2023group}. Alternatively, one can set $\delta=2$ and $Q_{\infty}=O\cap F$ in Proposition~\ref{prop:equivalentconditionforEFF_1} and ~\ref{prop:construction_of_(r,delta)from_EFF_1} to recover \cite[Proposition 4.4]{ma2023group} by \cite[Proposition 4.1]{ma2023group} and Remark~\ref{rem:equivalentconditionforEFF_1} {\rm (i)}.
     \end{itemize}     
\end{remark} 
 As mentioned in Remark~\ref{rem:equivalentconditionforEFF_1} {\rm (i)}, when $r\geq 2$ and $\delta\geq 3$, the conditions \eqref{eq:equivconditionofEFFcons_3} and \eqref{eq:equivconditionofEFFcons_4} will not hold by default like the case of $r$-LRCs (i.e., $\delta=2$). 
In the next two subsections, we consider specific settings of the elliptic function fields $E/\F_q$ along with suitable automorphism subgroups $G=T_HA\leq \Aut(E/\F_q)$, under which there exist rational places $[P_1,\dots,P_{r+\delta-1}], \{P_{1,1},\dots,P_{1,r+\delta-1}\},\dots, \allowbreak \{P_{\ell,1},\dots,\allowbreak P_{\ell,r+\delta-1}\}$ satisfying conditions~\eqref{eq:equivconditionofEFFcons_3} and~\eqref{eq:equivconditionofEFFcons_4}. Consequently, we arrive at several families of explicit optimal $(r,3)$-LRCs and $(2,\delta)$-LRCs.

\subsection{Construction {\rm I} of Optimal $(r,3)$-LRCs and $(2,\delta)$-LRCs by the General Framework}\label{sec:3.2}
In this subsection, we construct our first explicit family of optimal $(r,3)$-LRCs and $(2,\delta)$-LRCs based on the framework proposed in the previous subsection. This class of constructions relies on elliptic function fields with odd rational places, utilizing a subgroup $T_H\langle[-1]\rangle\leq \Aut(E/\F_q)$.  
\begin{theorem} \label{thm:estimation_of_length_EFF_and_cons_1}
Let $E/\F_q$ be an elliptic function field with $N(E)$ rational places satisfying $2\nmid N(E)$. Let $H$ be a subgroup of $\Pbb_E^1$ of order $h\geq 3$, and let $A:= \langle[-1]\rangle \leq  \Aut(E,O)$ (see Remark~\ref{rem:invo_of_EFF_and_isogeny_is_endo} {\rm (i)}). Let $G:=T_{H} A\leq\Aut(E/\F_q)$, $F:=E^G$, and $r=2h-2$, $\delta=3$ (or $r=2$, $\delta=2h-1$). Then there exist rational places $[P_1,\dots,P_{r+\delta-1}]$, $\{P_{1,1},\dots,P_{1,r+\delta-1}\}$,\dots ,$\{P_{\ell,1},\dots,\allowbreak P_{\ell,r+\delta-1}\}$ of $E$ satisfying the conditions in Proposition~\ref{prop:equivalentconditionforEFF_1} (including conditions \eqref{eq:equivconditionofEFFcons_3} and \eqref{eq:equivconditionofEFFcons_4}), where $\ell=\frac{N(E)-h}{2h}-1$. 

Consequently, by Propositions~\ref{prop:equivalentconditionforEFF_1} and~\ref{prop:construction_of_(r,delta)from_EFF_1}, there exist an optimal $(r=2h-2,\delta=3)$-LRC and an optimal $(r=2,\delta=2h-1)$-LRC with parameters 
$
 [m\cdot 2h,tr+1,(m-t)\cdot 2h]_q     
$ 
for any $1\leq t<m\leq \ell=\frac{N(E)-h}{2h}-1$.
\end{theorem} 
\begin{proof}
It suffices to prove the case $r=2h-2$ and $\delta=3$, from which the case $r=2$ and $\delta=2h-1$ can be deduced by Remark~\ref{rem:equivalentconditionforEFF_1} {\rm (iii)}.

Let $r=2h-2$ and $\delta=3$. In the following, we select rational places $[P_1,\dots,P_{r+2}]$, $\{P_{1,1},\dots,P_{1,r+2}\}$, \dots, $\{P_{\ell,1},\dots,P_{\ell,r+2}\}$ of $E$ that satisfy the conditions of Proposition~\ref{prop:equivalentconditionforEFF_1} (including the equivalent condition \eqref{eq:equivconditionofEFFcons_4}).  
To this end, we first consider the number of rational places of $F$ that split completely in $E/F$. We claim that this number is equal to 
\begin{align}\label{eq:estima_1}
     \ell+1=\frac{N(E)-h}{2h}.  
\end{align} 
Indeed, the following two statements hold.

(1) for any $P\in H$, $P\cap F=O\cap F$ does not split completely in $E/F$;

(2) for any $P\in \Pbb_E^1\backslash H$, $P\cap F$ splits completely in $E/F$. 

To prove the statement (1) where $P\in H$, it suffices to observe that $\Con_{E/F}(O\cap F)=\sum_{R\in H}2R$ by Lemma~\ref{lem:thepropertyofGalExtension_1}. Consequently, $e(P|P\cap F)=e(P|O\cap F)=2$, and thus the statement (1) holds. 

To prove the statement (2) where $P\in \Pbb_E^1\backslash H$, we assume towards a contradiction that $P\cap F$ does not split completely in $E/F$. By Remark~\ref{rem:thepropertyofGalExtension_2} {\rm (i)}, this implies that $\sigma_1(P)=\sigma_2(P)$ for some $\sigma_1\neq \sigma_2\in G=T_H\langle [-1] \rangle$, which further implies 
\begin{align}\label{eq:20250520_1932}
    P=\sigma_1^{-1}\sigma_2(P).
\end{align} 
Let $\sigma_1^{-1}\sigma_2 = \tau_Q [-1]^i$ for some $Q \in H$ and $i \in \{0,1\}$. If $i = 0$, then \eqref{eq:20250520_1932} implies that $Q=O$, which would mean $\sigma_1^{-1}\sigma_2= \mathrm{id}$, contradicting with $\sigma_1\neq \sigma_2$. Thus, $i=1$, and it follows that $$P=\sigma_1^{-1}\sigma_2(P)=\tau_Q [-1](P)=[-1]P\oplus Q.$$ 
This implies $[2]P=Q\in H$, and then $P=[N(E)+1]P=\brackets{\frac{N(E)+1}{2}}([2]P)=\brackets{\frac{N(E)+1}{2}}Q\in H$, which contradicts $P\in \Pbb_E^1\backslash H$. Therefore, the statement (2) holds, hence the number of rational places of $F$ that split completely in $E/F$ is exactly $\ell+1=\frac{N(E)-h}{2h}.$ 
We denote the sets of all rational places of $E$ that lie over these $\ell+1$ rational places of $F$ by $\{P_1,\dots,P_{r+2}\}, \{P_{1,1},\dots,P_{1,r+2}\},\dots, \{P_{\ell,1},\dots,P_{\ell,r+2}\}$, respectively. 

We claim that there exists a reorder of $P_1,\dots,P_{r+2}$ such that $[P_1,\dots,P_{r+2}], \{P_{1,1},\dots,P_{1,r+2}\},\dots,\{P_{\ell,1},\dots,P_{\ell,r+2}\}$ satisfy the condition \eqref{eq:equivconditionofEFFcons_4}, i.e., 
    \begin{align}\label{eq:tempcondition}
    	P_{r+1}\oplus P_{r+2}\notin \bigcup_{i=1}^{\ell}
    	\braces{P_{i,u_1}\oplus P_{i,u_2}:\; 1\leq u_1< u_{2}\leq r+2}.
    \end{align} 
The proof of this claim is as follows. Since $E/F=E/E^{G}$ is a Galois extension with $\Gal(E/F)=G$, the group $G=T_{H} \langle[-1]\rangle$ acts transitively on each of $\{P_1,\dots,P_{r+2}\}, \{P_{1,1},\dots,P_{1,r+2}\}, \dots,\{P_{\ell,1},\dots,P_{\ell,r+2}\}$ by Lemma~\ref{lem:thepropertyofGalExtension_1} {\rm(i)}. Since the order of $G$ is equal to $r+2$, which is equal to the cardinality of these sets, these pairwise disjoint sets can be represented as 
    \begin{equation}\label{eq:202506120017}
    \begin{cases}
    	 \{P_1,\dots,P_{r+2}\}&=\{\sigma(P_1):\; \sigma\in G\}=\parentheses{P_1\oplus H}\sqcup \parentheses{[-1]P_1\oplus H},\\
    	  \{P_{1,1},\dots,P_{1,r+2}\}&=\{\sigma(P_{1,1}):\; \sigma\in G\}=\parentheses{P_{1,1}\oplus H}\sqcup \parentheses{[-1]P_{1,1}\oplus H},\\
    	  &\dots,\\
    	 \{P_{\ell,1},\dots,P_{\ell,r+2}\}&=\{\sigma(P_{\ell,1}):\;\sigma\in G\}=\parentheses{P_{\ell,1}\oplus H}\sqcup \parentheses{[-1]P_{\ell,1}\oplus H},\\ 
    \end{cases}
    \end{equation}
   where $P_1\oplus H$ denotes the coset $\{P_1\oplus Q:Q\in H\}$, and $\sqcup$ denotes the union without intersection. We now prove that 
   $$\parentheses{(P_1\oplus H) \oplus (P_1\oplus H)}\cap \parentheses{\bigcup_{i=1}^{\ell}
   \braces{\oplus_{j=1}^{2}P_{i,u_j}: 1\leq u_1< u_{2}\leq r+2}}=\varnothing,$$
   i.e., by \eqref{eq:202506120017},
   $$\parentheses{[2]P_1\oplus H}\cap \parentheses{\bigcup_{i=1}^{\ell}([2]P_{i,1}\oplus H)\cup H\cup  ([-2]P_{i,1}\oplus H)}=\varnothing,$$ 
   and then we can choose two arbitrary distinct elements in $P_1\oplus H\subseteq \{P_{1},P_2,\dots,P_{r+2}\}$ to serve as new $P_{r+1}$, $P_{r+2}$ such that \eqref{eq:tempcondition} holds.
   Assume towards a contradiction that 
   $$([2]P_1\oplus H)\cap \big( ([2]P_{i,1}\oplus H)\cup H\cup  ([-2]P_{i,1}\oplus H)\big)\neq \varnothing$$
  for some $1\leq i\leq \ell$. Then at least one of $[2]P_1\ominus[2]P_{i,1}, [2]P_1, [2]P_1\ominus[-2]P_{i,1}$ is in $H$ by the property of cosets. We consider them separately and derive a contradiction in each case.
  \begin{itemize}
  	\item $[2]P_1\ominus[2]P_{i,1}\in H$: In this case we have $[2](P_1\ominus P_{i,1})\in H$, which implies that $$P_1\ominus P_{i,1}=[N(E)+1](P_1\ominus P_{i,1})=\brackets{\frac{N(E)+1}{2}}\big([2](P_1\ominus P_{i,1})\big)\in H.$$
  	This leads to a contradiction with the fact that $(P_1\oplus H) \cap (P_{i,1} \oplus H)=\varnothing$ by \eqref{eq:202506120017}.
  	\item $[2]P_1\in H$: In this case we have $$P_1=[N(E)+1](P_1)=\brackets{\frac{N(E)+1}{2}}\big([2]P_1)\big)\in H.$$
  	This leads to a contradiction with the fact that $\{P_1,\dots,P_{r+2}\}=(P_1\oplus H)\sqcup ([-1]P_1\oplus H)$ is a set consisting of $r+2=2h=2|H|$ distinct rational places by \eqref{eq:202506120017}.
  	\item $[2]P_1\ominus[-2]P_{i,1}\in H$: In this case we have $[2](P_1\oplus P_{i,1})\in H$, which implies that $$P_1\oplus P_{i,1}=[N(E)+1](P_1\oplus P_{i,1})=\brackets{\frac{N(E)+1}{2}}\big([2](P_1\oplus P_{i,1})\big)\in H.$$
  	This leads to a contradiction with the fact that $(P_1\oplus H )\cap ([-1]P_{i,1} \oplus H)=\varnothing$ by \eqref{eq:202506120017}.
  \end{itemize} 
The claim is established, and therefore this theorem is proved.
\end{proof}
\begin{remark} 
In the above proof, we determined the precise number of rational places of $F$ that split completely in $E/F$: $\ell+1=\frac{N(E)-h}{2h}$ (see \eqref{eq:estima_1} and its corresponding argument) under the conditions $2\nmid N(E)$ and $G=T_H\langle[-1]\rangle$. This precise number can also be used to improve the (estimated) upper bound of the length of optimal $r$-LRCs presented in \cite[Proposition 4.6]{ma2023group} (in the case $2\nmid N(E)$). If $2\nmid N(E)$, then the range of the number of local repair groups $m$ in \cite[Proposition 4.6]{ma2023group} can be improved from ``$1\leq t<m\leq \lceil\frac{N(E)}{r+1}\rceil-2$'' to ``$1\leq t< m\leq \lceil\frac{N(E)}{r+1}\rceil-1$''.
\end{remark}
In the following, we present two representative explicit constructions of optimal $(r,3)$-LRCs and $(2,\delta)$-LRCs with lengths slightly less than $q+2\sqrt{q}$ from the above Theorem \ref{thm:estimation_of_length_EFF_and_cons_1}. There exist other constructions with code lengths exceeding $q+1$. We do not list them all here.
\begin{corollary}\label{cor:Cons_via_EFF_1.1}
  Let $q=2^{2s}$ for a positive integer $s$. For any positive divisor $h\geq 3$ of $q+2\sqrt{q}+1$, there exist an optimal $(r=2h-2,\delta=3)$-LRC and an optimal $(r=2,\delta=2h-1)$-LRC with parameters $[m \cdot 2h, tr+1,(m-t)\cdot 2h]_q$ for any integers $t$ and $m$ satisfying $1\leq t<m\leq \frac{q+2\sqrt{q}+1-h}{2h}-1$. 
\end{corollary}
\begin{proof}  
By \cite[Lemma 15]{li2019optimal}, there exists an explicit maximal elliptic function field $E/\F_q$ with $N(E)=q+2\sqrt{q}+1$. Since $2\nmid N(E)$ and there exists a subgroup $H\leq \Pbb_E^1$ of order $h$ for any divisor $h\geq 3$ of $N(E)$, this corollary holds by Theorem~\ref{thm:estimation_of_length_EFF_and_cons_1}.
\end{proof}
 
\begin{corollary}\label{cor:Cons_via_EFF_1.2}
      Let $q=p^{2s}$ for an odd prime $p$ and a positive integer $s$. For any positive divisor $h\geq 3$ of $q+2\sqrt{q}$, there exist an optimal $(r=2h-2,\delta=3)$-LRC and an optimal $(r=2,\delta=2h-1)$-LRC with parameters $[m\cdot 2h, tr+1, (m-t) \cdot 2h]_q$ for any integers $t$ and $m$ satisfying $1\leq t<m\leq \frac{q+2\sqrt{q}-h}{2h}-1$. 
\end{corollary}
\begin{proof} 
By Lemma~\ref{lem:classification_of_isogeny_class_of_EFF} (i), there exists an elliptic function field $E/\F_q$ with $N(E)=q+2\sqrt{q}$. Since $2\nmid N(E)$ and there exists a subgroup $H\leq \Pbb_E^1$ of order $h$ for any divisor $h\geq 3$ of $N(E)$, this corollary follows from Theorem~\ref{thm:estimation_of_length_EFF_and_cons_1}.
\end{proof}
 
\subsection{Construction {\rm II} of Optimal $(r,3)$-LRCs and $(2,\delta)$-LRCs by the General Framework}\label{sec:3.3}
In this subsection, we present the second family of optimal $(r,3)$-LRCs and $(2,\delta)$-LRCs based on the general framework in Section~\ref{sec:3.1}, making use of a subgroup $T_HA\leq \Aut(E/\F_q)$ with $[-1]\notin A$.
\begin{theorem}\label{thm:estimation_of_length_EFF_and_cons_2} 
Let $E/\F_q$ be an elliptic function field with $N(E)$ rational places. 
Let $H$ be a subgroup of order $h$ of $\Pbb_E^1$ and let $A$ be a nontrivial subgroup of order $a$ of $\Aut(E, O)$ such that $G:= T_HA$ is a subgroup of $\Aut(E/\F_q)$ (see Lemma~\ref{lem:isTAsubgroupofAutE}).
Assume $[-1]\notin A$ and $ah\geq 3$. Let $F:=E^G$, and $r=ah-2, \delta=3$ (or $r=2, \delta=ah-1$). Then there exist rational places $[P_1,\dots,P_{r+\delta-1}],\{P_{1,1},\dots,P_{1,r+\delta-1}\},\dots,\{P_{\ell,1},\dots,P_{\ell,r+\delta-1}\}$ of $E$ satisfying the conditions in Proposition~\ref{prop:equivalentconditionforEFF_1} (including conditions \eqref{eq:equivconditionofEFFcons_3} and \eqref{eq:equivconditionofEFFcons_4}), where 
	\begin{align}\label{eq:202505271725}
		\ell=\begin{cases}
			  \ceil{\frac{N(E)-(ah+3h)/2}{ah}},& \text{ if }2\nmid aN(E);\\ 
		2\ceil{\frac{N(E)-2ah-2h}{2ah}},& \text{ if }2\mid aN(E).\\ 
		\end{cases}
	\end{align} 
 Consequently, by Propositions~\ref{prop:equivalentconditionforEFF_1} and~\ref{prop:construction_of_(r,delta)from_EFF_1}, there exist an optimal $(r=ah-2,\delta=3)$-LRC and an optimal $(r=2,\delta=ah-1)$-LRC with parameters 
$
 [m\cdot ah,tr+1,(m-t)\cdot ah]_q     
$ 
for any $1\leq t<m\leq \ell$.
\end{theorem} 
\begin{proof}
It suffices to prove the case $r=ah-2$ and $\delta=3$, from which the case $r=2$ and $\delta=ah-1$ can be deduced by Remark~\ref{rem:equivalentconditionforEFF_1} {\rm (iii)}.

Let $r=ah-2$ and $\delta=3$. By Lemma~\ref{lem:thepropertyofGalExtension_1}, we obtain $\Con_{E/F}(O\cap F)=\sum_{P\in H}aP$. 
We directly define the list of rational places $[P_1,P_2,\dots,P_{r+2}]$ of $E$ to be an arbitrary list such that $\sum_{j=1}^{r+2}P_j=\Con_{E/F}(O\cap F)=\sum_{P\in H}aP$ (i.e., letting $Q_{\infty}=O\cap F$ in Proposition~\ref{prop:equivalentconditionforEFF_1}) and $P_{r+1}=P_{r+2}=O$, which is valid since $a=|A|\geq 2$. 
In the following, we divide our discussion into two cases: $2\nmid aN(E)$ and $2\mid aN(E)$, and we separately select sets of rational places $\{P_{1,1},\dots,P_{1,r+2}\}$, \dots, $\{P_{\ell,1},\dots,P_{\ell,r+2}\}$ of $E$ satisfying the conditions in Proposition~\ref{prop:equivalentconditionforEFF_1}, including the condition~\eqref{eq:equivconditionofEFFcons_4}, which now takes the form:
	\begin{align}\label{eq:tempcondition2}
 	P_{r+1}\oplus P_{r+2}=O\oplus O=O\notin \bigcup_{i=1}^{\ell}
		 \braces{P_{i,u_1}\oplus P_{i,u_2}:\; 1\leq u_1 < u_{2}\leq r+2}.
	\end{align}   
    
(1) When $2\nmid aN(E)$, we start to select sets of rational places $\{P_{1,1},\dots,P_{1,r+2}\}$, \dots, $\{P_{\ell,1},\dots,P_{\ell,r+2}\}$ of $E$ that satisfy the condition \eqref{eq:tempcondition2}. To this end, we first estimate the number of rational places of $F$ that split completely in $E/F$. We claim that this number is at least 
\begin{align}\label{eq:estima_2}
   \ell=\ceil{\frac{N(E)-(ah+3h)/2}{ah}}.
\end{align}
It can be proved as follows. Note that $F$ is a rational function field by Proposition~\ref{prop:equivalentconditionforEFF_1} {\rm (i)}. 
By the Hurwitz genus formula (see \eqref{eq:kurwitzformulacoro_1}), we have 
     \begin{align}\label{eq:20250522_0002}
     	 0=2g(E)-2&=(2g(F)-2)[E:F]+\deg\Diff(E/F)=-2ah+\deg\Diff(E/F).
     \end{align} 
By Lemma~\ref{lem:thepropertyofGalExtension_1}, we have $\Con_{E/F}(O\cap F)=\sum_{P \in H}aP$. 
Let $R_1,\dots,R_v$ denote all the distinct rational places of $E$ outside $H$ that are ramified in $E/F$. 
Then we have $e(R_i|R_i\cap F)\geq 3$ for each $1\leq i\leq v$. 
This is because $e(R_i|R_i\cap F)\mid [E:F]=|G|=ah$ by Lemma~\ref{lem:thepropertyofGalExtension_1} (iii), and $2\nmid ah$ by $2\nmid aN(E)$.  
By Dedekind's different theorem (see \eqref{eq:Dedekinddifferentthm}), we have 
$
\deg\Diff(E/F)\geq \deg\left(\sum_{P\in H}(a-1)P+\sum_{i=1}^{v}(e(R_i|R_i\cap F)-1)R_i \right)\geq (a-1)h+2v.
$ 
Hence, it holds that $2ah=\deg\Diff(E/F)\geq (a-1)h+2v$ by \eqref{eq:20250522_0002}, which implies $v\leq \frac{ah+h}{2}$. Including those rational places in $H$, there are at most $(ah+3h)/2$ rational places of $E$ that are ramified in $E/F$. Therefore, by Remark~\ref{rem:thepropertyofGalExtension_2} {\rm (ii)}, there exist at least $\ell=\ceil{\frac{N(E)-(ah+3h)/2}{ah}}$ distinct rational places of $F$ that split completely in $E/F$ into $\{P_{1,1},\dots,P_{1,r+2}\}, \dots, \{P_{\ell,1},\dots,P_{\ell,r+2}\}\subseteq \Pbb_E^1 $, respectively.  

We claim that these $\ell$ sets of rational places $\{P_{1,1},\dots,P_{1,r+2}\}$, \dots, $\{P_{\ell,1},\dots,P_{\ell,r+2}\}$ satisfy \eqref{eq:tempcondition2}.
To prove this claim, we only need to prove that for all $1\leq i\leq \ell$ and $1\leq u_1<u_2\leq r+2$,
\begin{align*}
	P_{i,u_1}\oplus P_{i,u_2}\neq O.
\end{align*}
Assume towards a contradiction that $P_{i,u_1}\oplus P_{i,u_2}=O$ for some $1\leq i\leq \ell$ and  $1\leq u_1<u_2\leq r+2$. Then we have $P_{i,u_2}=[-1]P_{i,u_1}$. Since $|G|=ah=r+2$ and $G$ acts transitively on $\{P_{i,1},\dots,P_{i,r+2}\}$ by Lemma~\ref{lem:thepropertyofGalExtension_1} (i), we have
\begin{align}\label{eq:negativeeq}
\{P_{i,1},\dots,P_{i,r+2}\}=\{\sigma([-1]P_{i,u_1}):\; \sigma\in G\}=\{[-1](\sigma(P_{i,u_1})):\; \sigma\in G\}=\{[-1]P_{i,1},\dots,[-1]P_{i,r+2}\},	 
\end{align} 
where the second ``$=$'' is due to 
\begin{align}
	\{\sigma[-1]:\; \sigma\in G\} \nonumber
	&=\{(\tau_P\alpha)[-1]:\; \tau_P \in T_{H}, \alpha \in A\}\\\nonumber 
	&\xlongequal{\text{by Remark~\ref{rem:invo_of_EFF_and_isogeny_is_endo} {\rm (ii)}}}\{\tau_P([-1]\alpha):\; P\in H, \alpha \in A\}\\\nonumber 
    &\xlongequal{\text{by Lemma~\ref{lem:groupstructure_of_Aut_of_EFF}}}\{[-1]\tau_{[-1]P}\alpha:\; P\in H, \alpha \in A\}\\\nonumber 
    &=\{[-1]\tau_{P}\alpha:\; P\in H, \alpha \in A\}\\\nonumber 
    &=\{[-1]\tau_{P}\alpha:\; \tau_P\in T_{H}, \alpha \in A\}\\  \label{eq:<[-1]>commutewithgroup}
    &=\{[-1]\sigma:\; \sigma\in G\}.  
\end{align}
Note that $2\nmid aN(E)$ implies $2\nmid ah=r+2$. By \eqref{eq:negativeeq}, there exists a rational place $P_{i,j_0}$ for some $1\leq j_0\leq r+2$ such that $P_{i,j_0}=[-1]P_{i,j_0}$, i.e. $[2]P_{i,j_0}=O$, which implies $P_{i,j_0}=[N(E)+1]P_{i,j_0}=[\frac{N(E)+1}{2}][2]P_{i,j_0}=[\frac{N(E)+1}{2}]O=O$. This leads to a contradiction since $e(P_{i,j_0}|P_{i,j_0}\cap F)=1$ while $e(O|O\cap F)>1$. Thus, these $\ell$ sets of rational places $\{P_{1,1},\dots,P_{1,r+2}\}$, \dots, $\{P_{\ell,1},\dots,P_{\ell,r+2}\}$ satisfy the condition \eqref{eq:tempcondition2}. 
 
(2) When $2\mid aN(E)$, we start to select sets of rational places $\{P_{1,1},\dots,P_{1,r+2}\}$, \dots, $\{P_{\ell,1},\dots,P_{\ell,r+2}\}$ of $E$ that satisfy the condition \eqref{eq:tempcondition2}. We briefly outline our method before proceeding. We first define $G':=G\langle[-1]\rangle$, a subgroup of order $2ah=2(r+2)$ of $\Aut(E/\F_q)$. We then consider the orbits of rational places of $E$ under the action of $G'$ that have length $2(r+2)$. Every such orbit is divided into two orbits under the action of $G$ that have length $r+2$. These $G$-orbits of length $(r+2)$ are finally selected as $\{P_{1,1},\dots,P_{1,r+2}\}$,\dots,$\{P_{\ell,1},\dots,P_{\ell,r+2}\}$ that satisfy the condition \eqref{eq:tempcondition2}. 

Proceeding to the proof. Let $G':=G\langle[-1]\rangle$. It is a subgroup of $\Aut(E/\F_q)$ of order $2(r+2)$ since $\langle[-1]\rangle\cap G=\{\id\}$ and $\langle[-1]\rangle G=G\langle[-1]\rangle$ by \eqref{eq:<[-1]>commutewithgroup}. Let $E^{G'}:=\{u
\in E:\; \sigma (u)=u \text{ for all }\sigma \in G'\}$. The field $E^{G'}$ is a subfield of the rational function field $F=E^{G}$ (see Proposition~\ref{prop:equivalentconditionforEFF_1} {\rm (i)}), and thus is also a rational function field by L\"uroth's Theorem (see \cite[Proposition 3.5.9]{stichtenoth2009algebraic}).
By the Hurwitz genus formula (see \eqref{eq:kurwitzformulacoro_1}), we have 
	\begin{align}\label{eq:20250522_1016}
		0=2g(E)-2=(2g(E^{G'})-2)[E:E^{G'}]+\deg\Diff(E/E^{G'})=-2(2ah)+\deg\Diff(E/E^{G'}).
	\end{align}  
By Lemma~\ref{lem:thepropertyofGalExtension_1}, we have $\Con_{E/E^{G'}}(O\cap E^{G'})=\sum_{P \in H}2aP$. Let $R_1,\dots,R_v$ denote all the distinct rational places of $E$ outside $H$ that are ramified in $E/E^{G'}$. Then we have $\deg\Diff(E/E^{G'})\geq \deg(\sum_{P\in H}(2a-1)P+\sum_{i=1}^{v}(2-1)R_i)\geq (2a-1)h+v$ by Dedekind's different theorem (see \eqref{eq:Dedekinddifferentthm}). 
Thus, it holds $4ah=\deg\Diff(E/E^{G'})\geq (2a-1)h+v$ by \eqref{eq:20250522_1016}, which implies $v\leq 2ah+h$. Including those rational places in $H$, there are at most $2ah+2h$ rational places of $E$ that are ramified in $E/E^{G'}$. 
Hence, by Remark~\ref{rem:thepropertyofGalExtension_2} {\rm (ii)}, there exist at least $l=\lceil\frac{N(E)-2ah-2h}{2ah}\rceil$ distinct rational places of $E^{G'}$ that split completely in $E/E^{G'}$ into $\{T_{1,1},\dots,T_{1,2(r+2)}\}, \dots, \{T_{l,1},\dots,T_{l,2(r+2)}\}\subseteq \Pbb_E^1$, respectively. 
Note that $|G'|=2(r+2)$ and $G'=G\langle[-1]\rangle=\langle[-1]\rangle G.$
Since $G'$ acts transitively on $\{T_{i,1},\dots,T_{i,2(r+2)}\}$ for each $1\leq i\leq l$ by Lemma~\ref{lem:thepropertyofGalExtension_1} {\rm (i)}, we have 
\begin{align}
	 \{T_{i,1},\dots,T_{i,2(r+2)}\}=\{\sigma(T_{i,1}):\; \sigma\in G'\}     
  &=\{\sigma(T_{i,1}):\; \sigma\in G\}\sqcup \{[-1](\sigma(T_{i,1})):\; \sigma\in G\}\label{eq:202506112325}\\
	 &=\{\sigma(T_{i,1}):\; \sigma\in G\}\sqcup \{\sigma([-1]T_{i,1}):\; \sigma\in G\}\label{eq:divisionoforbit1}\\
	 &\xlongequal{\text{denote by}} \{P_{2i-1,1},\dots,P_{2i-1,r+2}\}\sqcup \{P_{2i,1},\dots,P_{2i,r+2}\} \label{eq:202506112311}\\  
	 &= \{P_{2i-1,1},\dots,P_{2i-1,r+2}\}\sqcup \{[-1]P_{2i-1,1},\dots,[-1]P_{2i-1,r+2}\},\label{eq:divisionoforbit2}
\end{align} 
where ``$\sqcup$'' refers to the union without intersection, and \eqref{eq:divisionoforbit2} follows from \eqref{eq:202506112325}.
By \eqref{eq:divisionoforbit1} and \eqref{eq:202506112311}, we get $\ell=2l=2\lceil\frac{N(E)-2ah-2h}{2ah}\rceil$ pairwise disjoint $G$-orbits: $\{P_{1,1},\dots,P_{1,r+2}\},\dots,\{P_{\ell,1},\dots,P_{\ell,r+2}\}$, that lie over $\ell$ distinct rational places of $F$ other than $Q_{\infty}=O\cap F$, respectively. Moreover, $\{P_{1,1},\dots,P_{1,r+2}\},\dots,\{P_{\ell,1},\dots,P_{\ell,r+2}\}$ do satisfy the condition \eqref{eq:tempcondition2}; otherwise, for some $1\leq i\leq l$, the union in \eqref{eq:202506112311} would not be a disjoint union by \eqref{eq:divisionoforbit2}.
The proof is complete.
\end{proof}  
Based on Theorem~\ref{thm:estimation_of_length_EFF_and_cons_2}, we have the following corollary, which can be viewed as an extension of \cite[Theorem 4.8]{ma2023group}. 
\begin{corollary}\label{cor:Cons_via_EFF_2.1}
There exist an optimal $(r=3h-2,\delta=3)$-LRC and an optimal $(r=2,\delta=3h-1)$-LRC with parameters $[m\cdot 3h, tr+1,(m-t)\cdot 3h]_q$ for any integers $t$ and $m$ satisfying $1\leq t<m\leq \ell$, where the field size $q$ and integers $h$, $\ell$ take one of the following cases.
\begin{itemize}
\item[{\rm (i)}] $q=2^{2s}$ for a positive integer $s$, $h=h_0^2$ for a positive divisor $h_0$ of $\sqrt{q}+1$, and $\ell=\ceil{\frac{q+2\sqrt{q}+1-3h}{3h}}$;

\item[{\rm (ii)}] $q=3^{2s}$ for a positive integer $s$,  $h=h_0^2$ for a positive divisor $h_0$ of $\sqrt{q}+1$, and $\ell=2\ceil{\frac{q+2\sqrt{q}+1-8h}{6h}}$;

\item[{\rm (iii)}] $q=p^{2s}$ for an odd prime $p \equiv 2\pmod{3}$ and a positive integer $s$, $h=h_0^2$ for a positive divisor $h_0$ of $\sqrt{q}+1$, and $\ell=2\ceil{\frac{q+2\sqrt{q}+1-8h}{6h}}$.
\end{itemize}
\end{corollary}
\begin{proof}
   If $E/\F_q$ is a maximal elliptic function field, then the group structure of $\Pbb_E^1$ is 
\begin{equation}\label{eq:202507011500}
    \Pbb_E^1\cong \Zbb/(\sqrt{q}+1)\Zbb \times \Zbb/(\sqrt{q}+1)\Zbb
\end{equation}
by Lemma~\ref{lem:group_structure_of_EFF}. Therefore, for any positive divisor $h_0$ of $\sqrt{q}+1$, there exists a subgroup of $\Pbb_E^1$ of order $h=h_0^2$ defined by $H:=\braces{P\in \Pbb_E^1:\; \brackets{h_0}P=O}$, which corresponds via the isomorphism in \eqref{eq:202507011500} to the subgroup $<\overline{{(\sqrt{q}+1)}/{h_0}}>\times <\overline{{(\sqrt{q}+1)}/{h_0}}>$ of $\Zbb/(\sqrt{q}+1)\Zbb \times \Zbb/(\sqrt{q}+1)\Zbb$. Let $A$ be a subgroup of $\Aut(E, O)$ of order $a$. Then for any $\sigma\in A$ and $P\in H=\braces{P\in \Pbb_E^1:\; \brackets{h_0}P=O}$, we have $[h_0](\sigma(P))=\sigma([h_0]P)=\sigma(O)=O$, where the first ``$=$'' follows from Remark~\ref{rem:invo_of_EFF_and_isogeny_is_endo} {\rm (ii)}. Hence $\sigma(P)\in H$. By Lemma~\ref{lem:isTAsubgroupofAutE}, $G:=T_HA$ must be a subgroup of order $ah$ of $\Aut(E/\F_q)$. We now can prove {\rm (i)}, {\rm (ii)}, and {\rm (iii)}. 

{\rm (i)} Let $q = 2^{2s}$ for a positive integer $s$. There exists an explicit maximal elliptic function field $E/\F_q$ with $|\Aut(E, O)| = 24$, as shown in \cite[Lemmas 9 and 15]{li2019optimal}. Let $A$ be a subgroup of $\Aut(E, O)$ of order $a = 3$, which must not contain the elliptic involution $[-1]$ since $[-1]$ is of order $2$. Then there exists a subgroup $T_HA$ of $\Aut(E/\F_q)$ of order $3h$ by the above preceding discussion. Note that $2 \nmid aN(E) = 3(2^{2s} + 2 \cdot 2^s + 1)$. By Theorem~\ref{thm:estimation_of_length_EFF_and_cons_2}, the proof of {\rm (i)} is completed.

{\rm (ii)} Let $q = 3^{2s}$ for a positive integer $s$. There exists an explicit maximal elliptic function field $E/\F_q$ with $|\Aut(E, O)| = 12$, as shown in \cite[Lemmas 10 and 16]{li2019optimal}. An argument analogous to the proof of item {\rm (i)} yields item {\rm (ii)}. The only difference is that, in this case, $2 \mid aN(E)=3(3^{2s}+2 \cdot 3^s+1)$, so we must employ the second estimate for $\ell$ in \eqref{eq:202505271725}.

{\rm (iii)} Let $q = p^{2s}$ for an odd prime $p \equiv 2 \pmod{3}$ and a positive integer $s$. There exists an explicit maximal elliptic function field $E/\F_q$ with $|\Aut(E, O)| = 6$, as shown in \cite[Lemmas 11 and 17]{li2019optimal}. The rest of the proof is the same as above. 
\end{proof}

Based on Theorem~\ref{thm:estimation_of_length_EFF_and_cons_2} and the computation presented in \cite[Section 4.5]{ma2023group}, we can also obtain $q$-ary optimal $(7,3)$-LRCs and $(2,8)$-LRCs with lengths at most $q+2\sqrt{q}-8$, where $q=4^{2s+1}$ for a positive integer $s$. 
  
Let $E=\F_q(x,y)$ be an elliptic function field defined by the equation $y^2+y=x^3$, where $q=4^{2s+1}$ for an arbitrary non-negative integer $s$. From the proof of \cite[Lemma 15]{li2019optimal}, $E/\F_q$ is a maximal elliptic function field.

Let $Q\in \Pbb_E^1$ be the unique common zero of $x$ and $y-1$. Consider the translation-by-$Q$ on the elliptic function field $E$ explicitly given by 
\(\tau_Q: (x\mapsto \frac{y+1}{x^2},\; y \mapsto \frac{y+1}{y})\)
from Group Law Algorithm 2.3 in \cite{silverman2009arithmetic}. The order of $\tau_Q$ is $3$, since 
\[x\mapsto \frac{y+1}{x^2} \mapsto \frac{x}{y+1} \mapsto x \text{ and } y \mapsto \frac{y+1}{y} \mapsto \frac{1}{y+1}  \mapsto y.\]
Define $\sigma\in \Aut(E,O)$ of order $3$ by $\sigma: (x\mapsto u^2x,\; y\mapsto y)$, where $u$ is a primitive third root of unity in $\F_q$. Let $A:=<\sigma>, H:=<Q>$; then $a=|A|=3, h=|H|=3$. 
Note that $\sigma(Q)=Q$, thus $G:=T_HA$ is a subgroup of $\Aut(E/\F_q)$ by Lemma~\ref{lem:isTAsubgroupofAutE}. Note that $2\nmid aN(E)=3\cdot (2^{2s+1}+1)^2$.
By Theorem~\ref{thm:estimation_of_length_EFF_and_cons_2}, we have the following corollary.
\begin{corollary}\label{cor:Cons_via_EFF_2.2_(7,3)}
 Let $q=4^{2s+1}$ for a positive integer $s$.
 Then there exist an optimal $(r=7,\delta=3)$-LRC and an optimal $(r=2,\delta=8)$-LRC with parameters 
 $[9m,tr+1,9(m-t)]_q$
 for any $1\leq t< m\leq \ell=\lceil\frac{q+2\sqrt{q}-8}{9}\rceil=\frac{q+2\sqrt{q}-8}{9}$.
\end{corollary} 
\section{Constructions of Optimal $(r,\delta)$-LRCs via Automorphism Groups of Hyperelliptic Function Fields}\label{sec:4}
In this section, we investigate the construction of optimal $(r,\delta)$-LRCs using hyperelliptic function fields. We first introduce a general framework for constructing optimal $(r,3)$-LRCs via automorphism groups of hyperelliptic function fields of genus 2, and then apply it to obtain explicit optimal $(4,3)$-LRCs with lengths slightly below $q+4\sqrt{q}$. In the final subsection, we present the construction of optimal $(g+1-g',g+1+g')$-LRCs ($0\leq g' \leq g-1$) with lengths up to $q+2g\sqrt{q}$ by employing specific hyperelliptic function fields of genus $g \geq 2$.  
\subsection{A Framework for Constructing Optimal $(r,3)$-LRCs via Automorphism Groups of Hyperelliptic Function Fields of Genus $2$ }\label{sec:4.1}
In the following, we introduce the general framework.
\begin{proposition}\label{prop:sufficientconditionfor_(r,3)_viaHEFF_1}
	Let $E/\F_q$ be a hyperelliptic function field of genus $2$ defined by $y^2=f(x)$ with $\deg(f)=5$  and $2\nmid \Char(\F_q)$. Suppose that $G$ is a subgroup of $\Aut(E/\F_q)$. Let $|G|=r+2$ with $r\geq 2$. Let $F:=E^G$, and let $[P_1,P_2,\dots,P_{r+2}]$ be a list\footnote{Here, a \Emph{list} refers to an ordered multiset of rational places of $E$, or equivalently, a finite sequence of (not necessarily distinct) rational places of $E$.} of rational places of $E$ such that $\sum_{j=1}^{r+2}P_j=\Con_{E/F}(Q_{\infty})$ for a rational place $Q_{\infty}$ of $F$. Assume $\deg\Diff(E/F)>2$ (or equivalently, $F$ is a rational function field, by Lemma \ref{lem:is_rational_FF}).  
	Then the following statements hold.
	\begin{itemize} 
\item[{\rm (i)}] There exists a function $z\in F$ such $F=\F_q(z)$ and $(z)^E_{\infty}=\Con_{E/F}(Q_{\infty})=P_{1}+P_2+\dots+P_{r+2}$.
	
\item[{\rm (ii)}] There exist $r$ functions $w_0=1,w_1,\dots,w_{r-1}\in \Lcal_{E}(P_1+P_2+\dots+P_{r+1})$ that are $F$-linearly independent.
   
\item[{\rm (iii)}] Let $\{P_{i,1},P_{i,2},\dots,P_{i,r+2}\}$ be pairwise distinct rational places of $E$ lying over a rational place $Q_i$ of $F$ for each $1\leq i\leq \ell$, such that $Q_{\infty},Q_1,\dots,Q_{\ell}$ are pairwise distinct. 
Then all $r\times r$ submatrices of the following matrix 
\begin{align}\label{eq:202507161329}
	M_i=\left(\begin{array}{cccc}
		w_0(P_{i,1})&w_0(P_{i,2})&\cdots&w_0(P_{i,r+2})\\
		w_1(P_{i,1})&w_1(P_{i,2})&\cdots&w_1(P_{i,r+2})\\
		\vdots &\vdots &\ddots &\vdots\\
		w_{r-1}(P_{i,1})&w_{r-1}(P_{i,2})&\dots&w_{r-1}(P_{i,r+2})\\
	\end{array}\right)  
\end{align} 
are invertible for each $1\leq i\leq \ell$ if the following conditions $(C1)$ and $(C2)$ are satisfied.

$(C1)$ $P_{r+2}=P_{\infty}$, where $P_{\infty}$ is the unique place at infinity of $E$.

$(C2)$ $\overline{P_{i,j}}\neq P_{i,j'}$ for any $1\leq i\leq \ell$ and $1\leq j<j'\leq r+2$.   
\end{itemize}
\end{proposition} 
\begin{proof}
	{\rm (i)} Note that $F$ is a rational function field. Let $z\in \Lcal_F(Q_{\infty})\backslash\F_q$. We have 
    $
    (z)_{\infty}^F=Q_{\infty}$, $[F:\F_q(z)]=\deg((z)_{\infty}^F)=\deg(Q_{\infty})=1
    $ 
    and
    $
	 (z)^E_{\infty}=\Con_{E/F}((z)_\infty^{F})=\Con_{E/F}(Q_{\infty})=P_{1}+P_2+\dots+P_{r+2}.
    $
	 
	 {\rm (ii)} Note that $\Lcal_{E}(P_1)\subseteq \Lcal_{E}(P_1+P_2)\subseteq \dots \subseteq \Lcal_{E}(P_1+\dots+P_{r+1})$, $\dim_{\F_q}(\Lcal_{E}(P_1))=1$, and $\dim_{\F_q}(\Lcal_{E}(P_1+\dots+P_{r+1}))=r+1+1-2=r$ by the Riemann-Roch theorem (see \eqref{eq:ell(D)_of_deg(D)_geq_2g-1}). Furthermore, $\dim_{\F_q}(\Lcal_E(\sum_{j=1}^{I+1}P_j))-\dim_{\F_q}(\Lcal_E(\sum_{j=1}^{I}P_j))\leq 1$ for each $1\leq I\leq r$ by \cite[Lemma 1.4.8]{stichtenoth2009algebraic}. Therefore, among the $r$ sets $V_I:=\Lcal_{E}(\sum_{j=1}^{I+1}P_j)\backslash \Lcal_{E}(\sum_{j=1}^{I}P_j)$ for $1\leq I\leq r$, exactly $r-1$ of them are non-empty. Let $1\leq \urep_1<\dots<\urep_{r-1}\leq r$ be the $r-1$ indexes such that $V_{\urep_1},\dots,V_{\urep_{r-1}}$ are non-empty. Let $w_0=1$, and let $w_1,\dots,w_{r-1}$ be arbitrary elements in $V_{\urep_1},\dots,V_{\urep_{r-1}}$, respectively.  
	 
 Now we show that $w_0,w_1,\dots,w_{r-1}$ are $F$-linearly independent. Assume towards a contradiction that there exist rational functions $\Frep_{0}(z),\dots,\Frep_{r-1}(z)\in F=\F_q(z)$, not all zero, such that $\sum_{i=0}^{r-1}\Frep_i(z)w_i=0$. By clearing denominators, we may assume that $\Frep_{0}(z),\dots,\Frep_{r-1}(z)$ are polynomials of $z$. We define $t:=\max\{0\leq i\leq r-1:\;\deg(\Frep_i(z))=\max\{\deg(\Frep_j(z)):\; 0\leq j\leq r-1\}\}$ ($\deg(0):=-\infty$). It holds that $\Frep_t(z)\neq 0$. Then we consider the following two possible cases to derive a contradiction. Before proceeding, we denote $e:=e(P_1|Q_{\infty})=\dots=e(P_{r+2}|Q_{\infty})$, and observe that $v_{P_1}(z)=\dots=v_{P_{r+2}}(z)=-e$.
\begin{itemize}
\item  If $t=0$, then we have $\deg(\Frep_{0}(z))>\deg(\Frep_i(z))$ for any $1\leq i\leq r-1$. 
Therefore, we have $$v_{P_{r+2}}(\Frep_{0}(z)w_0)=-e\deg(\Frep_{0}(z))<-e\deg(\Frep_i(z))-(e-1)\leq -e\deg(\Frep_i(z))+v_{P_{r+2}}(w_i)=v_{P_{r+2}}(\Frep_i(z)w_i)$$ for any $1\leq i\leq r-1$ such that $\Frep_i(z)\neq 0$, where the symbol ``$\leq$'' in the above inequality is due to $w_i\in \Lcal_{E}(P_1+\dots+P_{r+2}-P_{r+2})=\Lcal_{E}(\sum_{P|Q_{\infty}}eP-P_{r+2})$. This implies $v_{P_{r+2}}(\sum_{i=0}^{r-1}\Frep_i(z)w_i)=v_{P_{r+2}}(\Frep_{0}(z)w_0)\neq \infty$.

\item  If $1\leq t\leq r-1$, then we have $\deg(\Frep_t(z))\geq \deg(\Frep_i(z))$ for any $0\leq i\leq t-1$; and $\deg(\Frep_t(z))>\deg(\Frep_i(z))$ for any $t+1\leq i\leq r-1$. Therefore, for any $0\leq i\leq t-1$ such that $\Frep_i(z)\neq 0$, we have
\begin{align}\label{eq:202507012157}
v_{P_{\urep_t+1}}(\Frep_t(z)w_t)=-e\deg(\Frep_t(z))+v_{P_{\urep_t+1}}(w_t)<-e\deg(\Frep_i(z))+v_{P_{\urep_t+1}}(w_i)= v_{P_{\urep_t+1}}(\Frep_i(z)w_i)
\end{align}
	 where ``$<$'' is due to $\deg(\Frep_t(z))\geq \deg(\Frep_i(z))$ and $v_{P_{\urep_t+1}}(w_t)<v_{P_{\urep_t+1}}(w_i)$ since $v_{P_{\urep_t+1}}(w_t)\leq  -v_{P_{\urep_t+1}}(\sum_{j=1}^{\urep_t+1}P_j)$,\footnote{This inequality holds because otherwise we would have $w_{t}\in \Lcal_E(\sum_{j=1}^{\urep_t}P_j)$, contradicting the fact that $w_t\in \Lcal_E(\sum_{j=1}^{\urep_t+1}P_j)\backslash \Lcal_E(\sum_{j=1}^{\urep_t}P_j)$.} $w_i\in \Lcal_{E}(\sum_{j=1}^{\urep_t+1}P_j-P_{\urep_t+1})$; and for any $t+1\leq i\leq r-1$ such that $\Frep_i(z)\neq 0$, we have
     \begin{align}\label{eq:202507012158}
     v_{P_{\urep_t+1}}(\Frep_t(z)w_t)=-e\deg(\Frep_t(z))+v_{P_{\urep_t+1}}(w_t)<-e\deg(\Frep_i(z))+v_{P_{\urep_t+1}}(w_i)=v_{P_{\urep_t+1}}(\Frep_i(z)w_i)
     \end{align} 
	 where the symbol ``$<$'' is due to $\deg(\Frep_t(z))> \deg(\Frep_i(z))$, $v_{P_{\urep_t+1}}(w_t)\leq -1$ and $v_{P_{\urep_t+1}}(w_i)\geq -e$ since $v_{P_{\urep_t+1}}(w_t)\leq -v_{P_{\urep_t+1}}(\sum_{j=1}^{\urep_t+1}P_j)$, $w_i\in \Lcal_{E}(\sum_{j=1}^{r+2}P_j)=\Lcal_{E}(\sum_{P|Q_{\infty}}eP)$. Inequalities~\eqref{eq:202507012157}~and~\eqref{eq:202507012158} imply $v_{P_{\urep_t+1}}(\sum_{i=0}^{r-1}\Frep_i(z)w_i)=v_{P_{\urep_t+1}}(\Frep_t(z)w_t)\neq \infty$.  
\end{itemize} 
	 In both cases, we arrive at a contradiction to the assumption $\sum_{i=0}^{r-1}\Frep_i(z)w_i=0$. 
	 Therefore, $w_0,w_1,\dots,w_{r-1}$ are $F$-linearly independent.
	 
	 {\rm (iii)}
    Assume towards a contradiction that the submatrix consisting of the $u_1,\dots,u_r$-th columns of
	the matrix $M_i$ in \eqref{eq:202507161329} is singular for some $1\leq \irep\leq \ell$ and $1\leq u_1<\dots<u_r\leq r+2$. 
Then there exist $c_0,\dots,c_{r-1}\in\F_q$, not all zero, such that $(c_0,\dots,c_{r-1})M_{\irep}$ vanishes at the positions $u_1,\dots,u_r$.
That is, the function $w:=c_0w_0+\dots+c_{r-1}w_{r-1}$ has zeros $P_{\irep,u_1},\dots,P_{\irep,u_r}$. Note that $w \in\Lcal_E(\sum_{j=1}^{r+1}P_j)$. The principal divisor of $w$ must be of the form
\begin{align}\label{eq:temp2}
    \parentheses{w}^E=P-\parentheses{\sum_{j=1}^{r+1}P_{j}}+\sum_{j=1}^{r}P_{\irep,u_j}, 
\end{align}
where $P\in \Pbb_E^1$ is an unknown rational place that will be discussed later.
    Note that $\Lcal_{F}(Q_{\infty}-Q_{\irep})\neq \{0\}$ since $F$ is a rational function field. Let $z'$ be a nonzero element in $\Lcal_F({Q_\infty}-Q_{\irep})$, then we have 
    \begin{align}\label{eq:temp5}
  (z')^E=\Con_{E/F}((z')^F)=\Con_{E/F}(-Q_{\infty}+Q_{\irep})=(-P_1-\dots-P_{r+2})+(P_{\irep,1}+\dots+P_{\irep,r+2}).
    \end{align} 
By \eqref{eq:temp2} and \eqref{eq:temp5}, we have 
\begin{align}\label{eq:thetransformationofdivisor3}
    \parentheses{\frac{w}{z'}}^E=P+P_{r+2}-\sum_{u\in [r+2]\backslash\{u_1,\dots,u_r\}}P_{\irep,u}.
\end{align}
We now consider all possible $P\in \Pbb_E^1$, divided into two cases.
\begin{itemize}
\item $P=P_{\irep,u'}$ for some $u'\in [r+2]\backslash\{u_1,\dots,u_r\}$. By  \eqref{eq:thetransformationofdivisor3}, we have 
$$\parentheses{\frac{w}{z'}}^E=P_{r+2}-\sum_{u\in [r+2]\backslash\{u_1,\dots,u_r,u'\}}P_{\irep,u}.$$ 
Then we have $\brackets{E:\F_q\parentheses{\frac{w}{z'}}}=\deg\parentheses{\parentheses{\frac{w}{z'}}^E_{0}}=\deg(P_{r+2})=1$. This contradicts the fact that $E$ is a hyperelliptic function field rather than a rational function field. 

  \item $P$ is a rational place with $P\neq P_{\irep,u'}$ for any $u'\in [r+2]\backslash\{u_1,\dots,u_r\}$.
 By  \eqref{eq:thetransformationofdivisor3}, we have 
 $$\parentheses{\frac{w}{z'}}^E=P+P_{r+2}-\sum_{u\in [r+2]\backslash\{u_1,\dots,u_r\}}P_{\irep,u}.$$
The above equation leads to $P+P_{r+2}\sim\sum_{u\in [r+2]\backslash\{u_1,\dots,u_r\}}P_{\irep,u}$, which is ridiculous by Lemma~\ref{lem:Cl^0ofHyperellipticFF}, along with conditions $(C1)$ and $(C2)$. 
\end{itemize}
Both cases lead to a contradiction, thereby completing the proof of {\rm (iii)}.
\end{proof}




\begin{proposition}\label{prop:construction_of_(r,3)from_HEFF_1}
	We adopt the settings in Proposition~\ref{prop:sufficientconditionfor_(r,3)_viaHEFF_1} and assume that conditions $(C1)$ and $(C2)$ are satisfied. 
	Let $t,m$ be integers such that $1\leq t< m\leq \ell$. Let $V:=\{a_{0,t}w_0z^t+\sum_{i=0}^{r-1}\sum_{j=0}^{t-1}a_{i,j}w_iz^j:a_{0,t}\in \F_q \text{ and }a_{i,j}\in \F_q \text{ for } 0\leq i\leq r-1, 0\leq j\leq t-1\}$, and $\Pcal:=\{P_{1,1},\dots,P_{1,r+2},\dots,P_{m,1},\dots,P_{m,r+2}\}$. Define a linear code $\Ccal(\Pcal,V)$ by  
	 \begin{align*}
	   \Ccal(\Pcal,V):=\left\{(\phi(P_{1,1}),\dots,\phi(P_{1,r+2}),\dots,\phi(P_{m,1}),\dots,\phi(P_{m,r+2})):\; \phi\in V\right\}.
	 \end{align*}
Then $\Ccal(\Pcal,V)$ is an optimal $(r=|G|-2,\delta=3)$-LRC with parameters $[m(r+2),tr+1,(m-t)(r+2)]_q.$  
\end{proposition}
\begin{proof}
 The proof is similar to that of Proposition~\ref{prop:construction_of_(r,delta)from_EFF_1}. So we omit it.
\end{proof}
  
 Using Propositions~\ref{prop:sufficientconditionfor_(r,3)_viaHEFF_1} and~\ref{prop:construction_of_(r,3)from_HEFF_1} with some explicit hyperelliptic curves of genus $2$, one can construct optimal $(3,3)$-LRCs with lengths approaching $q+4\sqrt{q}$, which is omitted here since it is subsumed by Theorem~\ref{thm:Cons_via_HEFF_(g+1-g',g+1+g')_p=2g+1_and_p_neq_2g+1} later in Section~\ref{sec:4.3}. 
  

\subsection{Construction of Optimal $(4,3)$-LRCs via Automorphism Groups of Hyperelliptic Function Fields of Genus $2$}\label{sec:4.2}
In this subsection, we present constructions of optimal $(4,3)$-LRCs with lengths slightly below $q+4\sqrt{q}$.
The following theorem provides a sufficient condition, under which we can select rational places $[P_1,\dots,P_{r+2}],\{P_{1,1},\allowbreak \dots,\allowbreak P_{1,r+2}\},\dots,\{P_{\ell,1},\dots,\allowbreak P_{\ell,r+2}\}$ that satisfy the conditions in Proposition~\ref{prop:sufficientconditionfor_(r,3)_viaHEFF_1}, including conditions $(C1)$ and $(C2)$. Before proceeding, we recall that $\iota$ denotes the hyperelliptic involution.
\begin{theorem}\label{thm:estimation_of_length_HEFF_and_cons_1}  
Let $E/\F_q$ be a hyperelliptic function field defined by $y^2=f(x)$ with $N(E)$ rational places, where $\deg(f)=5$ and $2\nmid \Char(\F_q)$. Let $G\leq \Aut(E/\F_q)$ with $|G|=r+2$, and let $F:=E^{G}$.
Assume that $\iota \notin G$, $|G|\geq 5$, and $|G| \nmid N(E)$. 
Then $\deg\Diff(E/F)>2$, and there exist rational places $[P_{1},\dots,P_{r+2}], \{P_{1,1},\dots,P_{1,r+2}\}, \dots, \{P_{\ell,1},\allowbreak \dots,P_{\ell,r+2}\}$ of $E$ satisfying the conditions in Proposition~\ref{prop:sufficientconditionfor_(r,3)_viaHEFF_1} (including conditions $(C1)$ and $(C2)$), where 
$
 \ell=2\ceil{\frac{N(E)-4|G|-2}{2|G|}}-1.
$

Consequently, there exists an optimal $(r=|G|-2,\delta=3)$-LRC with parameters 
$
[m|G|,tr+1,(m-t)|G|]_q
$ 
 for any integers $t,m$ satisfying $1\leq t< m\leq \ell$, by Propositions~\ref{prop:sufficientconditionfor_(r,3)_viaHEFF_1} and~\ref{prop:construction_of_(r,3)from_HEFF_1}.
\end{theorem} 
\begin{proof}
Since $|G| \nmid N(E)$, there exists a rational place $P'$ of $E$ that is ramified in $E/F$; otherwise, all rational places of $E$ would be unramified in $E/F$, and it would follow that $|G|=[E:F] \mid N(E)$ by Lemma~\ref{lem:thepropertyofGalExtension_1}~{\rm(ii)} and {\rm (iii)}.  
By Lemma~\ref{lem:thepropertyofGalExtension_1} {\rm (ii)}, we have $e(P|P'\cap F)-1>0$ for each $P\in \Pbb_E$ lying over $P'\cap F$, and thus
$$
\deg\parentheses{\sum_{P|P'\cap F}\parentheses{e(P|P'\cap F)-1}P}
\geq \frac{1}{2}\deg\parentheses{\sum_{P|P'\cap F}e(P|P'\cap F)P}
=\frac{1}{2}\deg\parentheses{\Con_{E/F}(P'\cap F)}
=\frac{1}{2}|G|
\geq \frac{5}{2}
>2,
$$ 
which implies $\deg\Diff(E/F)>2$ by Dedekind's different theorem (see \eqref{eq:Dedekinddifferentthm}). Thus, $F$ is a rational function field by Lemma~\ref{lem:is_rational_FF}.

Let $[P_1,\dots,P_{r+2}]$ be a list of rational places of $E$ such that $P_{r+2}{=}P_{\infty}$ and $\sum_{j=1}^{r+2}P_j{=}\Con_{E/F}(P_{\infty}\cap F)$ (i.e., letting $Q_{\infty}=P_{\infty}\cap F$ in Proposition~\ref{prop:sufficientconditionfor_(r,3)_viaHEFF_1}), then the condition $(C1)$ is satisfied. 
As for the selection of $\{P_{1,1},\dots,P_{1,r+2}\},\dots,\{P_{\ell,1},\dots,\allowbreak P_{\ell,r+2}\}$ satifying $(C2)$,
  we use a similar method as that used in the proof of Theorem~\ref{thm:estimation_of_length_EFF_and_cons_2} for the case $2\mid aN(E)$. Recall that the hyperelliptic involution $\iota$ is of order $2$ and commutes with all elements of $\Aut(E/\F_q)$ (see Remark~\ref{rem:hyperelliptic_involution}). Let $G':=<\iota>G=\{\sigma:\sigma \in G\}\cup\{\iota\sigma:\sigma \in G\}$ be a new larger subgroup of order $2|G|=2(r+2)$ of $\Aut(E/\F_q)$. We now consider the function field extension $E/E^{G'}$. Since $F=E^{G}$ is a rational function field, its subfield $E^{G'}$ is also a rational function field by L\"uroth's Theorem (see \cite[Proposition 3.5.9]{stichtenoth2009algebraic}). By the Hurwitz genus formula (see \eqref{eq:kurwitzformulacoro_1}), we have 
\begin{align*}
 2=2g(E)-2=(2g(E^{G'})-2)[E:E^{G'}]+\deg\Diff(E/E^{G'})=-4|G|+\deg\Diff(E/E^{G'}).
\end{align*}
Thus, by Dedekind's different theorem (see \eqref{eq:Dedekinddifferentthm}), there are at most $4|G|+2$ rational places of $E$ that are ramified in $E/E^{G'}$. Therefore, by Remark~\ref{rem:thepropertyofGalExtension_2} {\rm (ii)}, there are at least $l=\ceil{\frac{N(E)-4|G|-2}{2|G|}}$ rational places of $E^{G'}$ that split completely in $E/E^{G'}$ into
$\{T_{1,1},\dots,T_{1,2(r+2)}\}\allowbreak,\dots,\{T_{l,1},\dots, T_{l,2(r+2)}\} \subseteq \Pbb_E^1$, respectively. Since $|G'|=2(r+2)$ and $G'$ acts transitively on each of these sets by Lemma~\ref{lem:thepropertyofGalExtension_1} {\rm (i)}, for each $1\leq i\leq l$ we have 
\begin{align} 
	 \{T_{i,1},\dots,T_{i,2(r+2)}\}=\{\sigma(T_{i,1}):\; \sigma\in G'\}
  &=\{\sigma(T_{i,1}):\; \sigma\in G\}\sqcup \{\iota(\sigma( T_{i,1})):\;\sigma\in G\} 
     \label{eq:202506211855}
     \\ 
  \label{eq:divisionoforbit3}  
	 &=\{\sigma(T_{i,1}):\;\sigma\in G\}\sqcup \{\sigma(\iota (T_{i,1})):\; \sigma\in G\} 
     \\ 
  \label{eq:202506211843} 
	 &\xlongequal{\text{denote by}} \{P_{2i-1,1},\dots,P_{2i-1,r+2}\} \sqcup \{P_{2i,1},\dots,P_{2i,r+2}\}
     \\\label{eq:divisionoforbit4}
	 &= \{P_{2i-1,1},\dots,P_{2i-1,r+2}\}\sqcup \{\iota(P_{2i-1,1}),\dots,\iota(P_{2i-1,r+2})\}, 
\end{align} 
where ``$\sqcup$'' denotes the union without intersection, and \eqref{eq:divisionoforbit4} follows from \eqref{eq:202506211855}.
By \eqref{eq:divisionoforbit3} and \eqref{eq:202506211843}, we get at least $\ell{=}2l{-}1=2\ceil{\frac{N(E)-4|G|-2}{2|G|}}-1$ pairwise disjoint $G$-orbits\footnote{To account for the worst-case scenario, we may, without loss of generality, assume that the last $G$-orbit $\{P_{2l,1},\dots,P_{2l,r+2}\}$ lies over $Q_{\infty}=P_{\infty}\cap F$. We then simply discard this orbit and work with the remaining $\ell=2l-1$ orbits that do not lie over $Q_{\infty}$.}: $\{P_{1,1},\dots,P_{1,r+2}\},\dots,\{P_{\ell,1},\dots,P_{\ell,r+2}\}$, that lie over $\ell$ distinct rational places of $F$ other than $Q_{\infty}=P_{\infty}\cap F$, respectively. 
Moreover, $\{P_{1,1},\dots,P_{1,r+2}\},$$\dots,$$\{P_{\ell,1},\allowbreak \dots,\allowbreak P_{\ell,r+2}\}$ do satisfy the condition $(C2)$ in Proposition~\ref{prop:sufficientconditionfor_(r,3)_viaHEFF_1}; otherwise, for some $1\leq i\leq l$, \eqref{eq:202506211843} would not be a disjoint union by \eqref{eq:divisionoforbit4}, leading to a contradiction. 
Based on the above selection of $\{P_{1,1},\dots,P_{1,r+2}\}$, $\{P_{2,1},\dots,P_{2,r+2}\}, \dots, \{P_{\ell,1},\dots,P_{\ell,r+2}\}$, this theorem is proved.
\end{proof} 
Using Theorem~\ref{thm:estimation_of_length_HEFF_and_cons_1}, we obtain optimal $(4,3)$-LRCs over two classes of finite fields, together with an explicit example. 
\begin{corollary}\label{cor:Cons_via_HEFF_(4,3)_p=5}
Let $q$ be a prime power of one of the following forms:
\begin{itemize}
    \item [{\rm(i)}] $q=5^{2s}$ for an odd positive integer $s$;
    \item [{\rm(ii)}] $q=\oq^{2s}$ for an odd positive integer $s$ and a prime power $\oq$ with $\oq\neq 5$ and $\oq\equiv 5,15,21, \text{ or } 23\pmod{24}$.
\end{itemize}
Then for any integers $t,m$ with $1\leq t< m\leq \ell=2\ceil{\frac{q+4\sqrt{q}-25}{12}}-1$, there exists an optimal $(4,3)$-LRC with parameters 
$
    [6m,4t+1, 6(m-t)]_q.
$ 
\end{corollary}  
\begin{proof} 
{\rm (i)} We consider the hyperelliptic curve defined by $y^2=x^5+x$ over $\F_q$, where $q=5^{2s}$ for an odd positive integer $s$. It is a maximal hyperelliptic curve with $q+4\sqrt{q}+1$ rational points by Lemma~\ref{lem:ismaximalhypercurve_1} and Lemma~\ref{lem:maximal_curve_lift}. Let $E/\F_q$ be its function field. 
It has an automorphism $\sigma\in \Aut(E/\F_q)$ defined by the associated matrix 
$\parentheses{\begin{array}{cc}
    \alpha  & -1 \\
    -1 & 0
\end{array}}$ (see \eqref{eq:matrix_resentation_of_HEFF_auto}),
where $\alpha\in \F_{25}\subseteq \F_q$ satisfies $\alpha^2=2$.
Let $G:=<\sigma>\leq \Aut(E/\F_q)$.  
It is direct to verify that $|G|=6$ and $\iota \notin G$. Note that $6\nmid N(E)=q+4\sqrt{q}+1=5^{2s}+4\cdot 5^s+1=(6-1)^{2s}+4\cdot(6-1)^s+1$ since $2\nmid s$.
Then by Theorem~\ref{thm:estimation_of_length_HEFF_and_cons_1}, item {\rm (i)} holds. 

{\rm (ii)} We consider the hyperelliptic curve defined over $\F_q$ by the equation $y^2=x^5+x$, where $q=\oq^{2s}$ for an arbitrary odd prime power $\oq\neq 5$ with $\oq\equiv 5 \text{ or } 7\pmod{8}$ and an odd positive integer $s$.  It is a maximal hyperelliptic curve by   Lemma~\ref{lem:ismaximalhypercurve_1} and Lemma~\ref{lem:maximal_curve_lift}. Let $E/\F_q$ be its function field. 
By Lemma~\ref{lem:autogroupbuhuangandzhao} and some concrete computations, we obtain an automorphism $\sigma_1\in \Aut(E/\F_q)$ of order $2$ defined by the associated matrix 
$\parentheses{\begin{array}{cc}
      0& -1\\
    -1& 0
\end{array}}$, and an automorphism $\sigma_2 \in \Aut(E/\F_q)$ of order $3$ defined by the associated matrix 
$\parentheses{\begin{array}{cc}
    2^{-1}(\alpha^2-1)  & 2^{-1}(\alpha-\alpha^3) \\
    2^{-1}(\alpha^3-\alpha)& 2^{-1}(-\alpha^2-1)
\end{array}}$ (see \eqref{eq:matrix_resentation_of_HEFF_auto}),
where $\alpha=u^{\frac{q-1}{8}}$, with $u$ a primitive element of $\F_{q}$. 
It is worth verifying that $\sigma_2:(x\mapsto \frac{(\alpha^2-1)x+(\alpha-\alpha^3)}{(\alpha^3-\alpha)x+(-\alpha^2-1)},y\mapsto \frac{y}{2^{-3}((\alpha^3-\alpha)x+(-\alpha^2-1))^3})$ is indeed an automorphism in $\Aut(E/\F_q)$. 
To this end, we first check that $(\sigma_2(y))^2=(\sigma_2(x))^5+\sigma_2(x)$, which is equivalent to verifying that 
$$2^6y^2=\big((\alpha^2-1)x+(\alpha-\alpha^3)\big)^5\big((\alpha^3-\alpha)x+(-\alpha^2-1)\big)+\big((\alpha^2-1)x+(\alpha-\alpha^3)\big)\big((\alpha^3-\alpha)x+(-\alpha^2-1)\big)^5.$$
Note that $(1+\alpha^2)^2=2\alpha^2$ since $\alpha^4=-1$. 
Multiplying both sides by $(1+\alpha^2)^6$, the above identity to be verified becomes 
$$
(2\alpha^2)^3\cdot 2^6 y^2=(-2x+2\alpha)^5(-2\alpha x-2\alpha^2)+(-2x+2\alpha)(-2\alpha x-2\alpha^2)^5.
$$
The right hand side equals $(-2x+2\alpha)(-2\alpha x-2\alpha^2)((-2x+2\alpha)^4+\alpha^4(-2 x-2\alpha)^4)=-8\alpha^2\cdot 2^6(x^5+x)$, which is equal to the left hand side. 
Define $\sigma_3$ as $\sigma_3:(x\mapsto \frac{(-\alpha^2-1)x+(\alpha^3-\alpha)}{(\alpha-\alpha^3)x+(\alpha^2-1)},y\mapsto \frac{y}{2^{-3}((\alpha-\alpha^3)x+(\alpha^2-1))^3})$. We can similarly verify that $(\sigma_3(y))^2=(\sigma_3(x))^5+\sigma_3(x)$, and that $\sigma_3$ is the inverse of $\sigma_2$. Hence, $\sigma_2$ is indeed an element of $\Aut(E/\F_q)$.

Let $G:=<\sigma_1,\sigma_2>$. It is a dihedral subgroup of $\Aut(E/\F_q)$ of order $6$, satisfying the relation $\sigma_1\sigma_2\sigma_1=\sigma_2^{-1}$, and it does not contain the hyperelliptic involution $\iota$. To apply Theorem~\ref{thm:estimation_of_length_HEFF_and_cons_1}, we examine under what conditions $N(E)=q+4\sqrt{q}+1$ is not divisible by 6.
\begin{itemize}
    \item $\oq\equiv 5 \pmod 8$. In this case, we consider three subcases $\oq\equiv 5,13,21 \pmod{ 24}$. In these three subcases, we have $N(E)=q+4\sqrt{q}+1=\oq^{2s}+4\oq^s+1\equiv -2,0,4 \pmod{6}$, respectively, where $s$ is an odd positive integer.  
    Thus, when the condition ``$\oq\equiv 5 \pmod 8$'' is strengthened to ``$\oq\equiv 5\text{ or }21 \pmod {24}$'', we have $6\nmid N(E)$.
    \item $\oq\equiv 7 \pmod{8}$. In this case, we consider three subcases $\oq\equiv 7,15,23 \pmod{24}$. In these three subcases, we have $N(E)=q+4\sqrt{q}+1=\oq^{2s}+4\oq^s+1\equiv 0,4,-2\pmod{6}$, respectively, where $s$ is an odd positive integer. Thus, when the condition ``$\oq\equiv 7 \pmod 8$'' is strengthened to ``$\oq\equiv 15\text{ or } 23 \pmod{24}$'', we have $6\nmid N(E)$.
\end{itemize}
Based on the above discussion, the proof is complete by Theorem~\ref{thm:estimation_of_length_HEFF_and_cons_1}.
\end{proof}
Indeed, the above $\ell=2\ceil{\frac{q+4\sqrt{q}-25}{12}}-1$ in Corollary~\ref{cor:Cons_via_HEFF_(4,3)_p=5} is just a worst-case estimation on the number of local repair groups,
when it comes to the explicit constructions over concrete finite fields, sometimes the number of local repair groups can be greater than $\ell$, we give an explicit example to illustrate this. In the following example, we present an optimal $(4,3)$-LRCs over $\F_{25}$ with length $36$, which is greater than the worst-case estimation $\ell\cdot 6=(2\ceil{\frac{25+4\sqrt{25}-25}{12}}-1)\cdot 6=18$.

\begin{example}
We consider the hyperelliptic function field $E/\F_{25}$ defined by the equation $y^2=x^5+x$ over $\F_{25}=\F_5(u)$, where $u$ is a primitive element of $\F_{25}$ satisfying the equation $u^2+4u+2=0$. Let $\alpha=u^3$, which satisfies $\alpha^2=2$.
Let $G$ be a subgroup of $\Aut(E/\F_q)$ of order $6$ generated by $\sigma: (x\mapsto -u^3+\frac{1}{x},\; y\mapsto \frac{y}{x^3})$, whose associated matrix is $\begin{pmatrix}
   u^3&-1\\
   -1&0\\
\end{pmatrix}$. Let $r=4,\delta=3$. With the help of the MAGMA calculator \cite{bosma1997magma}, we select $[P_1,\dots,P_{6}]$, $\{P_{1,1},\dots,P_{1,6}\}$, \dots, $\{P_{6,1},\dots,P_{6,6}\}$ satisfying the conditions in Proposition \ref{prop:sufficientconditionfor_(r,3)_viaHEFF_1} (including the conditions $(C1)$ and $(C2)$) as follows:
\begin{align*}
    [P_{1},\dots P_{6}]=[P_{(u^3 , 0)}, P_{(u^{21} , 0)}, P_{(u^9,0)}, P_{(u^{15},0)}, P_{(0,0)}, P_{\infty}], 
\end{align*} 
\renewcommand{\arraystretch}{1}
\begin{align*}
\left(\begin{array}{ccc}
   P_{1,1}& \cdots & P_{1,6} \\
    \vdots&\ddots& \vdots\\ 
   P_{6,1} &\cdots&P_{6,6}
\end{array}\right)=
 \left(\begin{array}{cccccc}  
 P_{(u^{13},2)}&P_{(u^{22},u^{21})}&P_{(u^7,u^3)}&P_{(u^8,2)}&P_{(u^{10},u^{15})}&P_{(4,u^9)}\\
 P_{(u^4,4)}&P_{(u^{14},u^9)}&P_{(1,u^{15})}&P_{(u^2,u^{15})}&P_{(u^5,1)}&P_{(u^{23},u^9)}\\
 P_{(u^4,1)}&P_{(1,u^3)}&P_{(u^5,4)}&P_{(u^2,u^3)}&P_{(u^{23},u^{21})}&P_{(u^{14},u^{21})}\\
 P_{(u^{17},3)}&P_{(u^{11},u^{15})}&P_{(2,2)}&P_{(u,4)}&P_{(u^{19},u^9)}&P_{(3,1)}\\
 P_{(u^{17},2)}&P_{(u^{19},u^{21})}&P_{(u^{11},u^3)}&P_{(u,1)}&P_{(3,4)}&P_{(2,3)}\\
 P_{(u^7,u^{15})}&P_{(u^8,3)}&P_{(u^{22},u^9)}&P_{(4,u^{21})}&P_{(u^{13},3)}&P_{(u^{10},u^3)}\\
\end{array}\right).
\end{align*}
Based on the above selections and Proposition \ref{prop:sufficientconditionfor_(r,3)_viaHEFF_1}, we define the functions $z,w_0,w_1,w_2,w_3$ as 
$$ 
z=(u^5x^3+u^7x^2+u^{15}x+u^5)/y+u^5,   
w_0=1,\\
w_1=u^7y/(x^3+u^{15}x^2+2x+u^{21})+1, $$
$$
w_2=u^9y/(x^3+u^{21}x^2+3x+u^{15})+u^{10}, 
w_3=u^{21}y/(x^3+u^{21}x^2+4x)+2. 
$$
At last, by Proposition \ref{prop:construction_of_(r,3)from_HEFF_1} (taking $t=1,m=\ell=6$), we get the following $5\times 36$ generator matrix, which generates an optimal $(4,3)$-LRC with parameters $[36,5,30]_{25}$. Here, the vertical lines are used to separate the local repair groups. The parameters of this linear code, including its $(4,3)$-locality, are all verified by the MAGMA calculator \cite{bosma1997magma}.
\setlength{\arraycolsep}{1pt}
\begin{align*}\left[
\begin{array}{cccccc|cccccc|cccccc|}
1&1&1&1&1&1&1&1&1&1&1&1&1&1&1&1&1&1\\
3&u^{19}&u^{5}&u^{17}&u^{5}&2&2&u^{8}&u^{17}&4&u^{4}&u^{9}&0&u^{3}&u^{19}&3&u^{8}&u^{9}\\
u^{13}&u^{13}&u^{20}&3&u^{8}&u^{17}&u^{21}&u^{20}&u^{19}&u^{5}&u^{14}&u^{8}&u^{17}&3&u^{5}&u^{14}&u^{15}&4\\
u^{10}&0&u^{17}&u^{9}&u^{23}&u^{16}&u^{5}&u^{2}&u^{17}&3&u^{19}&u^{8}&u^{2}&u^{13}&u^{21}&1&u^{16}&u^{5}\\
0&0&0&0&0&0&4&4&4&4&4&4&u^{16}&u^{16}&u^{16}&u^{16}&u^{16}&u^{16}\\
\end{array}
\right.\sim \qquad\\
\qquad\sim
\left.
\begin{array}{|cccccc|cccccc|cccccc}
1&1&1&1&1&1&1&1&1&1&1&1&1&1&1&1&1&1\\
0&u^{13}&u^{19}&u^{3}&u^{20}&3&2&u^{23}&u^{15}&u^{17}&4&u^{4}&u^{22}&u^{3}&u^{4}&0&4&u^{22}\\
u^{2}&u^{17}&4&u^{11}&3&u^{15}&u^{7}&u^{19}&u^{21}&1&u^{8}&u^{20}&4&u^{19}&u^{3}&u^{21}&u^{3}&u^{15}\\
u^{13}&u^{19}&u^{19}&u^{16}&u^{10}&u&u^{17}&u&u^{21}&u^{8}&u^{10}&u^{21}&u^{13}&u^{23}&4&u^{8}&u&u^{9}\\
u^{7}&u^{7}&u^{7}&u^{7}&u^{7}&u^{7}&u^{15}&u^{15}&u^{15}&u^{15}&u^{15}&u^{15}&u^{11}&u^{11}&u^{11}&u^{11}&u^{11}&u^{11}\\
\end{array}
\right].\end{align*} 
\setlength{\arraycolsep}{5pt}
\end{example}

\subsection{Construction of Optimal $(g+1-g',g+1+g')$-LRCs via Hyperelliptic Function Fields of Genus $g\geq 2$} \label{sec:4.3}
Inspired by \cite[Theorem 21]{huang2025optimal} where either optimal or almost optimal $(4,2)$-LRCs are presented, we consider the construction of optimal $(r,\delta)$-LRCs via hyperelliptic function fields of genus $g\geq 2$, and have the following result. 
\begin{theorem}\label{thm:Cons_via_HEFF_(g+1-g',g+1+g')_p=2g+1_and_p_neq_2g+1}
     Let $g\geq 2$ be an integer. 
    Let $q$ be a prime power of one of the following forms:
\begin{itemize} 
    \item[{\rm (i)}] $q=(2g+1)^{2s}$ for a positive integer $s$ (in this case, $2g+1$ must be a prime power); 
    
    \item[{\rm (ii)}] $q=\oq^{2s}$ for an odd prime power $\oq$ satisfying $\oq\equiv -1\pmod{2g+1}$ and a positive integer $s$.
\end{itemize}
Then there exists an $(r=g+1-g',\delta=g+1+g')$-LRC with parameters 
$$[n=m(2g+1),k=tr+1,d\geq (m-t)(2g+1)+\min\{0,2g'+1\}]_q$$ 
for any integers $g',t,m$ satisfying $-(g-1)\leq g'\leq g-1$ and $1\leq t<m\leq \ell=\floor{\frac{q+2g\sqrt{q}}{2g+1}}$.
In particular, when $0\leq g'\leq  g-1$, it is an optimal $(r=g+1-g',\delta=g+1+g')$-LRC with parameters 
 $[m(2g+1),tr+1,(m-t)(2g+1)]_q.$
\end{theorem}
\begin{proof}
For both {\rm (i)} and {\rm (ii)}, we first consider the case where $s$ is odd. The case where $s$ is even will be handled using a slightly different twisted curve.

{\rm (i)} 
Let $q=(2g+1)^{2s}$ for an odd positive integer $s$.
 We consider the hyperelliptic curve $\Cfrak$ defined over $\F_q$ by the equation $y^2=x^{2g+1}+x$. It is a maximal hyperelliptic curve of genus $g$ with $q+2g\sqrt{q}+1$ rational points by Lemma~\ref{lem:ismaximalhypercurve_1} and Lemma~\ref{lem:maximal_curve_lift}. Let $E/\F_q$ be its function field. 
Since $\F_{(2g+1)^2}\subseteq \F_q$, the equation $\Tr_{(2g+1)^2/(2g+1)}(u)=u^{2g+1}+u=0$ has $2g+1$ distinct roots $\alpha_1,\dots,\alpha_{2g+1}\in \F_{q}$. Let $G:=\{\sigma_i:\; 1\leq i\leq 2g+1\}\subseteq \Aut(E/\F_q)$, where 
$\sigma_i$ is defined by $\sigma_i: (x\mapsto x+\alpha_i,\; y\mapsto y)$.  
Then $G$ is a subgroup of $\Aut(E/\F_q)$ of order $2g+1$. Let $F:=E^{G}$.
For the place at infinity $P_{\infty}$ of $E$, we have $\sigma(P_{\infty})=P_{\infty}$ for any $\sigma \in G$, and thus $e(P_{\infty}|Q_{\infty})=2g+1$ by Lemma~\ref{lem:thepropertyofGalExtension_1}, where $Q_{\infty}:=P_{\infty}\cap F$. By Dedekind's different theorem (see \eqref{eq:Dedekinddifferentthm}), $\deg\Diff(E/F)\geq 2g$, which, along with Lemma \ref{lem:is_rational_FF}, implies that $F$ is a rational function field.   
Take $z\in \Lcal_F(Q_{\infty})\backslash \F_q$. 
We define the vector space $V$ of functions of $E$ that will be used for evaluation by 
\begin{align*} 
V:=\bigg\{a_{0,t}x^0z^t+\sum_{i=0}^{r-1}\sum_{j=0}^{t-1}a_{i,j}x^iz^j:a_{0,t}\in \F_q \text{ and }a_{i,j}\in \F_q \text{ for } 0\leq i\leq r-1, 0\leq j\leq t-1\bigg\}.
\end{align*} 

We now prove that $V$ is of dimension $tr+1$. For this, it suffices to show that $x^0,x^1,\dots,x^{r-1}$ are $F$-linearly independent. Assume towards a contradiction that there exist rational functions $\Frep_{0}(z),\dots,\Frep_{r-1}(z)\in F=\F_q(z)$, not all zero, such that $\sum_{i=0}^{r-1}\Frep_i(z)x^i=0$. By clearing denominators, we may assume that $\Frep_{0}(z),\dots,\Frep_{r-1}(z)$ are polynomials of $z$. Note that $v_{P_{\infty}}(x)=-2$ and $v_{P_{\infty}}(z)=-(2g+1)$ since $(z)^{E}_{\infty}=\Con_{E/F}((z)^{F}_{\infty})=\Con_{E/F}(Q_{\infty})=(2g+1)P_{\infty}$. 
Let $0\leq i_1,i_2\leq r-1$ be two (not necessarily distinct) integers such that $\Frep_{i_1}(z)\neq 0,\Frep_{i_2}(z)\neq 0$, and $v_{P_{\infty}}(\Frep_{i_1}(z)x^{i_1})=v_{P_{\infty}}(\Frep_{i_2}(z)x^{i_2})$. 
Then we have $-(2g+1)\deg(\Frep_{i_1}(z))-2i_1=v_{P_{\infty}}(\Frep_{i_1}(z)x^{i_1})=v_{P_{\infty}}(\Frep_{i_2}(z)x^{i_2})=-(2g+1)\deg(\Frep_{i_2}(z))-2i_2,$ which implies $2i_1\equiv 2i_2 \pmod{2g+1}$, and therefore, $i_1\equiv i_2 \pmod{2g+1}$. Since $0\leq i_1,i_2\leq r-1=g-g'\leq 2g-1$, we have $i_1=i_2$.  
Thus, the valuations $v_{P_{\infty}}(\Frep_i(z)x^i)$ ($0\leq i\leq r-1$) with $\Frep_i(z)\neq 0$ are pairwise distinct, which, together with the strict triangle inequality, leads to a contradiction with the assumption $\sum_{i=0}^{r-1}\Frep_i(z)x^i=0$. Therefore, $x^0,x^1,\dots,x^{r-1}$ are $F$-linearly independent, and then $\dim_{\F_q}(V)=tr+1$.

We then select the rational places that will be used for evaluation.
For any affine rational place $P_{(\alpha,\beta)}$ of $E$, the orbit $G\parentheses{P_{(\alpha,\beta)}}:=\{P_{(\alpha+\alpha_{i},\beta)}:\; 1\leq i\leq 2g+1\}$ is a $G$-orbit of length $2g+1$. Since $E$ has $q+2g\sqrt{q}$ affine rational places, there are totally $\ell=\frac{q+2g\sqrt{q}}{2g+1}$ $G$-orbits of the form $\{P_{(\alpha+\alpha_i,\beta)}: 1\leq i\leq 2g+1\}$. We denote all these $G$-orbits by $\{P_{1,1},\dots,P_{1,2g+1}\}$, \dots, $\{P_{\ell,1},\dots,P_{\ell,2g+1}\}$. The rational places in these $\ell$ orbits lie over $\ell$ distinct rational places of $F$ other than $Q_{\infty}=P_{\infty}\cap F$, respectively. Let $\Pcal:=\{P_{1,1},\dots,P_{1,2g+1},\dots,P_{m,1},\dots,P_{m,2g+1}\}$ be the places for evaluation. 

We define a linear code $\Ccal(\Pcal,V)$ by 
\begin{align*}
\Ccal(\Pcal,V):=\{(\phi(P_{1,1}),\dots,\phi(P_{1,2g+1}),\dots,\phi(P_{m,1}),\dots,\phi(P_{m,2g+1})):\; \phi\in V\}.
\end{align*}
Note that $V\subseteq \Lcal_E(t(2g+1)P_{\infty}+\max\{0,2(r-1)-(2g+1)\}P_{\infty})=\Lcal_E(t(2g+1)P_{\infty}+\max\{0,-2g'-1\}P_{\infty})$ since $(x)_{\infty}^E=2P_{\infty}$ and $(z)_{\infty}^E=(2g+1)P_{\infty}$. By Section \ref{sec:2.1}, the dimension of $\Ccal(\Pcal,V)$ is $k=\dim_{\F_q}(V)=tr+1$ and the minimum distance $d$ of $\Ccal(\Pcal,V)$ satisfies $d\geq m(2g+1)-(t(2g+1)+\max\{0,-2g'-1\})=(m-t)(2g+1)+\min\{0,2g'+1\}$.

To prove that $\Ccal(\Pcal,V)$ is an $(r=g+1-g',\delta=g+1+g')$-LRC, it suffices to note that for any $1\leq i\leq m$, $$\Ccal(\Pcal,V)|_{\{(i-1)(r+\delta-1)+1,\dots,i(r+\delta-1)\}}=\{(\phi(P_{i,1}),\dots,\phi(P_{i,r+\delta-1})):\phi\in \Span_{\F_q}\{1,x,\dots,x^{r-1}\}\},$$ and that $x(P_{i,1}),\dots,x(P_{i,r+\delta-1})$ are pairwise distinct, which imply that $\Ccal(\Pcal,V)|_{\{(i-1)(r+\delta-1)+1,\dots,i(r+\delta-1)\}}$ is a Reed-Solomon code with minimum distance $((r+\delta-1)-(r-1))=\delta$. Invoking the Singleton-type bound~\eqref{eq:SLboundfor_r_delta}, we have $d\leq (m-t)(2g+1)$. Thus, when $0\leq g'\leq g-1$, the minimum distance $d$ is determined to be $(m-t)(2g+1)$ and $\Ccal(\Pcal,V)$ is an optimal $(r=g+1-g',\delta=g+1+g')$-LRC.

As for the case $q=(2g+1)^{2s}$ for an even positive integer $s$, we consider the (twisted) hyperelliptic curve $\Cfrak'/\F_q$ defined by $\gamma y^2=x^{2g+1}+x$, where $\gamma\in \F_q\backslash \{\beta^2:\beta\in \F_q\}$, i.e., a quadratic non-residue in $\F_q$. 
Note that the number of distinct roots in $\F_q$ of $\gamma y^2=\alpha^{2g+1}+\alpha$ and the number of distinct roots in $\F_q$ of $y^2=\alpha^{2g+1}+\alpha$ sum to $2$ for any $\alpha\in \F_q$. Thus, the numbers of affine rational points of $\Cfrak'/\F_q$ and $\Cfrak/\F_q$ (defined by $y^2=x^{2g+1}+x$) sum to $2q$. By Lemmas~\ref{lem:ismaximalhypercurve_1} and~\ref{lem:maximal_curve_lift}, the curve $\Cfrak/\F_q$ is a minimal curve of genus $g$, and then $\Cfrak'/\F_q$ is a maximal hyperelliptic curve of genus $g$. The rest of the proof is similar to the above argument where $s$ is odd, using the curve $\Cfrak'/\F_q$.

{\rm(ii)} Let $q=\oq^{2s}$ for an odd prime power $\oq$ satisfying $\oq\equiv -1\pmod{2g+1}$ and an odd positive integer $s$. We consider the hyperelliptic curve defined over $\F_q$ by $y^2=x^{2g+1}+1$. It is a maximal hyperelliptic curve of genus $g$ by Lemma~\ref{lem:ismaximalhypercurve_2} and Lemma~\ref{lem:maximal_curve_lift}. Let $E/\F_q$ be its function field. 
It has an automorphism $\sigma\in \Aut(E/\F_q)$ defined as $\sigma: (x\mapsto u^{\frac{q-1}{2g+1}}x,\; y\mapsto y)$,
where $u$ is a primitive element of $\F_q$. Let $G:=<\sigma>$, a cyclic group of order $2g+1$, and let $F:=E^{G}$. The rest of the proof is similar to the proof of {\rm (i)}. The only difference worth emphasizing is that the number $\ell$ of $G$-orbits of length $2g+1$ becomes $\ell=\frac{q+2g\sqrt{q}-2}{2g+1}$, since there are two affine rational places $P_{(\alpha,\beta)}$ of $E$ satisfying $\alpha=0$, $P_{(0,1)}$ and $P_{(0,-1)}$.

As for the case $q=\oq^{2s}$ for an odd prime power $\oq$ satisfying $\oq\equiv -1\pmod{2g+1}$ and an even positive integer $s$, we consider the hyperelliptic curve defined over $\F_q$ by $\gamma y^2=x^{2g+1}+1$, where $\gamma$ is a quadratic non-residue in $\F_q$. Arguing similarly to the end of the proof of {\rm(i)}, this is a maximal hyperelliptic curve of genus $g$. The remainder of the proof is similar to the above (the definitions of $E$, $G$, and $F$ are all the same, so we omit them). The only difference worth emphasizing is that the number $\ell$ of $G$-orbits of length $2g+1$ becomes $\ell=\frac{q+2g\sqrt{q}}{2g+1}$, since any affine rational places $P_{(\alpha,\beta)}$ of $E$ satisfies $\alpha\neq 0$.

Note that in all the above four subcases, it holds $\ell=\floor{\frac{q+2g\sqrt{q}}{2g+1}}$. This theorem is proved.
\end{proof} 

\begin{remark}
\phantomsection
\label{rem:Cons_via_HEFF_(g+1-g',g+1+g')}\begin{itemize}

\item[{\rm (i)}] Letting $g=2, g'=-1$ and $2\nmid s$, Theorem~\ref{thm:Cons_via_HEFF_(g+1-g',g+1+g')_p=2g+1_and_p_neq_2g+1} implies \cite[Theorem 21]{huang2025optimal}. 


\item[\rm (ii)] When $g=1$ and $g'=0$, the statements in Theorem~\ref{thm:Cons_via_HEFF_(g+1-g',g+1+g')_p=2g+1_and_p_neq_2g+1} still hold by \cite[Theorem 1]{li2019optimal}. Therefore, Theorem~\ref{thm:Cons_via_HEFF_(g+1-g',g+1+g')_p=2g+1_and_p_neq_2g+1} can be viewed as an extension of \cite[Theorem 1]{li2019optimal}. 
\end{itemize}
\end{remark} 
After completing the above proof of Theorem~\ref{thm:Cons_via_HEFF_(g+1-g',g+1+g')_p=2g+1_and_p_neq_2g+1}, we observe that the function $z\in \Lcal_{F}(Q_{\infty})\backslash \F_q=\Lcal_{F}(P_{\infty}\cap F)\backslash \F_q$ can, in fact, be explicitly chosen as $z=y$. This motivates us to explore further constructions. In the next section, we present optimal $(r,\delta)$-LRCs with even longer code lengths, using some superelliptic curves adapted from the Norm-Trace curves. 


\section{Constructions of Optimal $(r,\delta)$-LRCs via Superelliptic Curves from Norm-Trace Curves} \label{sec:5} 
In this section, we present constructions of optimal $(r,\delta)$-LRCs via some superelliptic curves from Norm-Trace curves, and particularly the Hermitian curves.    
Before presenting them, we recall the definition and some related properties of superelliptic curves and Norm-Trace curves. We refer to \cite{galbraith2002arithmetic} by Galbraith \textit{et al.} and \cite{geil2003codes} by Geil, respectively.
\begin{definition}[{\cite[Definition 1]{galbraith2002arithmetic}}] \label{def:superelliptic}
Let $\F_q$ be a finite field with $q$ elements. Let \( f(x) \in \F_q[x] \) be a monic\footnote{The polynomial $f(x)$ is allowed to be not monic in this paper; see Remark~\ref{rem:202506232112} for details.} polynomial of degree $N$ such that \( \gcd(f(x), f'(x)) = 1 \), where \( f'(x) \) is the formal derivative of \( f(x) \). Let $M$ be a positive integer such that $\gcd(M,N)=1$ and $\gcd(M,\Char(\F_q))=1$. Then the curve 
$
 \Cfrak:\; y^{M} =f(x)
$
is called a \Emph{superelliptic curve}. 
\end{definition}
\begin{lemma}[Part of {\cite[Proposition 2]{galbraith2002arithmetic}}] \label{lem:property_of_superelliptic}
    Let $\Cfrak$ be a superelliptic curve over \( \F_q \) as in Definition~\ref{def:superelliptic}. Then
\begin{itemize}
    \item[\rm (i)] $\Cfrak$ is nonsingular as an affine curve.
    \item[\rm (ii)] There is only one point, \( P_\infty \), at infinity on the normalisation of $\Cfrak$ and this point is defined over \( \F_q \).
    \item[\rm (iii)] The genus of $\Cfrak$ is \( \frac{1}{2}(M-1)(N-1) \).
\end{itemize}
\end{lemma} 
\begin{remark}\label{rem:202506232112}
    The polynomial $f(x)$ in Definition~\ref{def:superelliptic} is required to be monic. However, Lemma~\ref{lem:property_of_superelliptic} remains valid even without this condition. Indeed, assume that $f(x)$ is not monic and satisfies all other conditions in Definition~\ref{def:superelliptic}.
     Since $\gcd(M,N)=1$, we can transform the equation $y^M=f(x)$ into a new equation $y^M=F(x)$ with $F(x)$ monic and satisfy all conditions in Definition~\ref{def:superelliptic} 
    by an invertible polynomial map $\sigma: (x\mapsto \alpha x,\; y\mapsto \beta y)$, where $\alpha,\beta\in \F_q^*$. This is an $\F_q$-isomorphism, which does not affect the properties listed in Lemma~\ref{lem:property_of_superelliptic}. In what follows, we do not require $f(x)$ to be monic, for simplicity.
\end{remark}
Let $E/\F_q$ be the function field of a superelliptic curve $\Cfrak/\F_q$. For an affine rational point $(\alpha,\beta)\in \Cfrak(\F_q)$, we denote its corresponding rational place of $E$ by $P_{(\alpha,\beta)}$, which is the unique common zero of $x-\alpha$ and $y-\beta$. For the point at infinity $P_{\infty}$, we still denote its corresponding rational place of $E$ by $P_{\infty}$, which is the unique common pole of $x$ and $y$. 

Next, we briefly review the Norm-Trace curves. 
\begin{definition}
Let $\oq$ be a prime power and $s$ be a positive integer. The \Emph{Norm-Trace} curve \( \Xcal_{\oq,s} \) over \( \F_{\oq^{s}} \) is defined  by the equation
$
y^{\frac{\oq^{s} - 1}{\oq - 1}} = x^{\oq^{s-1}} + x^{\oq^{s-2}} + \cdots + x^{\oq^0}.
$
 When $s=2$, it is the well-known \Emph{Hermitian curve}.
\end{definition} 
It has exactly one point at infinity, along with $\oq^{2s-1}$ affine rational points. Its genus is
$
\frac{\oq^{s-1} - 1}{2} \left( \frac{\oq^{s} - 1}{\oq - 1} - 1 \right).
$  
\subsection{A General Framework for Constructing Optimal $(r,\delta)$-LRCs via Superelliptic Curves} 
\label{sec:5.1}
We are now ready to introduce our general framework for constructing optimal $(r,\delta)$-LRCs via superelliptic curves. The use of superelliptic curves here is for convenience only; the framework can naturally be extended to a broader class of curves.
\begin{proposition}\label{prop:construction_of_(r,delta)from_SEFF}
Let $\Cfrak$ be a superelliptic curve defined over $\F_q$ by an equation of the form $\gamma y^M=x^N+(\text{lower order terms of } x),$ 
where $\gamma\in \F_q^*$. Let $E/\F_q$ be its function field. Then the following two classes of optimal $(r,\delta)$-LRCs exist under certain assumptions. 

\begin{itemize} 
  \item[{\rm (i)}] Assume $M<N$. Let $r=\floor{\frac{N-1}{M}}+1-b'$ and $\delta=N+1-r\geq 2$ for an integer $0\leq b'\leq \floor{\frac{N-1}{M}}-1$. Assume also that there exist pairwise disjoint sets of affine rational places of $E$: 
$\{P_{1,1},\dots,P_{1,r+\delta-1}\},\dots,\{P_{\ell,1},\dots,P_{\ell,r+\delta-1}\}$, such that for each $1\leq i\leq \ell$, $x(P_{i,1}),\dots,x(P_{i,r+\delta-1})$ are pairwise distinct and $y(P_{i,1})=\dots=y(P_{i,r+\delta-1})$. Then for any $1\leq t<m\leq \ell$, there exists 
an optimal $(r=\floor{\frac{N-1}{M}}+1-b',\delta=N+1-r)$-LRC with parameters
$
   [mN,tr+1,(m-t)N]_{q}.
$ 

\item[{\rm (ii)}] Assume $M>N$. Let $r=\floor{\frac{M-1}{N}}+1-b'$ and $\delta=M+1-r\geq 2$ for an integer $0\leq b'\leq \floor{\frac{M-1}{N}}-1$. Assume also that there exist pairwise disjoint sets of affine rational places of $E$: 
$\{P_{1,1},\dots,P_{1,r+\delta-1}\},\dots,\{P_{\ell,1},\dots,P_{\ell,r+\delta-1}\}$ such that for each $1\leq i\leq \ell$, $y(P_{i,1}),\dots,y(P_{i,r+\delta-1})$ are pairwise distinct and $x(P_{i,1})=\dots=x(P_{i,r+\delta-1})$. Then for any $1\leq t<m\leq \ell$, there exists an optimal $(r=\floor{\frac{M-1}{N}}+1-b',\delta=M+1-r)$-LRC with parameters
$
   [mM,tr+1,(m-t)M]_{q}.
$
\end{itemize}
\end{proposition}
\begin{proof}
{\rm (i)} 
Let $\Pcal:=\{P_{1,1},\dots,P_{1,r+\delta-1},\dots,P_{m,1},\dots,P_{m,r+\delta-1}\}$, and $V:=\{a_{0,t}x^0y^t+\sum_{i=0}^{r-1}\sum_{j=0}^{t-1}a_{i,j}x^iy^j:\; a_{0,t}\in \F_{q}\text{ and }a_{i,j}\in \F_{q} \text{ for } 0\leq i\leq r-1,0\leq j\leq t-1\}$. Define a linear code $\Ccal(\Pcal,V)$ by
\begin{align*} \Ccal(\Pcal,V):=\left\{(\phi(P_{1,1}),\dots,\phi(P_{1,r+\delta-1}),\dots,\phi(P_{m,1}),\dots,\phi(P_{m,r+\delta-1})):\; \phi\in V\right\}.
\end{align*}

First, we prove that $\Ccal(\Pcal,V)$ is an $(r=\floor{\frac{N-1}{M}}+1-b',\delta=N+1-r)$-LRC. For this, it suffices to note that for any $1\leq i\leq m$, $\Ccal(\Pcal,V)|_{\{(i-1)(r+\delta-1)+1,\dots,i(r+\delta-1)\}}=\{(\phi(P_{i,1}),\dots,\phi(P_{i,r+\delta-1})):\; \phi\in \Span_{\F_q}\{x^0,x^1,\dots,x^{r-1}\}\},$ and that $x(P_{i,1}),\dots,x(P_{i,r+\delta-1})$ are pairwise distinct, which imply that $\Ccal(\Pcal,V)|_{\{(i-1)(r+\delta-1)+1,\dots,i(r+\delta-1)\}}$ is a Reed-Solomon code with minimum distance $((r+\delta-1)-(r-1))=\delta$.  

Next, we prove that $V$ has dimension $tr+1$. To this end, it suffices to show that the $tr+1$ functions $x^0y^t,x^iy^j(0\leq i\leq r-1,0\leq j\leq t-1)$, which span $V$, have pairwise distinct valuations at the place at infinity $P_{\infty}$. Note that the functions $x$ and $y$ each have a unique pole $P_{\infty}$. We have $v_{P_{\infty}}(x)=-[E:\F_{q}(x)]=-[\F_q(x,y):\F_{q}(x)]=-M$ and $v_{P_{\infty}}(y)=-[E:\F_{q}(y)]=-[\F_q(x,y):\F_{q}(y)]=-N$ by \cite[Theorem 1.4.11]{stichtenoth2009algebraic}. Thus, $v_{P_{\infty}}(x^iy^j)=-Mi-Nj$.
Since $0< r-1=\floor{\frac{N-1}{M}}-b'<\frac{N}{M}$, the functions $x^0y^t,x^iy^j(0\leq i\leq r-1,0\leq j\leq t-1)$ have pairwise distinct valuations at $P_{\infty}$. Therefore, these $tr+1$ functions are $\F_{q}$-linearly independent by the strict triangle inequality. Then we have $\dim_{\F_{q}}(V)=tr+1$. 

As a by-product of the above argument, we also obtain $V\subseteq \Lcal_{E}(tNP_{\infty})$. 
Hence, $\Ccal(\Pcal,V)$ has dimension $tr+1$ and minimum distance at least $(m-t)N$ by Section~\ref{sec:2.1}. By the Singleton-type bound~\eqref{eq:SLboundfor_r_delta}, we conclude that the minimum distance of $\Ccal(\Pcal,V)$ is exactly $(m-t)N$, and $\Ccal(\Pcal,V)$ is an optimal $(r=\floor{\frac{N-1}{M}}+1-b',\delta=N+1-r)$-LRC. 

{\rm (ii)} By interchanging $x$ and $y$, and swapping $M$ with $N$ in the above proof of {\rm (i)}, we can prove {\rm (ii)}.
\end{proof}

\subsection{Construction of Optimal $(r,\delta)$-LRCs via Superelliptic Curves from Norm-Trace Curves}
In this subsection, we present constructions of optimal $(r,\delta)$-LRCs based on Proposition~\ref{prop:construction_of_(r,delta)from_SEFF}. Before proceeding, we introduce the curves that will be used in our constructions.
\begin{lemma}\label{lem:good_superelliptic_curve_1}
 Let $q=\oq^s$ for a prime power $\oq$ and an integer $s\geq 2$. Let $b$ be a positive divisor of $\frac{\oq^{s}-1}{\oq-1}$, and let $c$ be a positive divisor of $s$. Denote $M=\frac{\oq^s-1}{b(\oq-1)}$ and $N=\oq^{s-c}$. Then the curve $\Cfrak$ defined over $\F_{q}$ by 
$$y^{M}=x^{N}+x^{\oq^{s-2c}}+\dots+x^{\oq^c}+x^{\oq^0}=\Tr_{\oq^s/\oq^c}(x)$$ is a superelliptic curve of genus $\frac{(M-1)(N-1)}{2}$. It has $\parentheses{\gcd(b,\frac{\oq^c-1}{\oq-1})\cdot\frac{q(q-1)}{b\oq^c}+\frac{q}{\oq^c}+1}$ rational points, including one rational point at infinity. 
\end{lemma}
\begin{proof} 
 The genus of the superelliptic curve $\Cfrak$ is $\frac{(M-1)(N-1)}{2}$ by Lemma \ref{lem:property_of_superelliptic}. We now count the affine rational points of $\Cfrak/\F_{q}=\Cfrak/\F_{\oq^s}$. That is, the number of pairs $(\alpha,\beta)\in \F_{\oq^s}^2$ satisfying $\beta^{\frac{\oq^s-1}{b(\oq-1)}}=\alpha^{\oq^{s-c}}+\alpha^{\oq^{s-2c}}+\dots+\alpha^{\oq^{c}}+\alpha^{\oq^0}=\Tr_{\oq^s/\oq^c}(\alpha)$.

By the property of the trace map $\Tr_{\oq^s/\oq^c}$ from $\F_{\oq^s}$ to $\F_{\oq^c}$, for any $\beta\in \F_{\oq^s}$ satisfying $\beta^{\frac{\oq^s-1}{b(\oq-1)}}\in \F_{\oq^c}$, there are exactly $\oq^{s-c}$ distinct $\alpha\in \F_{\oq^s}$, such that $\beta^{\frac{\oq^s-1}{b(\oq-1)}}=\Tr_{\oq^s/\oq^c}(\alpha)$. Also, for any $\beta\in \F_{\oq^s}$ satisfying $\beta^{\frac{\oq^s-1}{b(\oq-1)}}\notin \F_{\oq^c}$, there exists no $\alpha\in \F_{\oq^s}$ such that $\beta^{\frac{\oq^s-1}{b(\oq-1)}}=\Tr_{\oq^s/\oq^c}(\alpha)$. Therefore, to determine the number of affine rational points of $\Cfrak/\F_{\oq^s}$, we only need to count the elements $\beta\in \F_{\oq^s}$ satisfying $\beta^{\frac{\oq^s-1}{b(\oq-1)}}\in \F_{\oq^c}$, and then multiply the result by $\oq^{s-c}$. We consider the following two cases.
\begin{itemize}
    \item $\beta=0$. Clearly, $\beta^{\frac{\oq^s-1}{b(\oq-1)}}=0\in \F_{\oq^c}$.
    \item $\beta=u^i$ for a primitive element $u$ of $\F_{\oq^s}$ and an integer $0\leq i\leq \oq^s-2$. Then 
    $\beta^{\frac{\oq^s-1}{b(\oq-1)}}\in \F_{\oq^c}$ if and only if $(u^{i\frac{\oq^s-1}{b(\oq-1)}})^{\oq^c-1}=1$, which is equivalent to $\frac{b}{\gcd(b,\frac{\oq^c-1}{\oq-1})}\mid i$. 
\end{itemize}
Thus, the number of elements $\beta\in \F_{\oq^s}$ satisfying $\beta^{\frac{\oq^s-1}{b(\oq-1)}}\in \F_{\oq^c}$ is $1+(\oq^s-1)/\frac{b}{\gcd(b,\frac{\oq^c-1}{\oq-1})}=\gcd(b,\frac{\oq^c-1}{\oq-1})\cdot\frac{(\oq^s-1)}{b}+1$. Multiplying it by $\oq^{s-c}$, together with the unique rational point at infinity (see Lemma~\ref{lem:property_of_superelliptic}), we complete the proof. 
\end{proof}

Then we have the following explicit construction of optimal $(r,\delta)$-LRCs.
\begin{theorem}\label{thm:cons_via_SEFF_1}
   Let $q=\oq^s$ for a prime power $\oq$ and an integer $s\geq 2$. Let $b$ be a positive proper divisor of $\frac{\oq^s-1}{\oq-1}$, and let $c$ be a positive proper divisor of $s$. Denote $M=\frac{\oq^s-1}{b(\oq-1)}$ and $N=\oq^{s-c}$. Then we have the following two classes of explicit optimal $(r,\delta)$-LRCs, depending on whether $M<N$ or $M>N$.

\begin{itemize} 
  \item[{\rm (i)}] If $M<N$, then for any $0\leq b'\leq \floor{\frac{N-1}{M}}-1$, there exists an optimal $(r=\floor{\frac{N-1}{M}}+1-b',\delta=N+1-r)$-LRC with parameters
$
   [mN,tr+1,(m-t)N]_{q} 
$
for any $1\leq t<m\leq \ell=\frac{\gcd(b,\frac{\oq^c-1}{\oq-1})\cdot \frac{(q-1)q}{b\oq^c}+\frac{q}{\oq^c}}{N}$.

  \item[{\rm (ii)}] If $M>N$, then for any $0\leq b'\leq \floor{\frac{M-1}{N}}-1$, there exists an optimal $(r=\floor{\frac{M-1}{N}}+1-b',\delta=M+1-r)$-LRC with parameters 
$
 [mM,tr+1,(m-t)M]_{q}
$
for any $1\leq t<m\leq \ell=\frac{\gcd(b,\frac{\oq^c-1}{\oq-1})\cdot \frac{(q-1)q}{b\oq^c}}{M}$.
\end{itemize}
\end{theorem} 
 \begin{proof}
{\rm (i)} $M<N.$
We use the superelliptic curve $\Cfrak$ given in Lemma~\ref{lem:good_superelliptic_curve_1}, defined over $\F_{q}=\F_{\oq^s}$ by the equation $y^{M}=x^{N}+x^{\oq^{s-2c}}+\dots+x^{\oq^c}+x^{\oq^0}=\Tr_{\oq^s/\oq^c}(x)$. For each affine rational point $(\alpha,\beta)\in \Cfrak(\F_{q})$, we have a set $\{(\alpha+\alpha_i,\beta):1\leq i\leq N\}\subseteq \Cfrak(\F_{q})$, where $\alpha_1,\dots,\alpha_{N}\in \F_{q}$ are the $N=\oq^{s-c}$ distinct roots of the equation $\Tr_{\oq^s/\oq^c}(z)=0$. These sets are either disjoint or identical. There are totally $\ell=\frac{\gcd(b,\frac{\oq^c-1}{\oq-1})\cdot \frac{(q-1)q}{b\oq^c}+\frac{q}{\oq^c}}{N}$ such sets since $\Cfrak/\F_{q}$ has $\parentheses{\gcd(b,\frac{\oq^c-1}{\oq-1})\cdot \frac{(q-1)q}{b\oq^c}+\frac{q}{\oq^c}}$ affine rational points by Lemma~\ref{lem:good_superelliptic_curve_1}. Let $E/\F_{q}$ be the function field of $\Cfrak/\F_{q}$. We convert the affine rational points in each of these $\ell$ sets into the corresponding affine rational places, and denote the $\ell$ new sets of affine rational places by $\{P_{i,1},\dots,P_{i,r+\delta-1}\}$ ($1\leq i\leq \ell$). For each $1\leq i\leq \ell$, the values $x(P_{i,1}),\dots,x(P_{i,r+\delta-1})$ are pairwise distinct, while $y(P_{i,1}),\dots,y(P_{i,r+\delta-1})$ are all the same. By Proposition~\ref{prop:construction_of_(r,delta)from_SEFF} {\rm (i)}, the proof is complete.

{\rm (ii)} $M>N.$
We also use the superelliptic curve $\Cfrak$ in Lemma~\ref{lem:good_superelliptic_curve_1}, defined over $\F_{q}=\F_{\oq^s}$ by the equation $y^{M}=x^{N}+x^{\oq^{s-2c}}+\dots+x^{\oq^c}+x^{\oq^0}=\Tr_{\oq^s/\oq^c}(x)$. For each affine rational point $(\alpha,\beta)\in \Cfrak(\F_{q})$ with $\beta\neq 0$, we have a set $\{(\alpha,\beta_i\beta):1\leq i\leq M\}\subseteq \Cfrak(\F_{q})$, where $\beta_1,\dots,\beta_{M}\in \F_{q}$ are the $M$ distinct roots of the equation $z^{M}=1$ (note that $M=\frac{\oq^s-1}{b(\oq-1)}\mid (q-1)$). These sets are either disjoint or identical. 
There are totally $\ell=\frac{\gcd(b,\frac{\oq^c-1}{\oq-1})\cdot \frac{(q-1)q}{b\oq^c}}{M}$ such sets since there are totally $\parentheses{\gcd(b,\frac{\oq^c-1}{\oq-1})\cdot \frac{(q-1)q}{b\oq^c}}$ affine rational points $(\alpha,\beta)$ of $\Cfrak/\F_{q}$ with $\beta\neq 0$ by the proof of Lemma~\ref{lem:good_superelliptic_curve_1}. Let $E/\F_{q}$ be the function field of $\Cfrak/\F_{q}$. We convert the affine rational points in each of these $\ell$ sets into the corresponding affine rational places, and denote the $\ell$ new sets of affine rational places by $\{P_{i,1},\dots,P_{i,r+\delta-1}\}$ ($1\leq i\leq \ell$). For each $1\leq i\leq \ell$, the values $y(P_{i,1}),\dots,y(P_{i,r+\delta-1})$ are pairwise distinct, while $x(P_{i,1}),\dots,x(P_{i,r+\delta-1})$ are all the same. By Proposition~\ref{prop:construction_of_(r,delta)from_SEFF} {\rm (ii)}, the proof is complete.
\end{proof} 

In the following example, we present two representative constructions by Theorem~\ref{thm:cons_via_SEFF_1}. 
\begin{example}\label{exm:cons_via_SEFF_1} 
 {\rm (i)} Fixing $c=1$ in Theorem~\ref{thm:cons_via_SEFF_1}, for any prime power $q=\oq^s$ with $s\geq 2$, any divisor $b$ of $\frac{\oq^s-1}{\oq-1}$ satisfying $1< b <\frac{\oq^s-1}{\oq-1}$, and any $0\leq b'\leq \big\lfloor{b\frac{(\oq-1)(\oq^{s-1}-1)}{\oq^s-1}\big\rfloor}-1$, we have an explicit $q$-ary optimal $(r=\big\lfloor{b\frac{(\oq-1)(\oq^{s-1}-1)}{\oq^s-1}}\big\rfloor+1-b',\delta=\oq^{s-1}+1-r)$-LRC with length up to $\frac{1}{b\oq}q^2+\frac{b-1}{b\oq}q$.

 {\rm (ii)} Fixing $b=1$ in Theorem~\ref{thm:cons_via_SEFF_1}, for any prime power $q=\oq^s$ with $s\geq 2$, any positive proper divisor $c$ of $s$, and any $0\leq b'\leq \big\lfloor\frac{(\oq^s-\oq)}{\oq^{s-c}(\oq-1)}\big\rfloor-1$, we have an explicit $q$-ary optimal $(r=\big\lfloor\frac{(\oq^s-\oq)}{\oq^{s-c}(\oq-1)}\big\rfloor+1-b',\delta=\frac{\oq^s-1}{\oq-1}+1-r)$-LRC with length up to $\frac{q(q-1)}{\oq^c}$.
\end{example}
As illustrated in the above example,
Theorem~\ref{thm:cons_via_SEFF_1} can produce constructions of optimal $(r,\delta)$-LRCs with considerable code lengths. However, when we aim to fix $r,\delta$ and construct optimal $(r,\delta)$-LRCs as the field size $q$ tends to infinity, it is not ideal. We therefore consider a different class of constructions. In the following final part of the main result, employing a class of maximal superelliptic curves from Hermitian curves and their constant field extensions, we derive a new class of optimal $(r,\delta)$-LRCs, which generalizes and improves Theorem~\ref{thm:Cons_via_HEFF_(g+1-g',g+1+g')_p=2g+1_and_p_neq_2g+1} {\rm (i)} in Section~\ref{sec:4.3}. 


\subsection{Construction of Optimal $(r,\delta)$-LRCs via Maximal Superelliptic Curves from Hermitian Curves}
First, we introduce the maximal curves that will be used in our constructions, which are adapted from Hermitian curves.
\begin{lemma}\label{lem:good_superelliptic_curve_2}
Let $q=\oq^{2s}$ for a prime power $\oq$ and a positive integer $s$. Let $b$ be a positive divisor of $\oq+1$.
\begin{itemize}
    \item[{\rm (i)}] When $s$ is an odd positive integer, the curve $\Cfrak$ defined over $\F_{q}$ by the equation 
$
y^{\frac{\oq+1}{b}}=x^{\oq}+x
$
   is a maximal superelliptic curve of genus $\frac{(\frac{\oq+1}{b}-1)(\oq-1)}{2}$. It has $q+(\frac{\oq+1}{b}-1)(\oq-1)\sqrt{q}+1$ rational points, including one rational point at infinity. 

\item[{\rm (ii)}]
 When $\oq$ is odd, $b=\frac{\oq+1}{2}$, and $s$ is an even positive integer, the curve $\Cfrak'$ defined over $\F_{q}$ by the equation 
$
\gamma y^{\frac{\oq+1}{b}}=\gamma y^2=x^{\oq}+x
$ 
    is a maximal superelliptic curve of genus $\frac{\oq-1}{2}$,  
    where $\gamma$ is an arbitrary quadratic non-residue in $\F_{q}$.
\item[{\rm (iii)}]When $\oq=2$, $b=1$, and $s$ is an even positive integer, the curve $\Cfrak''$ defined over $\F_{q}$ by the equation 
$
y^{\frac{\oq+1}{b}}=y^3=x^2+x+\eta=x^{\oq}+x+\eta
$ is a maximal superelliptic curve of genus $1$, where $\eta$ is an arbitrary element in $\F_q\backslash \{\alpha^2+\alpha:\; \alpha\in \F_q\}$.
\end{itemize}
\end{lemma}
\begin{proof}  
{\rm (i)} By Lemma~\ref{lem:good_superelliptic_curve_1} (with the parameters in that lemma set to $s=2$ and $c=1$), $\Cfrak/\F_{\oq^2}$ is a superelliptic curve of genus $\frac{(\frac{\oq+1}{b}-1)(\oq-1)}{2}$, and has $\frac{\oq^3-\oq}{b}+\oq+1$ rational points. Thus, $\Cfrak/\F_{\oq^2}$ is a maximal curve since $\frac{\oq^3-\oq}{b}+\oq+1=\oq^2+2\frac{(\frac{\oq+1}{b}-1)(\oq-1)}{2}\oq+1$. By Lemma~\ref{lem:maximal_curve_lift}, $\Cfrak/\F_{q}$ is also a maximal curve, where $q=\oq^{2s}$ with $2\nmid s$. The item {\rm (i)} is proved.

{\rm (ii)} Note that the number of distinct roots in $\F_q$ of $\gamma y^{2}=\alpha^{\oq}+\alpha$ and the number of distinct roots in $\F_q$ of $y^{2}=\alpha^{\oq}+\alpha$ sum to $2$ for any $\alpha\in \F_q$. The total number of affine rational points of $\Cfrak'/\F_q$ defined by $\gamma y^{2}=x^{\oq}+x$ and $\Cfrak/\F_q$ defined by $y^{2}=x^{\oq}+x$ must be $2q$. 
Since $\Cfrak/\F_q$ is minimal by Lemma~\ref{lem:maximal_curve_lift} and the proof of {\rm (i)}, it follows that $\Cfrak'/\F_q$ is maximal.

For the proof of {\rm (iii)}, we refer to the proof of \cite[Lemma 15]{li2019optimal}, with the roles of 
$x
$ and 
$y$ swapped.
\end{proof}
Then we have the following explicit constructions.
\begin{theorem}\label{thm:cons_via_SEFF_2}
Let $q=\oq^{2s}$, where $\oq$ is a prime power and $s$ is an odd positive integer. Let $b$ be a positive proper divisor of $\oq+1$. We have the following two classes of $q$-ary optimal $(r,\delta)$-LRCs, depending on the value of $b$.
\begin{itemize} 
\item[{\rm (i)}] If $b>1$, then for any $0\leq b'\leq b-2$, there exists an optimal $(r=b-b',\delta=\oq+1-r)$-LRC with parameters 
$
[m\oq,tr+1,(m-t)\oq]_{q}
$
for any $1\leq t<m\leq \ell=\frac{q+(\frac{\oq+1}{b}-1)(\oq-1)\sqrt{q}}{\oq}$. In particular, when $\oq$ is odd and $b=\frac{\oq+1}{2}$, the integer $s$ no longer needs to be odd; it can be any positive integer.

\item[{\rm (ii)}] If $b=1$, then there exists an optimal $(r=2,\delta=\oq)$-LRC with parameters 
$
[m(\oq+1),tr+1,(m-t)(\oq+1)]_{q}
$
for any $1\leq t<m\leq \ell=\lfloor\frac{q+\oq(\oq-1)\sqrt{q}}{\oq+1}\rfloor$. In particular, when $\oq=2$, the integer $s$ no longer needs to be odd; it can be any positive integer.
\end{itemize}
\end{theorem}
\begin{proof}
To apply Proposition~\ref{prop:construction_of_(r,delta)from_SEFF}, we consistently take $M=\frac{\oq+1}{b}$ and $N=\oq$ below, although they do not explicitly appear in the proof. Note that $\floor{\frac{N-1}{M}}=b-1$ when $b>1$, and $\floor{\frac{M-1}{N}}=1$ when $b=1$.

{\rm(i)} $b>1$. Let $q=\oq^{2s}$ for a prime power $\oq$ and an odd positive integer $s$.  By Lemma~\ref{lem:good_superelliptic_curve_2} {\rm (i)}, the curve $\Cfrak/\F_{q}$ defined by $y^{\frac{\oq+1}{b}}=x^{\oq}+x$ is a maximal superelliptic curve of genus $\frac{(\frac{\oq+1}{b}-1)(\oq-1)}{2}$. For each affine rational point $(\alpha,\beta)\in \Cfrak(\F_{q})$, we have a set $\{(\alpha+\alpha_i,\beta):1\leq i\leq \oq\}\subseteq \Cfrak(\F_{q})$, where $\alpha_1,\dots,\alpha_{\oq}\in \F_{q}$ are the $\oq$ distinct roots of the equation $z^\oq+z=0$. These sets are either disjoint or identical. There are totally $\ell=\frac{q+(\frac{\oq+1}{b}-1)(\oq-1)\sqrt{q}}{\oq}$ such sets since $\Cfrak/\F_{q}$ has $q+(\frac{\oq+1}{b}-1)(\oq-1)\sqrt{q}$ affine rational points. Let $E/\F_{q}$ be the function field of $\Cfrak/\F_{q}$. We convert the affine rational points in each of these $\ell$ sets into the corresponding affine rational places, and denote the $\ell$ new sets of affine rational places by $\{P_{i,1},\dots,P_{i,r+\delta-1}\}$ ($1\leq i\leq \ell$). For each $1\leq i\leq \ell$, the values $x(P_{i,1}),\dots,x(P_{i,r+\delta-1})$ are pairwise distinct, while $y(P_{i,1}),\dots,y(P_{i,r+\delta-1})$ are all the same.  
By Proposition~\ref{prop:construction_of_(r,delta)from_SEFF} {\rm (i)}, the proof is complete.

As for the special case where $q=\oq^{2s}$ with $2\nmid \oq$, $b=\frac{\oq+1}{2}$, and $2\mid s$, 
we use the maximal curve $\Cfrak'/\F_{q}$ in Lemma~\ref{lem:good_superelliptic_curve_2} {\rm (ii)} defined by $\gamma y^{\frac{\oq+1}{b}}=\gamma y^2=x^{\oq}+x$, where $\gamma$ is an arbitrary quadratic non-residue in $\F_q$. The rest of the proof is the same as above. 

{\rm (ii)} $b=1$. Let $q=\oq^{2s}$ for a prime power $\oq$ and an odd positive integer $s$. By Lemma~\ref{lem:good_superelliptic_curve_2} {\rm (i)}, the curve $\Cfrak/\F_{q}$ defined by $y^{\frac{\oq+1}{b}}=y^{\oq+1}=x^{\oq}+x$ is a maximal superelliptic curve of genus $\frac{\oq(\oq-1)}{2}$. For each affine rational point $(\alpha,\beta)\in \Cfrak(\F_{q})$ with $\beta\neq 0$, we have a set $\{(\alpha,\beta_i\beta):1\leq i\leq \oq+1\}\subseteq \Cfrak(\F_{q})$, where $\beta_1,\dots,\beta_{\oq+1}\in \F_{q}$ are the $\oq+1$ distinct roots of the equation $z^{\oq+1}=1$ (note that $(\oq+1)\mid (q-1)$). These sets are either disjoint or identical. 
There are totally $\ell=\frac{q+\oq(\oq-1)\sqrt{q}-\oq}{\oq+1}$ such sets since there are totally $(q+\oq(\oq-1)\sqrt{q}-\oq)$ affine rational points $(\alpha,\beta)$ of $\Cfrak/\F_{q}$ with $\beta\neq 0$. Let $E/\F_{q}$ be the function field of $\Cfrak/\F_{q}$. We convert the affine rational points in each of these $\ell$ sets into the corresponding affine rational places, and denote the $\ell$ new sets of affine rational places by $\{P_{i,1},\dots,P_{i,r+\delta-1}\}$ $(1\leq i\leq \ell)$. For each $1\leq i\leq \ell$, the values $y(P_{i,1}),\dots,y(P_{i,r+\delta-1})$ are pairwise distinct, while $x(P_{i,1}),\dots,x(P_{i,r+\delta-1})$ are all the same. Applying Proposition~\ref{prop:construction_of_(r,delta)from_SEFF} {\rm (ii)}, the proof is complete.

As for the special case where $q=\oq^{2s}$ with $\oq=2$, $b=1$, and $2\mid s$, we use the maximal curve $\Cfrak''/\F_{q}$ in Lemma~\ref{lem:good_superelliptic_curve_2} {\rm (iii)} defined by $y^{\frac{\oq+1}{b}}=y^3=x^{2}+x+\eta=x^{\oq}+x+\eta$, where $\eta$ is an arbitrary element in $\F_q\backslash \{\alpha^2+\alpha:\alpha\in \F_q\}$. The rest of the proof is similar to the above. The only difference is that the number of the local repair groups becomes $\ell=\frac{q+\oq(\oq-1)\sqrt{q}}{\oq+1}$ since there are totally $q+\oq(\oq-1)\sqrt{q}$ affine rational points $(\alpha,\beta)$ of $\Cfrak''/\F_{q}$ with $\beta\neq 0$. In both cases, we have $\ell=\lfloor\frac{q+\oq(\oq-1)\sqrt{q}}{\oq+1}\rfloor$, the proof of {\rm (ii)} is complete. 
\end{proof}
\begin{remark}\label{rem:cons_via_SEFF_2_1}
    By setting $\oq=3,b=\frac{\oq+1}{2}=2$ and $b'=0$ in Theorem~\ref{thm:cons_via_SEFF_2} {\rm (i)} and setting $\oq=2$ and $b=1$ in Theorem~\ref{thm:cons_via_SEFF_2} {\rm (ii)},  we recover part of \cite[Theorem 1]{li2019optimal}, which presents a wide class of optimal $(2,2)$-LRCs with lengths approaching $q+2\sqrt{q}$.
\end{remark}
\begin{remark}\label{rem:cons_via_SEFF_2_2}
In Theorem \ref{thm:Cons_via_HEFF_(g+1-g',g+1+g')_p=2g+1_and_p_neq_2g+1} {\rm (i)}, for any odd prime power $2g+1\geq 5$, integer $0\leq g'\leq g-1$, we have a $q$-ary  optimal $(r=g+1-g',\delta=g+1+g')$-LRC with length up to $q+2g\sqrt{q}$, where $q=(2g+1)^{2s}$ for any positive integer $s$. By setting $\oq=2g+1$, $b=g+1$, $b'=g'$ (note that here $b=\frac{\oq+1}{2}$) in Theorem \ref{thm:cons_via_SEFF_2} {\rm (i)}, we directly recover Theorem~\ref{thm:Cons_via_HEFF_(g+1-g',g+1+g')_p=2g+1_and_p_neq_2g+1} {\rm (i)}. Moreover, Theorem~\ref{thm:cons_via_SEFF_2} {\rm (i)} has the following two advantages.
\begin{itemize}
\item  When the repair group size \( r + \delta - 1 \) is fixed and \( r \) becomes smaller, Theorem~\ref{thm:cons_via_SEFF_2} {\rm (i)} may yield longer optimal \( (r,\delta) \)-LRCs than Theorem~\ref{thm:Cons_via_HEFF_(g+1-g',g+1+g')_p=2g+1_and_p_neq_2g+1}. We illustrate this with examples under two distinct parameter settings.

(1) Let \( g = 11 \), \( g' = 0 \) in Theorem~\ref{thm:Cons_via_HEFF_(g+1-g',g+1+g')_p=2g+1_and_p_neq_2g+1}~{\rm (i)}. One obtains optimal \( q \)-ary \( (12,12) \)-LRCs of length \( q + 22\sqrt{q} \). In this case, no improvement can be made by Theorem~\ref{thm:cons_via_SEFF_2} {\rm (i)} towards deriving longer optimal $(12,12)$-LRCs.

(2) Let \( g = 11 \), \( g' = 8 \) in Theorem~\ref{thm:Cons_via_HEFF_(g+1-g',g+1+g')_p=2g+1_and_p_neq_2g+1}~{\rm (i)}. One obtains optimal \( q \)-ary \( (4,20) \)-LRCs of length \( q + 22\sqrt{q} \) for any \( q = 23^{2s} \). In this case, Theorem~\ref{thm:cons_via_SEFF_2} {\rm (i)} enables an improved construction: by choosing \( \oq = 23 \), \( b = 4 \), and \( b' = 0 \), we obtain longer optimal \( q \)-ary \( (4,20) \)-LRCs of length up to \( q + 110\sqrt{q} \) for any \( q = 23^{2s} \) with odd $s$.
The idea is to minimize \( b \) to maximize the code length.

\item Theorem~\ref{thm:cons_via_SEFF_2} {\rm (i)} allows constructions over a broader range of finite fields compared to Theorem~\ref{thm:Cons_via_HEFF_(g+1-g',g+1+g')_p=2g+1_and_p_neq_2g+1}~\textup{(i)}, including those of even characteristic.
\end{itemize} 
\end{remark}
 
\section{Concluding Remarks}
	\label{sec:6} 
In this paper, we studied the construction of optimal $(r,\delta)$-LRCs with flexible minimum distances, particularly for the case $\delta \geq 3$. By leveraging the automorphism groups of elliptic and genus-$2$ hyperelliptic function fields, together with their group of divisor classes of degree zero, we constructed several families of explicit optimal $(r,3)$-LRCs and $(2,\delta)$-LRCs with lengths approaching $q + 2\sqrt{q}$ or $q + 4\sqrt{q}$. We also employed some hyperelliptic and superelliptic curves of higher genus to construct explicit optimal $(r,\delta)$-LRCs with even longer lengths and flexible parameters. Most of these optimal $(r,\delta)$-LRCs have lengths exceeding $q+1$, and many of them attain the currently best-known code lengths.
  
To the best of our knowledge, all known constructions of optimal $(r,\delta)$-LRCs with flexible minimum distances obtained via evaluation-based methods primarily focus on the case of $r$-LRCs (i.e., $(r,\delta=2)$-LRCs), and then some of them are naturally extended to the case of $(r,\delta\geq 3)$-LRCs. This paper demonstrates that such extensions are not always straightforward (see Section~\ref{sec:3.1}, especially Remark~\ref{rem:equivalentconditionforEFF_1} {\rm (i)} for details), and that the construction of optimal $(r,\delta)$-LRCs for $\delta \geq 3$ deserves independent attention, rather than being merely treated as a by-product of the $r$-LRC case.
Moreover, our constructions demonstrate that algebraic geometry codes are also highly effective for constructing optimal \((r,\delta)\)-LRCs with flexible minimum distances, even for $\delta\geq 3$. 
Three avenues for future research may be worth exploring. 
\begin{itemize}
\item  
Constructing more optimal $(r,\delta)$-LRCs based on the general framework in Section~\ref{sec:3.1}, for example, 
optimal $(r,\delta)$-LRCs with $\delta=4,5,6,\dots$. 
\item
Exploring whether the group of divisor classes of degree zero of higher-genus hyperelliptic or superelliptic curves can be utilized to obtain a general framework like those in Section~\ref{sec:3.1} and~\ref{sec:4.1}, thus obtaining optimal $(r,\delta)$-LRCs with a wider range of parameters and longer code lengths. 
\item
Conducting a more refined study of the general framework developed in Section~\ref{sec:5.1}, to further generalize this framework, or to identify additional algebraic curves that fit this framework and can be used to construct long optimal $(r,\delta)$-LRCs.
\end{itemize}
\Review{In pursuing these directions, it may also be useful to take into account the related constructions of LRCs via algebraic geometry codes in \cite{Ballentine2019,Bartoli2020locally,Munuera2020locally}. The methods and curves used in these works may provide inspiration for further research.}

\appendices
\section{A Length Upper Bound for Optimal $(r,\delta)$-LRCs with Flexible Minimum Distances \Review{and Constant $r,\delta$}
}\label{appsec:1}
In Section~\ref{sec:1.1}, we mentioned that Guruswami \textit{et al.} \cite[Theorem 13 and Corollary 14]{Guruswami2019how} established an upper bound on the code length $n$ of $q$-ary optimal $r$-LRCs with minimum distance $d=\Theta(n)$ and constant $r$, yielding $n\leq O(q)$. We also said that this bound can be generalized to the case of $(r,\delta)$-LRCs. As \Review{some of} our constructions of optimal $(r,\delta)$-LRCs with flexible minimum distances \Review{and constant $r,\delta$} happen to fall within the scope of this generalized bound, we formally state and prove it.
\begin{theorem}\label{thm:bound_for_length_of_optimal_(r,delta)}
    The minimum distance $d$ of an optimal $(r,\delta)$-LRC $\Ccal$ with parameters $[n,k,d]_q$ satisfying $\frac{d}{n}\leq \frac{2}{3}$ is upper bounded by 
    \begin{align}\label{eq:appendix_bound}
        d\leq \frac{(r+\delta-1)(r+1)+\delta(\delta-1)}{r}q.
    \end{align}
     
    Consequently, when $r,\delta$ are fixed constants, the code length $n$ of optimal $(r,\delta)$-LRCs with $d=\Theta(n)$ and $\frac{d}{n}\leq \frac{2}{3}$ is upper bounded by $O(q)$.
\end{theorem}
\begin{proof}  
   By \cite[Corollary 2]{grezet2019alphabet}, we have 
   $
   k\leq \min_{t\in \Zbb_{\geq 0}}\{tr+k_{\rm opt}^{(q)}(n-t(r+\delta-1),d)\},
   $ 
   where $k_{\rm opt}^{(q)}(n-t(r+\delta-1),d)$ denotes the maximal possible dimension of a linear code with length $n-t(r+\delta-1)$ and minimum distance $d$. Letting $t=\big\lceil \frac{n-(1-\varepsilon)\frac{qd}{q-1}}{r+\delta-1}\big\rceil$, where $\varepsilon=\frac{1}{q^2}$, we have $t\geq 0$ since $\frac{d}{n}\leq \frac{2}{3}\leq \frac{q}{q+1}$ for any prime power $q$. By the plotkin bound, we have  
   $k_{\rm opt}^{(q)}(n-t(r+\delta-1),d)\leq k_{\rm opt}^{(q)}((1-\varepsilon)\frac{qd}{q-1},d)\leq \log_q(1/\varepsilon)=2$. Hence, we have
\begin{align}\label{eq:202506061714} 
    \ceil{\frac{n-(1-\frac{1}{q^2})\frac{qd}{q-1}}{r+\delta-1}}r+2=tr+2\geq k.
\end{align}
   
Moreover, we have a lower bound on $k$. Let $n_0=(r+\delta-1)\ceil{\frac{n}{r+\delta-1}}-n$,  
$n'=n+n_0$ and $d'=d+n_0$. Since $\Ccal$ is optimal, it holds $d=n-k+1-(\ceil{\frac{k}{r}}-1)(\delta-1)$. Then we have $(r+\delta-1)\mid n'$ and $d'=n'-k+1-(\ceil{\frac{k}{r}}-1)(\delta-1)$. By \cite[Lemma 1]{chen2022some}, we have 
\begin{align}\label{eq:202506061748}
    k=n'-d'+1-\frac{n'(\delta-1)}{r+\delta-1}+\parentheses{\floor{\frac{d'-\delta}{r+\delta-1}}+1}(\delta-1)&\geq n'-d'+1-\frac{\delta-1}{r+\delta-1}(n'-(d'-\delta))\\ \label{eq:202506061815}
    &=n-d+1-\frac{\delta-1}{r+\delta-1}(n-d+\delta). 
\end{align}
Combining \eqref{eq:202506061714},  \eqref{eq:202506061748} and \eqref{eq:202506061815}, we have 
$
\parentheses{\frac{n-(1-\frac{1}{q^2})\frac{qd}{q-1}}{r+\delta-1}+1}r+2\geq tr+2\geq k\geq n-d+1-\frac{\delta-1}{r+\delta-1}(n-d+\delta).
$
Solving this inequality with respect to $d$, we have $d\leq \frac{(r+\delta-1)(r+1)+\delta(\delta-1)}{r}q.$ The proof is complete.
\end{proof}

 
\bibliographystyle{IEEEtran}
\bibliography{References.bib}

\end{document}